\documentclass[final]{siamltex}
\usepackage{epsfig}
\usepackage{graphicx}
\usepackage{amsmath}
\usepackage{amssymb}
\usepackage{setspace}
\usepackage{enumerate}

\usepackage{graphicx}
\usepackage{array}
\newcolumntype{C}[1]{>{\centering\let\newline\\\arraybackslash\hspace{0pt}}m{#1}}
\usepackage[ruled,linesnumbered,vlined]{algorithm2e}
\usepackage{subfigure}
\usepackage{bm}

\newtheorem{defn}{\noindent $\mathbf{Definition}$}[section]

\newtheorem{prop}[defn]{$\mathbf{Proposition}$}
\newtheorem{thm}[defn]{$\mathbf{Theorem}$}

\title{TEMPO: Feature-Endowed Teichm\"{u}ller Extremal Mappings of Point Clouds}
\author{Ting Wei Meng, Gary Pui-Tung Choi and Lok Ming Lui}

\begin{document}

\maketitle

\begin{abstract}
In recent decades, the use of 3D point clouds has been widespread in computer industry. The development of techniques in analyzing point clouds is increasingly important. In particular, mapping of point clouds has been a challenging problem. In this paper, we develop a discrete analogue of the Teichm\"{u}ller extremal mappings, which guarantees uniform conformality distortions, on point cloud surfaces. Based on the discrete analogue, we propose a novel method called TEMPO for computing Teichm\"{u}ller extremal mappings between feature-endowed point clouds. Using our proposed method, the Teichm\"{u}ller metric is introduced for evaluating the dissimilarity of point clouds. Consequently, our algorithm enables accurate recognition and classification of point clouds. Experimental results demonstrate the effectiveness of our proposed method.
\end{abstract}


\pagestyle{myheadings}
\thispagestyle{plain}
\markboth{Meng, Choi and Lui}{TEMPO: Teichm\"{u}ller Extremal Mapping of Point Clouds}


\section{Introduction} \label{intro}
In recent years, the use of point clouds has been widespread for a large variety of industrial applications. In this emerging field, one important topic is point cloud mapping, which has a wide range of applications such as 3D facial recognition and object classification. The difficulty of the mapping problem arises from the limited information a point set provides. Unlike mesh representations of surfaces, point clouds do not contain any information about the connectivity structure. This hinders the extension of the well-established algorithms on surface meshes to point clouds.

In particular, accurate registration between two feature-endowed point clouds are challenging because of the enforcement of feature correspondences. Besides satisfying the feature correspondences, an admissible mapping between two feature-endowed point clouds should also retain the underlying geometric structures as complete as possible. In differential geometry, conformal mappings are angle preserving and hence preserve the shapes of infinitesimally small figures. Note that under general feature correspondences, conformal mappings between two surfaces may not exist. Nevertheless, we can consider quasi-conformal mappings, the extension of conformal mappings. Among all quasi-conformal mappings, a special class of landmark matching surface diffeomorphisms, called the {\em Teichm\"{u}ller mappings (T-maps)}, is particularly desirable as they ensure uniform conformality distortions. In other words, Teichm\"{u}ller mappings effectively achieve evenly distributed geometric distortions under prescribed feature point correspondences.

In \cite{Lui14}, Lui et al. proposed an efficient method for computing the Teichm\"{u}ller mapping on triangular meshes using the Beltrami coefficients. In this paper, we explore Teichm\"{u}ller mappings on point clouds in both the theoretical and the computational aspects. A novel discrete analogue of Teichm\"{u}ller mappings on point clouds is rigorously developed. Then, we extend and modify the algorithm in \cite{Lui14} for point clouds with underlying surfaces being simply-connected open surfaces. More specifically, we first introduce a hybrid quasi-conformal reconstruction scheme on point clouds with given Beltrami coefficients. Then, we propose an improved method for approximating differential operators arising from the computation of conformal and quasi-conformal mappings on point clouds with disk topology. With the abovementioned schemes, we propose an efficient method for computing the {\bf T}eichm\"uller {\bf E}xtremal {\bf M}apping of {\bf PO}int clouds, abbreviated as TEMPO. To register two disk-type point clouds with given feature correspondences, our TEMPO algorithm first parameterizes them onto two rectangular regions on the complex plane using conformal mappings. Then, the Teichm\"{u}ller mapping between the two rectangular regions with the prescribed feature correspondences is computed. The Teichm\"{u}ller parameterization naturally induces a metric, called the {\em Teichm\"{u}ller metric}, to assess the dissimilarity of two feature-endowed point clouds. Hence, the Teichm\"{u}ller metric provides us with an accurate way to classify point clouds. Experimental results are presented to demonstrate the effectiveness of our proposed scheme for feature-endowed disk-type point clouds.

The organization of this paper is outlined as follows. In Section \ref{previous}, we review the previous works on parameterizations of simply-connected surfaces and point clouds. The contributions of our work are highlighted in Section \ref{contribution}. In Section \ref{background}, we introduce some mathematical concepts related to our work. In Section \ref{theory}, we develop a novel discrete analogue of the continuous Teichm\"uller mappings on point clouds. In Section \ref{main}, we explain our proposed TEMPO algorithm for computing Teichm\"uller mappings between feature-endowed point clouds in details. Experimental results are reported in Section \ref{experiment} to illustrate the effectiveness of our proposed method. The paper is concluded in Section \ref{conclusion}.


\section{Previous works} \label{previous}
In this section, we review the previous works on conformal and quasi-conformal parameterizations of disk-type meshes and point clouds.

In the past few decades, extensive studies on mesh parameterization have been carried out by various research groups. Surveys of mesh parameterization methods can be found in \cite{Floater02,Floater05,Hormann07,Sheffer06}. In \cite{Floater97}, Floater introduced the shape-preserving mesh parameterization method, which involves solving linear systems based on convex combinations. Hormann and Greiner \cite{Hormann00} presented the Most Isometric Parameterization of Surfaces method (MIPS) for disk-like surface meshes. Sheffer and de Sturler \cite{Sheffer00} proposed the Angle Based Flattening (ABF) method for mesh parameterizations. The method solves the parameterization problem as a constrained optimization problem in terms of angles only. Later, Sheffer et al. \cite{Sheffer05} extended the ABF method to the ABF++ method. A new numerical solution technique and an efficient hierarchical technique are introduced to overcome the drawbacks of the ABF method. In \cite{Levy02}, L\'{e}vy et al. proposed the Least Square Conformal Maps (LSCM). The parameterization method is based on a least-square approximation of the Cauchy-Riemann equations. Desbrun et al. \cite{Desbrun02} proposed the Discrete, Natural Conformal Parameterization (DNCP) by computing the discrete Dirichlet energy without fixing the boundary positions. The LSCM and the DNCP were later shown to be equivalent. Motivated by the Mean Value Theorem for harmonic functions, Floater \cite{Floater03} derived a generalization of barycentric coordinates for planar parameterizations. Kharevych et al. \cite{Kharevych05} introduced a method for constructing conformal parameterizations of meshes based on circle patterns, which are arrangements of circles with prescribed intersection angles. In \cite{Gu02}, Gu and Yau computed conformal parameterizations by constructing a basis of holomorphic 1-forms and integrating holomorphic differentials. In \cite{Luo04}, Luo proposed the combinatorial Yamabe flow for the parameterization of triangulated surfaces. Jin et al. \cite{Jin05} computed disk conformal parameterizations using the double covering technique \cite{Gu03} and the spherical conformal mapping algorithm \cite{Gu04}. In \cite{Mullen08}, Mullen et al. presented a spectral approach to compute free-boundary conformal parameterizations of triangular meshes. The approach involves solving a sparse symmetric generalized eigenvalue problem. In \cite{Jin08}, Jin et al. proposed the discrete surface Ricci flow algorithm for conformal parameterizations of triangular meshes. The discrete curvature determines the deformation of edge lengths of the triangles and hence the discrete Ricci flow. Later, Yang et al. \cite{Yang09} generalized the traditional discrete Ricci flow by allowing the two circles either intersecting or separating from each other in circle packing. In \cite{Zhang14}, Zhang et al. introduced the unified theoretic framework for discrete surface Ricci flow. The unified framework improved the robustness and efficiency of conformal parameterizations. In \cite{Choi15b}, Choi and Lui proposed a two-step iterative scheme for computing disk conformal parameterizations of simply-connected open meshes. The conformality distortion of the parameterization is corrected in two steps using quasi-conformal theories. Recently, Choi and Lui \cite{Choi15c} introduced a linear algorithm for disk conformal parameterizations, with the aid of a linear spherical conformal parameterization algorithm \cite{Choi15a}.

In recent years, several methods for quasi-conformal parameterizations of surface meshes have been developed. Weber et al. \cite{Weber12} presented an algorithm for computing extremal quasi-conformal mappings for simply-connected open meshes using holomorphic quadratic differentials. In \cite{Lipman12}, Lipman introduced convex mapping spaces for triangular meshes with guarantees on the maximal conformal distortion and injectivity of their mappings. The methods were applied to compute quasi-conformal parameterizations. In \cite{Lui14}, Lui et al. proposed an iterative algorithm for computing Teichm\"uller mappings of simply-connected open meshes with the aid of Beltrami coefficients. The convergence of the algorithm was recently proved in \cite{Lui15}.

With the increasing popularity of point clouds, numerous methods on parameterizing and registering point clouds have been developed. When compared with the case of meshes, parameterizing point clouds with disk topology is a more challenging task because of the absence of the connectivity information. Floater and Reimers \cite{Floater00,Floater01} presented a meshless parameterization method by solving a sparse linear system. In \cite{Guo06}, Guo et al. used Riemann surface theory and Hodge theory to obtain point cloud conformal parameterizations. Zhang et al. \cite{Zhang12} proposed an as-rigid-as-possible (ARAP) parameterization method for point clouds with disk topology. Liang et al. \cite{Liang12} computed the spherical conformal parameterizations of genus-0 point clouds by approximating the Laplace-Beltrami operator. Choi et al. \cite{Choi15d} proposed the North-South reiterations for computing the spherical conformal parameterizations of genus-0 point clouds. The method can also be extended to obtain planar conformal parameterizations of disk-type point clouds. In \cite{Meng15}, Meng and Lui computed quasi-conformal parameterizations of disk-type point clouds by developing the approximation theories of quasi-conformal maps and Beltrami coefficients on point clouds. Table \ref{table:previous_work} lists different parameterization schemes for point clouds with disk topology.

\begin{table}[t]
    \centering
    \begin{tabular}{ |C{50mm}|C{40mm}|C{20mm}| }
    \hline
    Methods & Criteria for the distortion & Landmark constraints allowed? \\ \hline
    Meshless parameterization \cite{Floater00,Floater01} & / & No \\ \hline
    Point-based global parameterization \cite{Guo06} & Conformal except a number of zero points & No\\ \hline
    As-rigid-as-possible meshless parameterization \cite{Zhang12} & As-rigid-as-possible & No\\ \hline
    North-South reiteration \cite{Choi15d} & Conformal & No\\ \hline
    Quasi-conformal parameterization \cite{Meng15} & Quasi-conformal & No\\ \hline
    Our proposed TEMPO method & Teichm\"uller (uniform conformality distortion) & Yes\\ \hline
    \end{tabular}
    \caption{A comparison on the parameterization schemes for point clouds with disk topology.}
    \label{table:previous_work}
\end{table}


\section{Contribution} \label{contribution}
In this work, we are concerned with both the theory and the computation of Teichm\"uller extremal mappings on point clouds. To the best of our knowledge, this is the first work on the Teichm\"uller mappings on point clouds.

For the theoretical aspect, we develop the concept of PCT-maps, a novel discrete analogue of the continuous Teichm\"uller mappings on point clouds. Rigorous theories are established to show the relationship between the discrete PCT-maps and the continuous Teichm\"uller mappings.

For the computational aspect, our proposed TEMPO method for computing the Teichm\"uller mappings between feature-endowed point clouds with disk topology is advantageous for the following reasons:
\begin{enumerate}
 \item We introduce a new hybrid quasi-conformal reconstruction scheme for point clouds with given Beltrami coefficients. The scheme enables us to extend the existing mesh-based methods for computing quasi-conformal mappings to point cloud surfaces.
 \item We introduce an improved method for approximating differential operators arising from the computation of conformal and quasi-conformal mappings on point clouds with disk topology by combining the Moving Least Square (MLS) method in \cite{Liang12,Liang13,Choi15d} and the local mesh method in \cite{Lai13}.
 \item Unlike the existing approaches for point cloud parameterizations, our proposed Teichm\"uller parameterization method allows landmark constraints.
 \item As Teichm\"uller mappings ensure uniform conformality distortions, our method provides a natural registration result without uneven distortions.
 \item Our Teichm\"uller parameterization method is highly efficient. The computation typically completes within 1 minute for moderate point clouds, and requires only a few minutes even for very dense point clouds.
 \item The induced Teichm\"uller metric serves as an effective dissimilarity metric for accurate classification of point clouds.
 \item Besides the computational advantages, our TEMPO method is theoretically supported. Under proper assumptions on the input point cloud, the limit mapping obtained by our algorithm is proved to be a PCT-map.
\end{enumerate}


\section{Mathematical background} \label{background}

In this section, we introduce some mathematical concepts closely related to our work. Readers are referred to \cite{Schoen94,Gardiner00,Lui14} for more details.


\subsection{Conformal mappings}
First, we introduce the concept of conformal mappings.
\begin{defn}[Conformal mappings]
Let $M$ and $N$ be two Riemann surfaces. A mapping $f: M\to N$ is said to be a \emph{conformal mapping} if there exists a positive scalar function $\lambda$ such that
\begin{equation}
f^*ds_{\mathcal{N}}^2 = \lambda ds_{\mathcal{M}}^2.
\end{equation}
\end{defn}

The mentioned scalar function $\lambda$ is called the \emph{conformal factor}. It follows that every conformal mapping preserves angles and hence the shapes at an infinitesimal scale. By the uniformization theorem, every simply-connected Riemann surface is conformally equivalent to one of the following domains:

\begin{enumerate}
\item the open unit disk,
\item the complex plane, and
\item the Riemann sphere.
\end{enumerate}

In this work, we focus on point clouds sampled from simply-connected open surfaces. With the theoretical guarantee by the uniformization theorem, it is natural to consider conformally parameterizing the point clouds onto simple planar domains, such as the open unit disk or a rectangle, for further analyses. However, with the presence of feature point correspondences, conformal mappings between two surfaces may not exist. In this case, we need to introduce a generalization of conformal mappings.


\subsection{Quasi-conformal mappings}
\begin{figure}[t]
\centering
\includegraphics[width=0.8\textwidth]{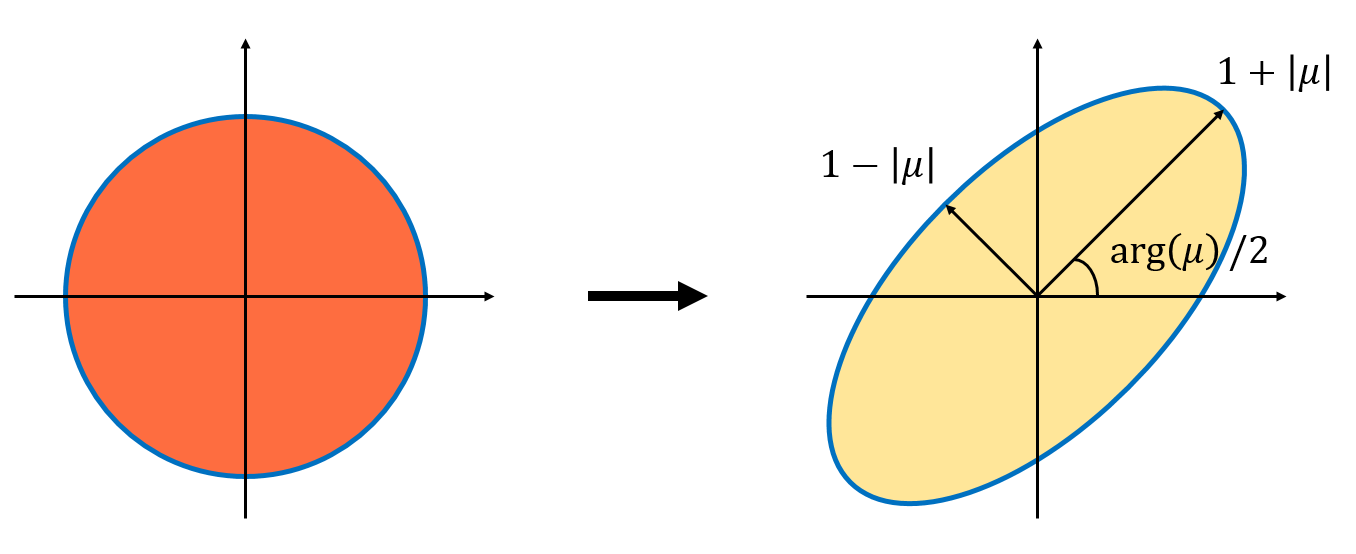}
\caption{An geometric illustration of quasi-conformal mappings. The information about the maximal magnification and shrinkage can be determined by the associated Beltrami coefficient $\mu$.}
\label{fig:BC}
\end{figure}
A generalization of conformal mappings is the quasi-conformal mappings, which are orientation preserving homeomorphisms between Riemann surfaces with bounded conformality distortion. Intuitively, conformal mappings map infinitesimal circles to infinitesimal circles, while quasi-conformal mappings map infinitesimal circles to infinitesimal ellipses with bounded eccentricity \cite{Gardiner00}. Mathematically, quasi-conformal mappings can be defined as follows.
\begin{defn}[Quasi-conformal mappings]
$f: \mathbb{C} \to \mathbb{C}$ is a \emph{quasi-conformal mapping} if it satisfies the Beltrami equation
\begin{equation}\label{eqt:beltrami}
\frac{\partial f}{\partial \overline{z}} = \mu(z) \frac{\partial f}{\partial z}
\end{equation}
for some complex-valued function $\mu$ satisfying $||\mu||_{\infty}< 1$ and that $\frac{\partial f}{\partial z}$ is non-vanishing almost everywhere. \\
\end{defn}

The complex-valued function $\mu$ is called the \emph{Beltrami coefficient} of the quasi-conformal mapping $f$. The Beltrami coefficient serves as a measure of the non-conformality of $f$. It measures how far the mapping $f$ at each point deviates from a conformal mapping. A geometric illustration of quasi-conformal mappings is provided in Figure \ref{fig:BC}. In addition, the \emph{maximal quasi-conformal dilation} of $f$ is given by
\begin{equation}
K(f) = \frac{1+||\mu||_{\infty}}{1-||\mu||_{\infty}}.
\end{equation}

For the composition of two quasi-conformal mappings, the associated Beltrami coefficient can be expressed in terms of the Beltrami coefficients of the two mappings.
\begin{theorem}\label{thm:composition}
Let $f: \Omega \subset \mathbb{C} \to f(\Omega)$ and $g: f(\Omega) \to \mathbb{C}$ be two quasi-conformal mappings. Then the Beltrami coefficient of $g \circ f$ is given by
\begin{equation}
\mu_{g \circ f} = \frac{\mu_f+(\overline{f_z}/f_z) (\mu_g \circ f)}{1+(\overline{f_z}/f_z)  \overline{\mu_f} (\mu_g \circ f)}.
\end{equation}
\end{theorem}

Conversely, a quasi-conformal mapping can also be determined by a complex function. Given a Beltrami coefficient $\mu:\mathbb{C}\to \mathbb{C}$ with $\|\mu\|_\infty < 1$, there is always a quasi-conformal mapping from $\mathbb{C}$ onto itself which satisfies the Beltrami equation (\ref{eqt:beltrami}) in the distribution sense \cite{Gardiner00}. Therefore, finding an optimal quasi-conformal mapping is equivalent to finding an optimal complex-valued function.

More explicitly, let $\mu$ be a given Beltrami coefficient associated with a quasi-conformal map $f:\mathbb{C} \to \mathbb{C}$. The computation of $f$ can be achieved by solving an equation analogous to the Laplace equation $\Delta f = 0$. The method is outlined as follows.

Denote $f = u+iv$ and $\mu = \rho + i \tau$. By considering the Beltrami Equation (\ref{eqt:beltrami}), each pair of the partial derivatives $v_x, v_y$ and $u_x, u_y$ can be expressed as linear combinations of the other \cite{Lui13},
\begin{equation} \label{eqt:firstorder}
 \begin{split}
  v_y &= \alpha_1 u_x + \alpha_2 u_y; \\
  -v_x &= \alpha_2 u_x + \alpha_3 u_y,
 \end{split}
  \ \ \text{ and } \ \
 \begin{split}
  -u_y &= \alpha_1 v_x + \alpha_2 v_y; \\
  u_x &= \alpha_2 v_x + \alpha_3 v_y,
 \end{split}
\end{equation}
where $\alpha_1 = \frac{(\rho-1)^2+\tau^2}{1-\rho^2-\tau^2}; \alpha_2 = -\frac{2\tau}{1-\rho^2-\tau^2}; \alpha_3 = \frac{(1+\rho)^2 + \tau^2}{1-\rho^2-\tau^2}$. Since $\nabla \cdot \left( \begin{array}{c}
 -v_y \\ v_x
\end{array} \right) = 0$ and $\nabla \cdot \left( \begin{array}{c}
 -u_y \\ u_x
\end{array} \right) = 0$ , $f$ can be computed by solving the following equations
\begin{equation}\label{eqt:secondorder}
\nabla \cdot \left(A \left(\begin{array}{c}
u_x\\
u_y \end{array}\right) \right) = 0\ \ \mathrm{and}\ \ \nabla \cdot \left(A \left(\begin{array}{c}
v_x\\
v_y \end{array}\right) \right) = 0
\end{equation}
where
$\displaystyle A = \left( \begin{array}{cc} \alpha_1 & \alpha_2\\
\alpha_2 & \alpha_3 \end{array}\right)$. We call Equation (\ref{eqt:secondorder}) the \emph{generalized Laplace equation}. It can be observed that if $\mu = 0$, then Equation (\ref{eqt:secondorder}) becomes the usual Laplace equation for solving conformal maps. In the discrete case, the above elliptic PDEs (\ref{eqt:secondorder}) can be discretized into sparse symmetric positive definite linear systems as described in \cite{Lui13,Choi15a}.

The analogy between the Laplace equation and the generalized Laplace equation can be interpreted in terms of the Riemannian metric. Given a quasi-conformal mapping $(f,|dz|^2)$ associated with the Beltrami coefficient $\mu$, we can consider the distorted metric tensor $|dz+\mu d\bar{z}|^2$. Then, the mapping $(f,|dz+\mu d\bar{z}|^2)$ is in fact a conformal mapping with respect to the distorted metric. This can be explained by the following proposition.

\begin{prop}[See \cite{Choi15d}, page 13]
Let $(S_1, \sigma |dz|^2)$ and $(S_2, \rho |dw|^2)$ be two domains on $\mathbb{C}$ with the metric tensors $\sigma |dz|^2$ and $\rho |dw|^2$, and let $\mu$ be a prescribed Beltrami coefficient on $S_1$ associated with a quasi-conformal map $f:(S_1, \sigma |dz|^2) \to (S_2, \rho |dw|^2)$. Then, the map solved by the generalized Laplace equation (\ref{eqt:secondorder}) is a harmonic map between $(S_1, |dz+\mu d\bar{z}|^2)$ and $(S_2, \rho |dw|^2)$. Consequently, solving the generalized Laplace equation (\ref{eqt:secondorder}) is equivalent to solving the Laplace equation with respect to the distorted metric tensor $|dz+\mu d\bar{z}|^2$.
\end{prop}

Finally, the abovementioned definitions and concepts of quasi-conformal maps between planar domains can be naturally extended onto Riemann surfaces by introducing the concept of {\it Beltrami differential}. A Beltrami differential $\mu(z)\frac{\overline{dz}}{dz}$ on a Riemann surface $S$ is an assignment to each chart $(U_\alpha, \phi_\alpha)$ of an $L_\infty$ complex-valued function $\mu_\alpha$, defined on local parameter $z_\alpha$ such that $ \mu_\alpha \frac{d \overline{z_\alpha}}{dz_\alpha} = \mu_\beta \frac{d \overline{z_\beta}}{dz_\beta},$ on the domain which is also covered by another chart $(U_\beta, \phi_\beta )$. Here, $\frac{dz_\beta}{dz_\alpha} = \frac{d}{dz_\alpha}\phi_{\alpha\beta}$ and $\phi_{\alpha\beta} = \phi_\beta \circ \phi_\alpha$. An orientation preserving diffeomorphism $f: M \to N$ is said to be \emph{quasi-conformal} associated with $\mu(z)\frac{\overline{dz}}{dz}$ if for any chart $(U_\alpha, \phi_\alpha)$ on $M$ and any chart $(U_\beta, \psi_\beta)$ on $N$, the mapping $f_{\alpha\beta} := \psi_\beta \circ f \circ \phi_\alpha^{-1}$ is quasi-conformal associated with $\mu_\alpha \frac{d \overline{z_\alpha}}{dz_\alpha}$.



\subsection{Teichm\"uller mappings}
Recall that quasi-conformal mappings are with bounded conformality distortions. In the space of all quasi-conformal mappings, it is natural to consider a type of quasi-conformal mappings, called the Teichm\"uller mappings, which produce uniform conformality distortions with prescribed landmark constraints. Mathematically, the Teichm\"uller mappings are defined as follows.

\begin{defn}[Teichm\"uller mapping]
Let $f:M \to N$ be a quasi-conformal mapping. $f$ is said to be a \emph{Teichm\"uller mapping (T-map)} associated with the quadratic differential $q = \varphi dz^2$ where $\varphi:M \to \mathbb{C}$ is a holomorphic function if its associated Beltrami coefficient is of the form
\begin{equation}\label{Teichmullermap}
\mu(f) = k \frac{\overline{\varphi}}{|\varphi|}
\end{equation}
for some constant $k < 1$ and quadratic differential $q \neq 0 $ with $||q||_1 = \int_{S_1} |\varphi| <\infty$.
\end{defn}

It follows that a Teichm\"uller mapping is a quasi-conformal mapping whose Beltrami coefficient has a constant norm. Hence, a Teichm\"uller mapping has a uniform conformality distortion over the whole domain.

Furthermore, Teichm\"uller mappings are closely related to a class of mappings called extremal quasi-conformal mappings. Before explaining their relationships, we introduce the precise definition of extremal quasi-conformal mappings.
\begin{defn}[Extremal quasi-conformal mapping]
Let $f:M \to N$ be a quasi-conformal mapping. $f$ is said to be an \emph{extremal quasi-conformal mapping} if for any quasi-conformal mapping $h:M \to N$ isotopic to $f$ relative to the boundary, we have
\begin{equation}\label{extremalmap}
K(f) \leq K(h)
\end{equation}
where $K(f)$ is the maximal quasi-conformal dilation of $f$. It is uniquely extremal if the inequality (\ref{extremalmap}) is strict when $h \neq f$.\\
\end{defn}

Under suitable boundary conditions, the Teichm\"uller mappings and the extremal quasi-conformal mappings from the unit disk to the unit disk are in fact equivalent.

\begin{thm}[Landmark-matching Teichm\"uller mapping of $D$]\label{landmarkteichmullerdisk}
Let $g:\partial D \to \partial D$ be an orientation-preserving diffeomorphism of $\partial D$, where $D$ is the unit disk. Suppose further that $g'(e^{i\theta}) \neq 0$ and $g''(e^{i\theta})$ is bounded. Let $\{p_i\}_{i=1}^n \in D$ and $\{q_i\}_{i=1}^n \in D$ be the corresponding interior landmark constraints. Then there exists a unique Teichm\"uller mapping $f:(D,\{p_i\}_{i=1}^n) \to (D, \{q_i\}_{i=1}^n)$  matching the interior landmarks, which is the unique extremal extension of $g$ to $D$. Here $(D,\{p_i\}_{i=1}^n)$ denotes the unit disk $D$ with prescribed landmark points $\{p_i\}_{i=1}^n$.
\end{thm}

Therefore, besides equipped with uniform conformality distortions, Teichm\"uller mappings are extremal in the sense that they minimize the maximal quasi-conformal dilation $K$. In this paper, our goal is to compute landmark-matching Teichm\"uller extremal mappings between point clouds.

\subsection{Point clouds and some related concepts}
In this subsection, we briefly introduce point clouds and some mathematical concepts about this data structure. To avoid ambiguity, normal characters are used for introducing the continuous theories and bold characters are used for the mappings of point clouds.

Intuitively, a {\em point cloud} $P=\{p_1,p_2, \dots, p_N\} \subset \mathbb{R}^3$ is a set of points sampled from a Riemann surface $\mathcal{M}$. $\mathcal{M}$ is called the underlying surface of $P$. In this paper, we focus on disk-like point clouds, the point clouds with the underlying surfaces being simply-connected open surfaces. As disk-like point clouds can be conformally mapped onto a rectangular domain, it suffices to introduce the concepts on planar point clouds on $\mathbb{R}^2$.

A {\em planar point cloud} $P=\{p_1,p_2, \dots, p_N\} \subset \Omega$ is a set of points sampled from a disk-type domain $\Omega\subset \mathbb{R}^2$, where both the domains and the sampled point clouds are assumed to satisfy the following conditions \cite{Wendland04}.

\begin{defn}[Interior cone condition] \label{def:domain}
A domain $\Omega\subset \mathbb{R}^2$ is said to \emph{satisfy the interior cone condition} with parameter $r>0$ and $\theta\in (0,\pi/2)$ if for each $x\in \Omega$, there exists a unit vector $d(x)$ such that $C(x,d,\theta, r)\subseteq \Omega$, where
\begin{equation}
C(x,d,\theta, r) = \{x+ty: y\in \mathbb{S}^1, \,y^Td(x)\geq \cos \theta, \,0\leq t\leq r\}.
\end{equation}
\end{defn}
\begin{defn}[Quasi-uniform point cloud] \label{def:pc}
Let $P$ be a point cloud sampled from a planar domain $\Omega\subset \mathbb{R}^2$. The \emph{fill distance} $h_{P,\Omega}$ and the \emph{separation distance} $q_{P}$ are respectively defined by
\begin{equation}
h_{P,\Omega} := \sup_{x\in \Omega} \min_{p\in P} \|x-p\|
\end{equation}
and
\begin{equation}
q_{P} := \frac{1}{2}\min_{p_1,p_2\in P\atop p_1\neq p_2} \|p_1-p_2\|.
\end{equation}
$P$ is said to be \emph{quasi-uniform} with positive constant $c_{qu}$ if
\begin{equation}
q_{P} \leq h_{P,\Omega} \leq c_{qu}q_{P}.
\end{equation}
\end{defn}

From now on, we assume that all planar point clouds are quasi-uniform and are sampled from domains which satisfy the interior cone condition. On point clouds, it is difficult to compute interpolation functions. In order to compute quasi-conformal mappings on point clouds, it is important to develop a suitable approximation scheme.

\subsection{The moving least square method for approximating differential operators on point clouds}
One common approach to approximate functions and derivatives on point clouds is the \emph{Moving Least Square (MLS)} method \cite{Liu97,Levin98,Nealen04,Breitkopf05,Liang12,Liang13,Choi15d}. A concise introduction of the second order MLS method is given below. For more details, we refer the readers to \cite{Wendland04, Mirzaei12, Mirzaei15}.

Given a function $f:\Omega\to \mathbb{R}$ and a disk-type point cloud $P$ sampled from $\Omega$ with fill distance $h$, the MLS method approximates $f$ locally near every point. Define the mappings $q,q_1,q_2:\Omega\to \mathbb{R}^6$ respectively by
\begin{equation}
 q(x_1,x_2)=[1,x_1,x_2,x_1^2,x_1x_2,x_2^2]^T,
\end{equation}
\begin{equation} \label{eqt:q1}
q_1(x_1,x_2)=\partial_1 q= [0,1,0,2x_1,x_2,0]^T,
\end{equation}
and
\begin{equation} \label{eqt:q2}
q_2(x_1,x_2)=\partial_2 q = [0,0,1,0,x_1,2x_2]^T.
\end{equation}
Then, the local approximation near $x$ is given by $c_x^Tq(y)$, where
\begin{equation}
c_x = argmin_c \, \sum_{i=1}^N w(\|x-p_i\|)\left(c^Tq(p_i)-f_i\right)^2.
\end{equation}
Here $w$ is a weight function compactly supported in $[0,C_\delta h]$ and $C_\delta$ is a constant defined in \cite{Wendland04}.
The solution $c_x$ can be expressed as
\begin{equation}
 c_x = (Q^TW(x)Q)^{-1}Q^TW(x)F,
\end{equation}
where $W(x)$ is a diagonal matrix whose diagonal elements are $W_{ii} = w(\|x-p_i\|)$, and $Q$, $F$ are respectively the matrix forms of $q$ and $f$. From now on, unless otherwise specified, we use the corresponding capital letter of a given point cloud mapping to denote its matrix form, in which the $i$-th row is the function value on $p_i$.

\begin{table}[t]
\centering
\begin{tabular}{|l|l|}
\hline
Weight & Formula of $w(d)$\\
\hline
Constant weight & $\displaystyle w(d)=1$\\
\hline
Exponential weight & $\displaystyle w(d)=\exp \left(-\frac{d^2}{D^2} \right)$\\
\hline
Inverse of squared distance weight & $\displaystyle w(d)=\frac{1}{d^2 + \epsilon^2}$\\
\hline
Wendland weight \cite{Wendland95,Wendland01} & $\displaystyle w(d)=\left(1-\frac{d}{D}\right)^4\left(\frac{4d}{D}+1\right)$\\
\hline
Special weight \cite{Liang12} & $\displaystyle w(d)=\left\{\begin{array}{ll} 1 & \text{if }d=0\\ \frac{1}{K} & \text{if }d\neq 0\\ \end{array} \right.$ \\
\hline
Gaussian weight \cite{Choi15d} & $\displaystyle w(d)=\left\{\begin{array}{ll} 1 & \text{if }d=0\\ \frac{1}{K}\exp(\frac{-\sqrt{K}d^{2}}{D^{2}})& \text{if }d\neq 0\\ \end{array} \right.$ \\
\hline
\end{tabular}
\caption{Some common weight functions for the moving least square approximation. Here $K$ is the number of points in the chosen neighborhood, and $D$ is the maximal distance of the neighborhood.}
\label{table:weight}
\end{table}

Numerous studies have been devoted to the choice of the weight function $w$ in the MLS method. Some common weight functions are listed in Table \ref{table:weight}. Denote
\begin{equation} \label{eqt:ax}
A_x = (Q^TW(x)Q)^{-1}Q^TW(x).
\end{equation}
 Then, the function is locally approximated by
\begin{equation}
\tilde{f}(y) = q^T(y)A_xF.
\end{equation}
Hence, the derivative approximation naturally follows. In Section \ref{main}, we establish an approximation scheme based on the MLS method to compute quasi-conformal map on planar point clouds.

The following error estimates for MLS have been established in \cite{Mirzaei12}:
\begin{prop}[See Corollary 4.10, \cite{Mirzaei12}]
\begin{equation}\label{eqt:mls_error}
\|q_j^T(x) A_x \|_1 = O(h^{-1}),
\end{equation}
where $h$ is the fill distance.
\end{prop}

\begin{prop}[See Theorem 4.3, \cite{Mirzaei12}]
\begin{equation}\label{eqt:mls_error2}
\|\partial_i \partial_j f - (\partial_i \partial_j q^T(y)) A_x f\| = O(h)
\end{equation}
where $h$ is the fill distance.
\end{prop}

\subsection{Quasi-conformal geometry on point clouds}
Recall that in quasi-conformal theory, quasi-conformal mappings are closely related to the Beltrami coefficients. Therefore, to compute quasi-conformal mappings between point clouds, it is necessary to extend the definition of the Beltrami coefficients to point clouds. In this work, we adopt the definition of \emph{discrete diffuse point cloud Beltrami coefficients (PCBC)} proposed by Meng and Lui in \cite{Meng15} for the computation of the Beltrami coefficient $\tilde{\bm{\mu}}$ on point clouds.
\bigbreak
\begin{defn}[Discrete diffuse point cloud Beltrami coefficients \cite{Meng15}] \label{def:pcbc}
Given a point cloud $P =\{p_i\in \mathbb{R}^2: i=1,\cdots,N\}$ with underlying domain $\Omega$ and a target point cloud function $\bold{f}:P \to \mathbb{R}^2$, where $\bold{f}=(\bold{u},\bold{v})^T$, the \emph{diffuse point cloud Beltrami coefficient (PCBC)} $\tilde{\mu}:\Omega\to \mathbb{C}$ is defined by
\begin{equation}
\tilde{\mu}(x) = \frac{
\begin{bmatrix}
q_1^T(x) & q_2^T(x)
\end{bmatrix}
\begin{bmatrix}
A_x & iA_x\\
iA_x & -A_x
\end{bmatrix}
\begin{bmatrix}
U\\V
\end{bmatrix}
}{\begin{bmatrix}
q_1^T(x) & q_2^T(x)
\end{bmatrix}
\begin{bmatrix}
A_x & iA_x\\
-iA_x & A_x
\end{bmatrix}
\begin{bmatrix}
U\\V
\end{bmatrix}
},
\end{equation}
where $q_1$, $q_2$ and $A_x$ are respectively defined as in Equation (\ref{eqt:q1}), Equation (\ref{eqt:q2}) and Equation (\ref{eqt:ax}). In addition, the discrete diffuse PCBC is a complex valued point cloud function defined by $\bm{\tilde{\mu}} = \tilde{\mu}|_P$.
\end{defn}

The following proposition provides an error estimate for the discrete diffuse PCBC.
\begin{prop}[See \cite{Meng15}, page 9] \label{prop:error}
Let $P$ be a point cloud sampled from $\Omega$ with fill distance $h$, and $\bf{f}$ be the point cloud map corresponding to $f$. Denote the Beltrami coefficient associated with $f$ by $\mu$, and the discrete diffuse PCBC associated with $\bf{f}$ by $\bm{\tilde{\mu}}$. Then there exists some constants $C(f)$ such that if $h \leq C(f)$, we have the following error bound
\begin{equation}
|\mu(x) - \bm{\tilde{\mu}} (x)| \leq C(f) h^2.
\end{equation}
\end{prop}

Note that the abovementioned concepts are defined on planar domains. To extend the concept of quasi-conformal point cloud mappings on planar domains to the general 3D case, Meng and Lui \cite{Meng15} introduced a concept called $\epsilon$-conformal parameterization for 3D point clouds, which is an analogue of conformal parameterization for 3D surfaces. Then, quasi-conformal point cloud mappings can be defined as a composition of $\epsilon$-conformal parameterizations and planar quasi-conformal mappings. Readers are referred to \cite{Meng15} for more details.

\section{A discrete analogue of Teichm\"uller mappings on point clouds} \label{theory}
The continuous Teichm\"uller mappings have been extensively studied by mathematicians. By contrast, the development of the Teichm\"uller mappings on point clouds remains a great challenge. In this section, we establish a novel discrete version of the Teichm\"uller extremal mappings on point clouds with theoretical guarantees. First, we define the discrete version of Teichm\"uller mappings on planar point clouds.

\bigbreak
\begin{defn}[Planar point cloud Teichm\"uller mappings] \label{def:pctm}
Let $P$ be a quasi-uniform point cloud with underlying domain $\Omega\subset \mathbb{R}^2$ satisfying the interior cone condition. A point cloud mapping $\bold{f}$ is called a \emph{planar point cloud Teichm\"uller mapping (planar PCT-map)} if its discrete diffuse PCBC $\bm{\mu}$ has a constant norm $k = |\bm{\mu}(p)|$ for any $p\in P$, and satisfies the following equation
\begin{equation}
(L_i\bm{\mu})\overline{\bm{\mu}(p_i)} \in \mathbb{R}
\end{equation}
for any interior point $p_i\in P$. Here $L$ is the approximation of the Laplace-Beltrami operator by the MLS method, and $L_i$ denotes the i-th row of $L$.
\end{defn}

\bigbreak
Our goal is to prove that under certain convergence conditions, the planar PCT-maps converge to the continuous T-maps defined on the underlying domain. We now consider the following lemma.

\bigbreak
\begin{lemma} \label{lemma:holo}
Let $R$ be any rectangular domain, and $u,v\in C^3(R)$ satisfy $u\Delta v = v \Delta u$ and $u^2+v^2=1$. Then there exists a real valued function $f$ such that $\phi = f\cdot (u+iv)$ is holomorphic, and $|f|\neq 0$ in the domain.
\end{lemma}
\begin{proof}
Consider the Cauchy-Riemann equation
\begin{equation}
\begin{split}
&f_x u + u_x f = f_y v +v_yf;\\
&f_y u + u_y f = -f_x v - v_x f.
\end{split}
\end{equation}
Since $u^2+v^2=1$, the above system is equivalent to
\begin{equation}
\begin{split}
&f_x=G_1f;\\
&f_y=G_2f,
\end{split}
\end{equation}
where
\begin{equation}
\begin{split}
&G_1= -uu_x-vv_x-vu_y+uv_y;\\
&G_2 = -uu_y-vv_y-uv_x+vu_x.
\end{split}
\end{equation}
Let the domain $R$ be $[a_1,a_2]\times [b_1,b_2]$. By solving the first equation above, we obtain
\begin{equation}\label{pf:lemma2-sol}
f = \exp\left(\int_{a_1}^x G_1(s,y) ds + g(y)\right),
\end{equation}
where $g$ is a function of $y$. Plugging this solution into the second equation, we have
\begin{equation}
\int _{a_1}^x \partial_y G_1(s,y) ds + g_y(y) = G_2.
\end{equation}
This equation has a solution if $\partial_y G_1 = \partial_x G_2$. By a direct calculation,
\begin{equation}
\partial_y G_1 - \partial_x G_2
= u\Delta v - v\Delta u = 0.
\end{equation}
Therefore, this system has a solution in the form of Equation (\ref{pf:lemma2-sol}), where
\begin{equation}
g = \int_{b_1}^y G_2(x,t)dt - \int_{a_1}^x (G_1(s,y)-G_1(s,b_1)) ds + C.
\end{equation}
Consequently, the solution $f$ is given by
\begin{equation}
f = \exp\left(\int_{b_1}^y G_2(x,t)dt + \int_{a_1}^x G_1(s,b_1)ds + C\right).
\end{equation}
It is obvious that $|f|\neq 0$. Moreover, as $f\cdot(u+iv)$ is a continuous function satisfying the Cauchy-Riemann equation, it is holomorphic. This proves the lemma.
\end{proof}

\bigbreak

With the aid of the above lemma, we are ready to establish the following theorem about the consistency between planar PCT-maps and the continuous T-maps.

\bigbreak
\begin{prop} \label{prop:pctm2tm}
Let $\{P_n\}$ be an ascending sequence of point clouds (in other words, $P_1 \subset P_2 \subset \cdots P_n \subset P_{n+1} \subset \cdots$) sampled from a rectangular domain $R_1$ with fill distance $h_n \to 0$. For each $P_n$, let $\bold{g}_n$ be a PCT-map defined on it with discrete diffuse PCBC $\tilde{\bm{\mu}}_n$. Further assume that there exists a smooth quasi-conformal mapping $g:R_1\to R_2$ with Beltrami coefficient $\mu$ such that for any point $p\in P:=\cup_m P_m$, $\bold{g}_n(p)$ converges to $g(p)$. Let the sup-norm of the error be given by $\epsilon_n = \|\bold{g}_n - g|_{P_n}\|_\infty$ and assume that $\lim_n \epsilon_nh_n^{-2} = 0$. Then $g$ is a Teichm\"uller mapping.
\end{prop}
\begin{proof}
First, we prove that $\tilde{\bm{\mu}}_n(p)$ converges to $\mu(p)$ for each $p\in P_m$, where $m$ is an arbitrary number.

Let $\tilde{\bm{\nu}}_n$ be the discrete diffuse PCBC of a point cloud mapping $g|_{P_n}$, and $\tilde{\bm{\nu}}_n(p)$ be its value at a point $p$. We have
\begin{equation}
|\mu(p) - \tilde{\bm{\mu}}_n(p)|
\leq |\mu(p) - \tilde{\bm{\nu}}_n(p)| + |\tilde{\bm{\nu}}_n(p) - \tilde{\bm{\mu}}_n(p)|.
\end{equation}
By Proposition \ref{prop:error}, we have $|\mu(p) - \tilde{\bm{\nu}}_n(p)|\leq C(g)h_n^2$. Let $\tilde{G}$ and $\tilde{G}_n$ be the vector forms of the point cloud mappings $g|_{P_n}$ and $\bold{g}_n$. For simplicity, we adopt the following notation.
\begin{equation}
\begin{split}
&D_1 = \left(\begin{array}{cc}
q_1^T(x) & q_2^T(x)
\end{array}\right)
\left(\begin{array}{cc}
A_x & iA_x\\
iA_x & -A_x
\end{array}\right), \\
&D_2 = \left(\begin{array}{cc}
q_1^T(x) & q_2^T(x)
\end{array}\right)
\left(\begin{array}{cc}
A_x & iA_x\\
-iA_x & A_x
\end{array}\right).
\end{split}
\end{equation}
Then we have
\begin{equation}
\begin{split}
&|\tilde{\bm{\nu}}_n(p) - \tilde{\bm{\mu}}_n(p)|\\
\leq & \left|\frac{D_1(p) \tilde{G}}{D_2(p) \tilde{G}}
-
\frac{D_1(p) \tilde{G}_n}{D_2(p) \tilde{G}_n}
\right|\\
\leq & \frac{
\left|D_1(p) \tilde{G}D_2(p)\left( \tilde{G}_n-\tilde{G} \right)\right| +
\left|D_2(p) \tilde{G}D_1(p)\left( \tilde{G}-\tilde{G}_n \right)\right| }
{\left|D_2(p)\tilde{G}\cdot D_2(p)\tilde{G}_n\right|}\\
\leq & \frac{2\|D_1(p)\|_1 \cdot \|D_2(p)\|_1 \cdot \|\tilde{G}\|_\infty \cdot \|\tilde{G}-\tilde{G}_n\|_\infty}{\left|D_2(p)\tilde{G}\cdot D_2(p)\tilde{G}_n\right|}\\
= & \frac{2O(\epsilon_n h_n^{-2}) }{\left|D_2(p)\tilde{G}\cdot D_2(p)\tilde{G}_n\right|}
\end{split}
\end{equation}
as $\|D_i(p)\|_1 \leq O(h^{-1})$ and $\tilde{G}$ is bounded. Also, by a similar analysis, we have
\begin{equation}
\begin{split}
& |D_2(p)\tilde{G} - g_z(p)| = O(h_n^2),\\
& |D_2(p)\tilde{G} - D_2(p)\tilde{G}_n| = O(\epsilon_nh_n^{-1}).
\end{split}
\end{equation}
Therefore, when $n$ is large enough,
\begin{equation}
\begin{split}
|\mu(p) - \tilde{\bm{\mu}}_n(p)|\leq  O(h_n^2) + |g_z(p)|^{-2}O(\epsilon_nh_n^{-2}) = O(h_n^2 +\epsilon_nh_n^{-2}).
\end{split}
\end{equation}
This implies that $\tilde{\bm{\mu}}_n(p)$ converges to $\mu(p)$ for each $p\in P$.

Then, since $|\tilde{\bm{\mu}}_n|=k_n$ is a constant for each $n$, it follows that $|\mu|=k$ on $P$ and $k=\lim_nk_n$. Also, since $\lim_nh_n=0$, $P$ is dense in $R_1$. By the continuity of $\mu$, $|\mu(x)|=k$ for any $x\in R_1$.

It remains to prove that $\mu = k\bar{\varphi}/|\varphi|$ where $\varphi$ is holomorphic. Assume $k\neq 0$. Now, we fix a point cloud $P_n$ and consider the matrix $L$ which is the approximation of the Laplace-Beltrami operator. By Equation (\ref{eqt:mls_error2}), for any interior point $p=p_i\in P_n$,
\begin{equation}
\begin{split}
&\left|\left(\Delta \mu(p) \right)\overline{\mu(p)} - \left(L_i \tilde{\bm{\mu}}_n\right)\overline{\tilde{\bm{\mu}}_n(p)}\right|\\
\leq &\|\Delta\mu\|_\infty \cdot\left|\mu(p)-\tilde{\bm{\mu}}_n(p)\right| + k_n|\Delta\mu(p) - L_i\tilde{\bm{\mu}}_n|\\
= & \|\Delta\mu\|_\infty O(h_n^2 +\epsilon_nh_n^{-2}) + k_nO(h_n)\\
= & O(h_n^2 +\epsilon_nh_n^{-2}).
\end{split}
\end{equation}
On the other hand, by the definition of PCT-map, for any $n$,
\begin{equation}
\left(L_i \tilde{\bm{\mu}}_n\right)\overline{\tilde{\bm{\mu}}_n(p)} \in \mathbb{R}.
\end{equation}
Thus, $\left(\Delta \mu(p) \right)\overline{\mu(p)}\in \mathbb{R}$ if $p$ is an interior point of $P_n$.

Consider an arbitrary point $p\in R_1$. If $p\in P- \partial R_1$, then the above conclusion holds. If $p \in P\cap \partial R_1$ or $p\in R_1-P$, since $P$ is dense in $R_1$, by the continuity of $\mu$ and the Laplace-Beltrami operator, the above conclusion also holds. Therefore, $\left(\Delta \mu(p) \right)\overline{\mu(p)}\in \mathbb{R}$ for any $p\in R_1$.

Let $\mu = k(\rho + i\tau)$. By a direct calculation,
\begin{equation}
\begin{split}
0&=\frac{1}{k^2}Im\left(\left(\Delta \mu \right)\overline{\mu}\right) = \rho \Delta \tau - \tau \Delta \rho.
\end{split}
\end{equation}
By taking $u=\rho$, and $v= -\tau$ in Lemma \ref{lemma:holo}, there exists a real-valued function $f$ such that $\varphi = f\cdot (\rho-i\tau)$ is holomorphic. Moreover, $\mu = k\bar{\varphi}/|\varphi|$. It follows that $g$ is a Teichm\"uller mapping.
\end{proof}

\bigbreak

From the above proposition, we have established the notion of the planar PCT-maps on planar disk-type point clouds, and proved the consistency between planar PCT-maps and the continuous T-maps. Then, the next step is to define the Teichm\"uller mappings between general disk-type point cloud surfaces. 

The generalization is outlined as follows. Let $\bold{f}$ be a point cloud mapping between two point cloud surfaces $P_1$ and $P_2$, and $\bm{\phi}_i$ be the $\epsilon$-conformal parameterization of $P_i$. Then we define $\bold{f}$ to be a \emph{point cloud Teichm\"uller mapping (PCT-map)} between the two point cloud surfaces if $\bm{\phi}_2^{-1}\circ \bold{f}\circ \bm{\phi}_1$ is a planar PCT-map.

Using the above definition, a property analogous to Proposition \ref{prop:pctm2tm} can be easily derived for the PCT-maps between two general point cloud surfaces with underlying disk-type Riemann surfaces. We omit the details here.

In conclusion, in order to compute the PCT-maps between two point cloud surfaces, we can simplify the problem by first computing the point cloud conformal parameterizations onto planar domains and then the planar PCT-maps. In the next section, we develop efficient and accurate algorithms for computing the PCT-maps between feature-endowed point clouds.


\section{Our proposed TEMPO method for computing PCT-maps} \label{main}
\begin{figure}[t]
 \centering
 \includegraphics[width=\textwidth]{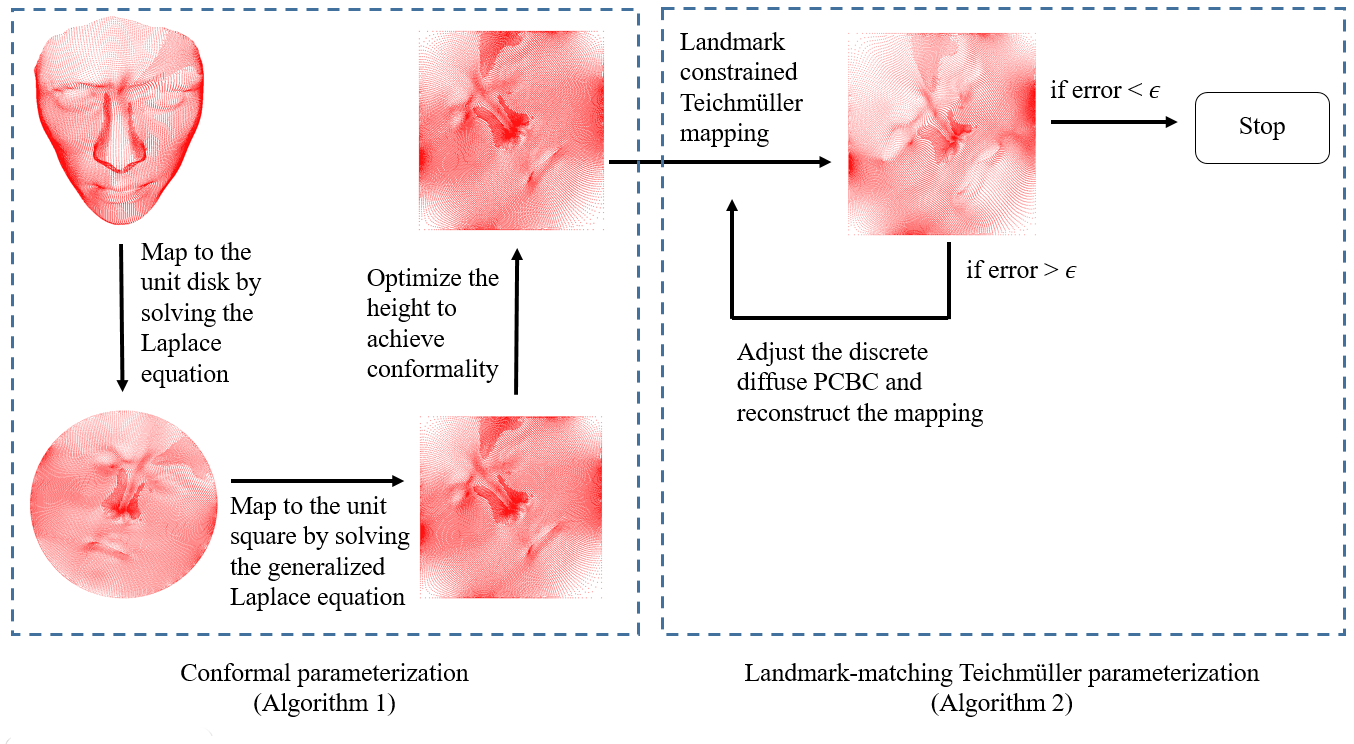}
 \caption{The pipeline of our TEMPO algorithm.}
 \label{fig:pipeline}
\end{figure}
In this section, we describe our proposed TEMPO method for computing Teichm\"uller mappings between point clouds with disk topology under prescribed feature correspondences. The pipeline of our proposed TEMPO algorithm is shown in Figure \ref{fig:pipeline}. It can be observed that our TEMPO algorithm involves computing a number of mappings. Hence, it is necessary to first establish effective numerical schemes for computing different mappings. Then, we can introduce our proposed TEMPO algorithm and finally discuss its application in shape analysis.

This section primarily consists of four subsections:
\begin{enumerate}
\item A hybrid quasi-conformal reconstruction scheme for computing quasi-conformal mappings on point clouds.
\item An improved approximation of the differential operators arising from the computation of conformal and quasi-conformal mappings on point clouds with disk topology.
\item The algorithmic details of the TEMPO algorithm.
\item The induced Teichm\"uller metric for shape analysis of disk-type point clouds.
\end{enumerate}



\subsection{A hybrid quasi-conformal reconstruction scheme for computing quasi-conformal mappings on planar point clouds} \label{tempo_subsection1}
Quasi-conformal mappings have been extensively studied on general smooth surfaces and meshes but not on point clouds before. In this subsection, we develop a hybrid scheme for computing quasi-conformal mappings on point clouds with a prescribed point cloud Beltrami coefficient (PCBC).

In the continuous case, there are several existing methods to compute quasi-conformal mappings associated with a given feasible Beltrami coefficient $\mu = \rho+i\tau$. One method is to solve the linear system (\ref{eqt:firstorder}). Another approach for computing quasi-conformal mappings is the \emph{Linear Beltrami Solver (LBS)} proposed by Lui et al. in \cite{Lui13}, which aims to solve the generalized Laplace equation (\ref{eqt:secondorder}).

The abovementioned methods work well on triangular meshes. However, for the case of point clouds, solely solving the system (\ref{eqt:firstorder}) may not result in any meaningful solutions. More specifically, the solution to the system (\ref{eqt:firstorder}) is highly sensitive to the input PCBC. If the PCBC is not sufficiently close to a feasible Beltrami coefficient in the continuous case, the system (\ref{eqt:firstorder}) may not lead to a quasi-conformal diffeomorphism for point clouds. Also, the solution to the system (\ref{eqt:secondorder}) may be associated with a PCBC which is largely different from the input Beltrami coefficients. In other words, the accuracy of the system (\ref{eqt:secondorder}) hinders the computation of quasi-conformal mappings on point clouds. Therefore, we cannot directly apply any of the above methods to compute quasi-conformal mappings on point clouds. To alleviate the abovementioned issues, we propose a hybrid quasi-conformal reconstruction scheme for point clouds by a fusion of the two methods.

Instead of merely adapting either of the two approaches, we consider taking a balance between them in order to achieve a practical scheme for computing quasi-conformal mappings on point clouds. Specifically, we consider solving the following hybrid PDE for computing quasi-conformal mappings:
\begin{equation}\label{eqt:hybrid-cont}
\left(\begin{array}{cc}
\alpha_1\partial_x + \alpha_2\partial_y & -\partial_y\\
\alpha_2\partial_x + \alpha_3\partial_y & \partial_x
\end{array}\right)
\left(\begin{array}{c}
u\\v
\end{array}\right)
+
\gamma
\left(\begin{array}{cc}
\mathcal{L} & O\\
O & \mathcal{L}
\end{array}\right)\left(\begin{array}{c}
u\\v
\end{array}\right) = 0,
\end{equation}
where $\mathcal{L}$ is the generalized Laplace-Beltrami operator, and $\gamma$ is a user-defined positive parameter.

To solve the hybrid PDE (\ref{eqt:hybrid-cont}) on point clouds, we first discretize the system (\ref{eqt:firstorder}) and denote it by
\begin{equation} \label{eqt:firstorder-disc}
 M_1(\bm{\mu})\left( \begin{array}{c}U\\V\end{array}\right) =0.
\end{equation}
Similarly, the system (\ref{eqt:secondorder}) is discretized and denoted by
\begin{equation} \label{eqt:secondorder-disc}
M_2(\bm{\mu})\left( \begin{array}{c}U\\V\end{array}\right)=0,
\end{equation}
where
\begin{equation}
M_2(\bm{\mu}) = \left(\begin{array}{cc}
M_3(\bm{\mu}) & O \\ O & M_3(\bm{\mu})
\end{array}\right),
\end{equation}
and $M_3(\bm{\mu})$ is the approximation of the generalized Laplace-Beltrami operator. Here the discretization methods are to be determined. The details of the abovementioned discretizations are described in Section \ref{sec:beltrami} and Section \ref{sec:glap}.

Finally, to compute a quasi-conformal mapping with a given PCBC $\bm{\mu}$, we propose to solve the hybrid system
\begin{equation}\label{eqt:hybrid}
\left(M_1(\bm{\mu}) + \gamma M_2(\bm{\mu})\right)\left( \begin{array}{c}U\\V\end{array}\right) = 0.
\end{equation}

One advantage of our proposed hybrid scheme is that it is theoretically supported. Since both of the system (\ref{eqt:firstorder}) and the system (\ref{eqt:secondorder}) aim to compute quasi-conformal mapping $f$ with the prescribed Beltrami coefficient $\mu$, this mapping $f= (u,v)^T$ is theoretically guaranteed to be a solution to the hybrid PDE (\ref{eqt:hybrid-cont}).

Moreover, our proposed scheme overcomes the drawbacks of each of the two approaches with the aid of the other approach. On one hand, the system (\ref{eqt:firstorder-disc}) produces accurate results but it is highly unstable. By including the second term $M_2$ in Equation (\ref{eqt:hybrid}), we significantly stabilize the computations. On the other hand, the system (\ref{eqt:secondorder-disc}) is stable but not accurate. The first term $M_1$ in Equation (\ref{eqt:hybrid}) helps boosting up the accuracy of the computation. In summary, with a fusion of the two approaches, both the accuracy and the stability of the computation of quasi-conformal mappings on point clouds can be improved.

It is noteworthy that the discretization schemes for the system (\ref{eqt:firstorder}) and the system (\ref{eqt:secondorder}) on point clouds are crucial to the hybrid equation (\ref{eqt:hybrid}). In the following subsection, two discretization schemes are proposed to accurately compute $M_1$ and $M_2$ on point clouds.

\subsection{Approximating the differential operators in Equation (\ref{eqt:firstorder}) and (\ref{eqt:secondorder})} \label{tempo_subsection2}
In this part, we present the numerical schemes to respectively approximate the first order PDE system (\ref{eqt:firstorder}) and the generalized Laplace equation (\ref{eqt:secondorder}) on disk-type point clouds.

For the discretization of Equation (\ref{eqt:firstorder}), we adopt a numerical scheme which involves the moving least square (MLS) method \cite{Liang12,Liang13,Choi15d}. For the discretization of the generalized Laplace equation (\ref{eqt:secondorder}), we propose a combined scheme which involves both the MLS method \cite{Liang12,Liang13,Choi15d} and the local mesh method \cite{Lai13}. It is noteworthy that the choice of the weight function $w$ in the MLS method is crucial for the accuracy of the approximations. In this work, we adopt the Gaussian weight function proposed by Choi et al. \cite{Choi15d}:
\begin{equation}
\displaystyle w(d)=\left\{\begin{array}{ll} 1 & \text{if }d=0\\ \frac{1}{K}\exp(\frac{-\sqrt{K}d^{2}}{D^{2}})& \text{if }d\neq 0\\ \end{array} \right.,
\end{equation}
where $K$ is the number of points in the chosen neighborhood, and $D$ is the maximal distance of the neighborhood.

The details of the discretization schemes are discussed below.

\subsubsection{Approximating the Beltrami equation on planar point clouds} \label{sec:beltrami}
We first present a scheme to discretize the Beltrami equation on point clouds for computing quasi-conformal mappings. Instead of approximating the ordinary Beltrami Equation (\ref{eqt:beltrami}), we aim to approximate the equivalent system (\ref{eqt:firstorder}) on a point cloud $P$.

For each point $p\in P$, the MLS method constructs an approximating function $f_p$ near $p$. To satisfy the system (\ref{eqt:firstorder}) on the point $p$, a straightforward calculation yields the following equation
\begin{equation} \label{eqt:pcbe1}
\left(\begin{array}{c}
q_2^T(p)A_p \\ -q_1^T(p)A_p
\end{array}\right) V
=
A(p)\left(\begin{array}{c}
q_1^T(p)A_p \\ q_2^T(p)A_p
\end{array}\right) U.
\end{equation}
The above system can be reformulated as a linear system of the form $M_1(\bm{\mu})[U^T,V^T]^T=0$, where $M_1(\bm{\mu})$ is a $2N \times 2N$ matrix, with $N$ being the number of points in $P$. More explicitly, for any $i = 1,2,\cdots,N$, we have
\begin{equation}
\left(\begin{array}{c}
(M_1(\bm{\mu}))_i\\
(M_1(\bm{\mu}))_{i+N}
\end{array}\right)
=
\left(
\begin{array}{cc}
\alpha_1(p_i)q_1^T(p_i) + \alpha_2(p_i)q_2^T(p_i) & -q_2^T(p_i)\\
\alpha_2(p_i)q_1^T(p_i) + \alpha_3(p_i)q_2^T(p_i) & q_1^T(p_i)
\end{array}
\right)
\left(
\begin{array}{cc}
A_{p_i} & O\\
O & A_{p_i}
\end{array}\right),
\end{equation}
where $E_i$ denotes the $i$-th row of any matrix $E$. The following lemma establishes a theoretical guarantee for our proposed scheme.
\begin{lemma}
Let $P$ be a quasi-uniform point cloud sampled from a domain that satisfies the interior cone condition. Suppose $\bold{f}=(\bold{u},\bold{v})^T$ is a point cloud mapping defined on $P$ with well-defined discrete diffuse PCBC. Let $\bm{\sigma}$ be a complex-valued point cloud function. Then $\bm{\sigma}$ is the discrete diffuse PCBC of $\bm{f}$ if and only if $M_1(\bm{\sigma})[U^T,V^T]^T =0$.
\end{lemma}
\begin{proof}
Note that $\bold{f}$ satisfies $M_1(\bm{\sigma})[U^T,V^T]^T =0$ if and only if the Equation (\ref{eqt:pcbe1}) holds for all $p\in P$, which is true if and only if
\begin{equation}
(q_1(p)+iq_2(p))^TA_p(U+iV) = \bm{\sigma}(p)(q_1(p)-iq_2(p))^TA_p(U+iV)
\end{equation}
holds for all $p\in P$. Since the discrete diffuse PCBC is well-defined for all points in $P$, by a direct calculation, the above equation holds if and only if $\bm{\sigma}(p)$ is the diffuse PCBC of $\bold{f}$ on $p$. This completes the proof.
\end{proof}

By the above lemma, the solution to the linear system $M_1(\bm{\mu})[U^T,V^T]^T=0$ is associated with a diffuse PCBC which is consistent with the input Beltrami coefficient $\bm{\mu}$, if it is admissible. This ensures the accuracy of the computation of quasi-conformal mappings on point clouds.

\subsubsection{Improving the approximation of the generalized Laplacian for planar point clouds with disk topology} \label{sec:glap}
Next, we introduce our proposed approximation scheme for the generalized Laplace equation (\ref{eqt:secondorder}) on point clouds with disk topology. Note that the generalized Laplacian is closely related to the Laplace-Beltrami operator. For approximating the Laplace-Beltrami operator on point clouds, a number of works have been established with the aid of the moving least square (MLS) method \cite{Liang12,Liang13,Choi15d}. Therefore, it is natural to consider the MLS method for the generalized Laplacian. In the following, we propose an approximation scheme for the generalized Laplacian on disk-type planar point clouds by combining the MLS approach in \cite{Choi15d} and the local mesh approach in \cite{Lai13}.

The existing MLS approaches work well for interior points in disk-type domains. However, unlike the approximations of the Laplace-Beltrami operator on point clouds and the generalized Laplace-Beltrami operator on meshes, the Beltrami coefficients for general point cloud mappings are not locally constant. Therefore, the matrix $A$ in Equation (\ref{eqt:secondorder}) is not constant. Consequently, we need to take its derivatives into considerations. Specifically, we have
\begin{equation}
\begin{split}
&\nabla \cdot \left(A \left(\begin{array}{c}
u_x\\
u_y \end{array}\right) \right)\\
= &
\partial_x(\alpha_1 u_x + \alpha_2 u_y) + \partial_y(\alpha_2 u_x + \alpha_3 u_y)\\
= &(\partial_x\alpha_1 + \partial_y\alpha_2)u_x + (\partial_x\alpha_2 + \partial_y\alpha_3)u_y + \alpha_1 u_{xx} + 2\alpha_2u_{xy} + \alpha_3u_{yy}.
\end{split}
\end{equation}
By applying the MLS method on the derivatives of $\alpha_i$ and $u$, the approximation linear system can be obtained. After certain calculations, it follows that if the point $p=p_i$ is an interior point, then the $i$-th row of matrix $M_3(\bm{\mu})$ is given by
\begin{equation}
\begin{split}
(M_3(\bm{\mu}))_i &= [\alpha_1(P)]^T B_{11}(p) + [\alpha_2(P)]^T(B_{12}(p) + B_{21}(p)) + [\alpha_3(P)]^T B_{22}(p)\\
& + 2\alpha_1(p) (A_p)_4 + 2\alpha_2(p) (A_p)_5 + 2\alpha_3(p)(A_p)_6
\end{split}
\end{equation}
where
\begin{equation}
B_{jl}(p) = A_p^T q_j(p)q_l^T(p)A_p,
\end{equation}
$[\alpha_j(P)]$ denotes a vector whose $l$-th element is $\alpha_j(p_l)$, and $E_j$ denotes the $j$-th row of any matrix $E$.

However, for the boundary points $p\in \partial \Omega$, the abovementioned approach cannot accurately approximate the operator, since the second order MLS method requires more points to provide an accurate result. To alleviate the problem, we apply the local mesh method \cite{Lai13}. For each boundary point $p_i\in \partial \Omega$, a 1-ring structure $\mathcal{T}_i$ is constructed using the Delaunay triangulation. If the point $p_j$ is not included in the 1-ring, the entry $(M_3(\bm{\mu}))_{ij}$ is zero. Otherwise, the entry is given by
\begin{equation}
(M_3(\bm{\mu}))_{ij} = \left\{
\begin{split}
\sum_{T=\{i,j,k\}\in \mathcal{T}_i} \frac{1}{|T|}(p_k-p_i)^T A(T)(p_k-p_j)
\text{\ \ \ if\ } j\neq i\\
\sum_{T=\{i,l,k\}\in \mathcal{T}_i} \frac{1}{|T|}(p_k-p_l)^T A(T)(p_k-p_l)
\text{\ \ \ if\ } j= i
\end{split}\right.
\end{equation}
where
\begin{equation}
A(\{i,j,k\}) = \frac{1}{3}(A(p_i)+A(p_j)+A(p_k)).
\end{equation}

Our proposed combined method is shown to produce more accurate approximations than the conventional approaches for the generalized Laplace-Beltrami operator on disk-type point clouds. Various numerical experiments are provided in Section \ref{experiment} to demonstrate the effectiveness of our proposed approximation scheme.

\subsection{Teichm\"uller parameterization of disk-type point clouds with landmark constraints}\label{tempo_subsection3}
Let $\mathcal{M}$ be a simply-connected open surface in $\mathbb{R}^3$, and $P$ be a point cloud sampled from $\mathcal{M}$. We aim to compute a parameterization $\varphi: P \to R\subset \mathbb{R}^2$ with suitable boundary conditions, such that $\varphi$ approximates a Teichm\"uller extremal mapping (T-map) from $\mathcal{M}$ to $R$. In this paper, we set the target domain $R$ to be a rectangle.

Our proposed method primarily consists of two steps, namely,
\begin{enumerate}
 \item Conformally parameterizing a disk-type point cloud onto a rectangular planar domain, with the following sub-steps:
    \begin{enumerate}
      \item Mapping the point cloud onto the unit disk by solving the Laplace equation,
      \item Mapping the unit disk onto the unit square by solving the generalized Laplace equation, and
      \item Optimizing the height of the unit square to achieve conformality.
    \end{enumerate}
 \item Computing a Teichm\"uller extremal mapping on the rectangular domain with the prescribed landmark constraints.
\end{enumerate}

\subsubsection{Mapping the point cloud onto the unit disk by solving the Laplace equation}
In the coming three subsections, we describe the procedure of conformally parameterizing a disk-type point cloud $P$ onto a rectangle $\tilde{R}$. In the continuous case, we consider the conformal parameterization of a Riemann surface $\mathcal{M}$. One common surface parameterization method is to compute the harmonic mapping by solving the Laplace equation with certain boundary conditions. As an initialization, the first step of our proposed method is to compute the harmonic mapping $\phi_0:\mathcal{M} \to D$ that satisfies
\begin{equation}
\Delta \phi_0 = 0
\end{equation}
with arc-length parameterization boundary constraints. Note that the mapping $\phi_0$ may not be conformal due to the arc-length parameterization boundary conditions. The conformality distortions are corrected in the following two subsections.

In the discrete case, note that it is hard to approximate the Laplace equation on a surface. Therefore, we perform the computation on a plane with the aid of the following strategy. Recall that in quasi-conformal theory, given an arbitrary diffeomorphism $\varphi: \mathcal{M}\to \varphi(\mathcal{M}) \subset \mathbb{C}$, solving for a harmonic map  $\phi_0:\mathcal{M} \to D$ is equivalent to solving the generalized Laplace equation from $\varphi(\mathcal{M})$ to $D$ with input Beltrami coefficient $\mu(\varphi^{-1})$. In practice, a diffeomorphism from $\mathcal{M}$ to a planar domain is still not easy to solve. Therefore, we locally solve this problem at each point $p$, and set $\varphi$ near $p$ to be the projection from the neighborhood of $p$ to its fitting plane computed using Principal Component Analysis (PCA). Then, for each point, one linear equation is obtained. By combining all points together, one can get a linear system with solution $\bm{\phi}_0$ being an approximation of the desired harmonic map.

\subsubsection{Mapping the unit disk onto the unit square by solving the generalized Laplace equation}

After obtaining the initial map $\phi_0:\mathcal{M} \to D$, we aim to find a mapping from the unit disk $D$ to a rectangle $\tilde{R}$. By scaling, we assume that $\tilde{R}$ is with height $h$ and width $1$. It is noteworthy that in general, $h \neq 1$. Our ultimate goal is to determine a suitable $h$ such that there exists a conformal mapping $ \phi: \mathcal{M} \to \tilde{R}$ with four specified boundary points mapped to the four corners of $\tilde{R}$. The existence of such $h$ and $\phi$ is guaranteed by the Riemann mapping theorem. By the Riemann mapping theorem, a conformal mapping from a simply-connected open bounded surface $\mathcal{M}$ to the unit disk $D$ exists. Moreover, the mapping is unique if three points are fixed on the boundary. Analogously, by restricting four given points to be mapped to four corner points, there exists a positive number $h$ such that the unique conformal mapping exists from $\mathcal{M}$ to the rectangle $\tilde{R}$ with height $h$ and unit width. To explicitly compute $h$ and the desired conformal mapping $\phi$, our strategy is to find a quasi-conformal mapping from the unit disk $D$ to the unit square $R$, and then develop a method to adjust the height of $R$ to an optimal $h$.

In the second step of our proposed method, we compute the desired quasi-conformal mapping from unit disk $D$ to the unit square $R$. Recall that we have obtained the initial map $\phi_0:\mathcal{M} \to D$. By quasi-conformal theory, if two quasi-conformal mappings $f$ and $g$ are associated with the same Beltrami coefficients, then $f \circ g^{-1}$ is conformal. Therefore, the problem boils down to solving for a quasi-conformal mapping from $D$ to $\tilde{R}$ with Beltrami coefficient $\mu(\phi_0^{-1})$. The following proposition guarantees that such a quasi-conformal mapping can be obtained by vertically scaling the mapping $\phi_1 = (u_1,v_1)^T:D\to R$, which solves the generalized Laplace equation (\ref{eqt:secondorder}).

\bigbreak

\begin{prop}
Let $f=u+iv:D \to \tilde{R}$ be a quasi-conformal mapping associated with the Beltrami coefficient $\mu$, and $\tilde{R}$ has unit width and height $h$. Suppose that $f_1 = u_1+iv_1$ is a quasi-conformal mapping from the unit disk $D$ to the unit square $R$. If $u_1$ and $v_1$ satisfy the generalized Laplace equation (\ref{eqt:secondorder}) with the coefficient matrix $A$ computed using $\mu$, with the corresponding boundary condition,
then $f=u_1+ihv_1$.
\end{prop}
\begin{proof}
Let $g$ be the metric tensor on the unit disk $D$ induced by $\mu$. Then $f:(D,g)\to R$ is a conformal map. Therefore, $u,v$ are solutions to the Laplace equation defined on $(D,g)$. More specifically,
\begin{align}
\frac{1}{\det(g)}\partial_i(\sqrt{\det(g)}g^{ij}\partial_j u) = 0,\\
\frac{1}{\det(g)}\partial_i(\sqrt{\det(g)}g^{ij}\partial_j v) = 0.
\end{align}
Moreover, since the Laplace-Beltrami operator is linear, $u$ and $v/h$ are also solutions to the Laplace equation with respect to the metric tensor. This implies that $u$ and $v/h$ are solutions to the PDE
\begin{equation}
\nabla\cdot(\sqrt{\det(g)}(g^{ij})\nabla x)=0.
\end{equation}

On the other hand, by a direct calculation, the coefficient matrix $A$ in the generalized Laplace-Beltrami operator with Beltrami coefficient $\mu$ satisfies
\begin{equation}
A = \sqrt{\det(g)}(g^{ij}).
\end{equation}
More explicitly, we have
\begin{equation}
A = \frac{1}{1-|\mu|^2}\begin{pmatrix}
1+|\mu|^2-2Re(\mu) & -2Im(\mu)\\
-2Im(\mu) & 1+|\mu|^2+2Re(\mu)
\end{pmatrix}.
\end{equation}
Therefore, $u$ and $v/h$ are solutions of the generalized Laplace equation with Beltrami coefficient $\mu$ from the unit disk $D$ to the unit square with specific boundary conditions. And the generalized Laplace equation is an elliptic PDE, whose solution is unique with the given boundary condition. Hence, $u=u_1$, and $v/h = v_1$.
\end{proof}
\bigbreak
With the guarantee by the above proposition, the final result $\phi= (u_1,hv_1)^T \circ \phi_0$ can be obtained by finding a quasi-conformal mapping $\phi_1=(u_1,v_1)^T$ from the unit disk $D$ to the unit square $R$, and scaling its height by a scalar $h$.

In the discrete case, to approximate the Beltrami coefficient $\mu(\phi_0^{-1})$, the diffuse PCBC in Definition \ref{def:pcbc} is adapted. With a given PCBC, the desired quasi-conformal mapping $\bm{\phi}_1=(\bold{u}_1,\bold{v}_1)^T$ from point cloud $P$ to the unit square $R$ is solved by the generalized Laplace equation (\ref{eqt:secondorder-disc}).

\subsubsection{Optimizing the height of the unit square to achieve conformality}
In the final step, it remains to determine the scalar $h$ such that $\phi_2=(u_1,hv_1)^T$ is a quasi-conformal mapping with Beltrami coefficient equals $\mu(\phi_0^{-1})$. To compute such an optimal $h$, it is natural to consider the following energy minimization problem
\begin{equation}\label{eqt:energy-h}
 h = argmin \int_\mathbb{D} |\mu(\phi_2)-\mu(\phi_0^{-1})|^2.
\end{equation}
The energy is zero when $h$ is optimal.

After obtaining the optimal $h$, the desired conformal parameterization $\phi :\mathcal{M} \to \tilde{R}$ is given by
\begin{equation}
\phi = \phi_2 \circ \phi_0.
\end{equation}

In the discrete case, we aim to simplify the computation of the optimal height $h$ of the rectangle. Recall that the point clouds are assumed to be quasi-uniform. Therefore, it is natural to replace the integral in Equation (\ref{eqt:energy-h}) by a summation. Now, it remains to solve the following minimization problem
\begin{equation}
h = argmin \ \ \|\tilde{\bm{\sigma}} - \tilde{\bm{\mu}}\|^2,
\end{equation}
where $\tilde{\bm{\sigma}}$ and $\tilde{\bm{\mu}}$ are respectively the discrete diffuse PCBCs of $\bm{\phi}_2 = (\bold{u}_1,h\bold{v}_1)^T$ and $\bm{\phi}_0^{-1}$. This discretization largely simplifies the computation of the optimal $h$ and the computation for $h$ becomes straightforward.

Our proposed algorithm for computing the conformal parameterizations of disk-type point clouds is summarized in Algorithm \ref{alg:conformal_map}.

\begin{algorithm}[H]
 \KwIn{A point cloud surface $P$ with disk topology, four boundary points $p_i, i = 1, 2 ,3, 4$}
 \KwOut{The point cloud conformal parameterization $\bm{\phi}$ of $P$ to a rectangle $\tilde{R}$ with all $p_i$ mapped to the four corners}
 Solve the Laplace equation with arc-length parameterization boundary constraints and obtain $\bm{\phi}_0: P\to D$\;
 Compute the discrete diffuse PCBC $\tilde{\bm{\mu}}$ of the point cloud mapping $\bm{\phi}_0^{-1}$\;
 Solve the generalized Laplace equation (\ref{eqt:secondorder-disc}) with the associated PCBC $\tilde{\bm{\mu}}$, and obtain a mapping $\bm{\phi}_1=(\bold{u}_1,\bold{v}_1)^T$ from $D$ to the unit square $R$, with four points $p_1,p_2,p_3,p_4$ mapped to the corresponding four corners\;
 Adjust the height of $R$ to be $h$ by minimizing $\|\tilde{\bm{\sigma}} - \tilde{\bm{\mu}}\|^2$, where $\tilde{\bm{\sigma}}$ is the discrete diffuse PCBC of $(\bold{u}_1,h\bold{v}_1)^T$\;
 The conformal parameterization is given by $\bm{\phi} = (\bold{u}_1, h\bold{v}_1)^T \circ \bm{\phi}_0$ from $P$ to a rectangle $\tilde{R}$ with height $h$ and width $1$.
 \caption{Conformal parameterization of disk-type point clouds}
 \label{alg:conformal_map}
\end{algorithm}

\subsubsection{Teichm\"uller parameterization of disk-like point clouds with landmark constraints}
We now aim to develop a method for computing landmark-matching Teichm\"uller parameterization of point clouds with disk topology. By using Algorithm \ref{alg:conformal_map} proposed in the last section, every disk-type point cloud can be conformally mapped to a rectangular domain in $\mathbb{R}^2$. Therefore, in this part, we focus on the Teichm\"uller extremal mapping of planar point clouds with prescribed landmark constraints. In \cite{Lui14}, Lui et al. proposed a method for computing T-maps on triangular meshes. The convergence of the method under the continuous setting was proved in \cite{Lui15}. In this work, we develop an iterative algorithm for PCT-maps on point clouds and study the convergence properties.

Recall that searching for a quasi-conformal mapping is equivalent to searching for a specific Beltrami coefficient in the space of complex-valued functions. In this section, we propose an iterative algorithm to solve for a T-map on point clouds by manipulating the Beltrami coefficients. Specifically, the desired Beltrami coefficients should be with constant norms and satisfy the specific form for the argument part. We start our algorithm by initializing the mapping to be the identity map. Then, in each iteration, four steps as introduced below are performed on the associated discrete diffuse PCBC of the mapping.

In the first step of each iteration, we compute the Beltrami coefficient $\tilde{\bm{\mu}}_n$ of a given point cloud mapping $\bold{f}_n$ using the definition of PCBC in Definition \ref{def:pcbc}.

In the second step, we aim to manipulate the norm of the PCBC so that the new PCBC gets closer to the Beltrami coefficient of a T-map. Specifically, a suitable constant $k_n$ is chosen to be the norm of the new PCBC. Here, we consider taking the average value of the norms of $\tilde{\bm{\mu}}_n$ at the points where the norms are less than 1. In other words, we take
\begin{equation}
k_n = \left\{ \begin{array}{cc}
               \frac{\sum_{p\in S_n } |\tilde{\bm{\mu}}_{n}(p)|}{|S_n|} & \text{ if } |S_n|>0\\
               0 & \text{ if } |S_n|=0,
              \end{array}\right.
\end{equation}
where
\begin{equation}
S_n = \{p\in P : |\tilde{\bm{\mu}}_{n}(p)|<1\}.
\end{equation}
This step guarantees that the new norm is a feasible norm.

In the third step, we update the argument part of the PCBC by using that of $\tilde{\bm{\mu}}_n$. Define $\bm{\nu}_n$ by
\begin{equation}
\bm{\nu}_n(p_i) =
\left\{
\begin{split}
&\frac{\tilde{\bm{\mu}}_{n}(p_i)}{|\tilde{\bm{\mu}}_{n}(p_i)|} &&\text{if\ } \tilde{\bm{\mu}}_{n}(p_i)\neq 0,\\
&1 &&\text{otherwise}.
\end{split}
\right.
\end{equation}
Then, we apply the weighted Laplacian smooth operator and the normalization operator on $\bm{\nu}_n$ and obtain
\begin{equation}
\bm{\tau}_n(p_i) = \left\{
\begin{split}
&\frac{L_i \bm{\nu}_n}{|L_i \bm{\nu}_n|} &&\text{if\ } L_i \bm{\nu}_n \neq 0,\\
&\bm{\nu}_n(p_i) &&\text{otherwise.}
\end{split}
\right.
\end{equation}
Here $L$ is the approximation of the Laplace-Beltrami operator on point clouds as described in Section \ref{sec:glap}, and $L_i \bm{\nu}_n$ stands for $(L\bm{\nu}_n)(p_i)$. With the new norm $k_n$ and the new argument part $\bm{\tau}_n$, we construct a new complex-valued point cloud function
\begin{equation}
\bm{\sigma}_n = k_n\bm{\tau}_n.
\end{equation}

In the last step, a new mapping $\bold{f}_{n+1}$ is computed by solving the hybrid equation (\ref{eqt:hybrid}) with the PCBC $\bm{\sigma}_n$ and the prescribed landmark constraints. It is noteworthy that since $\bm{\sigma}_n$ may not be a feasible PCBC and the landmark constraints are enforced, the resulting mapping $\bold{f}_{n+1}$ may not be exactly associated with $\bm{\sigma}_n$. In this case, this step can be regarded as a projection of $\bm{\sigma}_n$ onto the space of all feasible PCBCs.

By repeating the abovementioned steps, the mapping converges to the landmark-aligned Teichm\"uller parameterization. The entire algorithm is summarized in Algorithm \ref{alg:tmap}.

\begin{algorithm}[ht]
 \KwIn{A planar point cloud $P$, target rectangle $R$ and the landmark constraints}
 \KwOut{The landmark-matching Teichm\"uller parameterization $\bold{f}: P \to R$}
 Initialize the mapping $\bold{f}_0$ to be identity map\;
\While{$\|{F}_{n}-{F}_{n-1}\|_2 \geq \epsilon$}{
Calculate the discrete diffuse PCBC $\tilde{\bm{\mu}}_{n}$ of $\bold{f}_n$\;
Compute $k_n = \left\{ \begin{array}{cc}
               \frac{\sum_{p\in S_n } |\tilde{\bm{\mu}}_{n}(p)|}{|S_n|} & \text{ if } |S_n|>0\\
               0 & \text{ if } |S_n|=0,
              \end{array}\right.$, where $S_n = \{p\in P : |\tilde{\bm{\mu}}_{n}(p)|<1\}$\;
Calculate the argument part of $\tilde{\bm{\mu}}_{n}$ and denote it as $\bm{\nu}_n$\;
Construct a complex valued point cloud function $\bm{\tau}_n$ by
$$
\bm{\tau}_n(p_i) = \left\{
\begin{split}
&\frac{L_i \bm{\nu}_n}{|L_i \bm{\nu}_n| } &&\text{if\ } L_i \bm{\nu}_n \neq 0,\\
&\bm{\nu}_n(p_i) &&\text{otherwise}
\end{split}
\right.
$$
where $L=M_2(0)$ is the approximation of the Laplace-Beltrami operator by our proposed scheme described in Section \ref{sec:glap}\;
Construct a new complex valued point cloud function $\bm{\sigma}_n = k_n\bm{\tau}_n$\;
Compute the point cloud mapping $\bold{f}_{n+1}$ by solving the hybrid equation (\ref{eqt:hybrid}) with a given parameter $\gamma_{n}\in [0,\infty)$, the input BC $\bm{\sigma}_n$ and the prescribed landmark constraints\;
Denote the matrix form of $\bold{f}_{n+1}$ by ${F}_{n+1}$\;
Update $n$ by $n+1$\;
}
Obtain the final mapping $\bold{f} = \bold{f}_{n}$.
\caption{Landmark-matching Teichm\"uller parameterization of planar point clouds}
\label{alg:tmap}
\end{algorithm}

After describing our proposed algorithm for computing landmark-matching Teichm\"uller mappings on point clouds, we prove that the limit function obtained by our algorithm is indeed a PCT-map.

\bigbreak
\begin{prop}\label{prop:tmap conv 1pc}
Let $P$ be a quasi-uniform point cloud sampled from a domain which satisfies the interior cone condition. Assume that the sequence $\bold{f}_n$ in Algorithm \ref{alg:tmap} converges to a point cloud mapping $\bold{f}$, and its discrete diffuse PCBC $\tilde{\bm{\mu}}$ exists. Further assume that $\lim_n\gamma_{n}=0$, where $\gamma_n$ denotes the relative weight parameter for the hybrid equation (\ref{eqt:hybrid}) in the $n$-th iteration. Then $\bold{f}$ is a PCT-map. Moreover, both $\tilde{\bm{\mu}}_{n}$ and $\bm{\sigma}_n$ converge to $\tilde{\bm{\mu}}$.
\end{prop}

\begin{proof}
First, we prove that $\tilde{\bm{\mu}}_n$ converges to $\tilde{\bm{\mu}}$. Now we adapt the notation $D_1$, $D_2$ as defined in the proof of Proposition \ref{prop:pctm2tm}.
By definition, for any point $x\in P$,
\begin{equation}
\begin{split}
&|\tilde{\bm{\mu}}_n(x) - \tilde{\bm{\mu}}(x)|\\
= &\left|  \frac{D_1
\left[ \left(\begin{array}{c}
U_n\\V_n
\end{array}\right) - \left(\begin{array}{c}
U\\V
\end{array}\right)\right] + D_1\left(\begin{array}{c}
U\\V
\end{array}\right)
}{D_2
\left[ \left(\begin{array}{c}
U_n\\V_n
\end{array}\right) - \left(\begin{array}{c}
U\\V
\end{array}\right)\right] + D_2\left(\begin{array}{c}
U\\V
\end{array}\right)}
-
\frac{D_1\left(\begin{array}{c}
U\\V
\end{array}\right)
}{D_2
\left(\begin{array}{c}
U\\V
\end{array}\right)}
\right|.
\end{split}
\end{equation}
By the MLS error estimate in Equation (\ref{eqt:mls_error}), for $j=1,2$,
\begin{equation}
|q_j^T(x)A_x(U_n-U)|\leq \|U_n-U\|_\infty \|q_j^T(x)A_x\|_1 \leq \|U_n-U\|_\infty O(h^{-1}),
\end{equation}
where $h$ is the fill distance of $P$ and hence constant. Since $\bold{f}_n$ converges to $\bold{f}$ and both of them are defined on a finite set $P$, we have $\|U_n-U\|_\infty \to 0$, which implies that $|q_j^T(x)A_x(U_n-U)| \to 0$. A similar result also holds for $V$. Therefore, for $j=1,2$,
\begin{equation} \label{pf:tmap1pc-bcn-conv-bc}
\lim_n\left|D_j
\left[ \left(\begin{array}{c}
U_n\\V_n
\end{array}\right) - \left(\begin{array}{c}
U\\V
\end{array}\right)\right]\right|=0.
\end{equation}
This shows that $\tilde{\bm{\mu}}_n \to \tilde{\bm{\mu}}$.

Then, we prove that $\bm{\sigma}_n$ also converges to $\tilde{\bm{\mu}}$. Recall that in Algorithm \ref{alg:tmap}, $\bold{f}_{n+1}$ solves the hybrid equation (\ref{eqt:hybrid}) with input PCBC $\bm{\sigma}_{n}$. Let
\begin{equation}
E_n = -\gamma_{n} M_2(\bm{\sigma}_{n})\left(\begin{array}{c}
U_{n+1}\\V_{n+1}
\end{array}\right).
\end{equation}
Then
\begin{equation}
M_1(\bm{\sigma}_{n}) \left(\begin{array}{c}
U_{n+1}\\V_{n+1}
\end{array}\right) = E_n.
\end{equation}
By a direct calculation, for any $x=p_j\in P$,
\begin{equation}
\begin{split}
&D_1(x)\left(\begin{array}{c}
U_{n+1}\\V_{n+1}
\end{array}\right)\\
= &q_1^T(x)A_xU_{n+1} + iq_2^T(x)A_xU_{n+1} + iq_1^T(x)A_xV_{n+1} - q_2^T(x)A_xV_{n+1}\\
= &q_1^T(x)A_xU_{n+1} + iq_2^T(x)A_xU_{n+1} - i(\alpha_2(\bm{\sigma}_{n}|_x)q_1^T(x) + \alpha_3(\bm{\sigma}_{n}|_x) q_2^T(x))A_xU_{n+1} \\
&- (\alpha_1(\bm{\sigma}_{n}|_x) q_1^T(x) + \alpha_2(\bm{\sigma}_{n}|_x)q_2^T(x))A_xU_{n+1} + i(E_n)_{j+N} + (E_n)_j\\
=
&\left(2\frac{\bm{\sigma}_{n}(x)-|\bm{\sigma}_{n}(x)|^2}{1-|\bm{\sigma}_{n}(x)|^2}q_1(x)-2i\frac{\bm{\sigma}_{n}(x)+|\bm{\sigma}_{n}(x)|^2}{1-|\bm{\sigma}_{n}(x)|^2}q_2(x)\right)^TA_xU_{n+1}\\
& + i(E_n)_{j+N}+(E_n)_j.
\end{split}
\end{equation}
Similarly, the following equation also holds.
\begin{equation}
\begin{split}
&D_2(x)\left(\begin{array}{c}
U_{n+1}\\V_{n+1}
\end{array}\right)\\
=
& \left(2\frac{1-\overline{\bm{\sigma}_{n}(x)}}{1-|\bm{\sigma}_{n}(x)|^2}q_1(x)
-2i\frac{1 +\overline{\bm{\sigma}_{n}(x)}}{1-|\bm{\sigma}_{n}(x)|^2}q_2(x)\right)^TA_xU_{n+1} + i(E_n)_{j+N} - (E_n)_j.
\end{split}
\end{equation}
Then,
\begin{equation} \label{pf:tmap1pc-sigma-to-mu}
\begin{split}
\bm{\sigma}_n(x)
= \frac{D_1(x)\left(\begin{array}{c}
U_{n+1}\\V_{n+1}
\end{array}\right) - i(E_n)_{j+N} - (E_n)_j
}{D_2(x)\left(\begin{array}{c}
U_{n+1}\\V_{n+1}
\end{array}\right) - i(E_n)_{j+N} + (E_n)_j}.
\end{split}
\end{equation}

Since the point cloud $P$ is fixed and $|\bm{\sigma}_{n}(x)|\leq k_n< 1$ is bounded, the matrix $M_2(\bm{\sigma}_n)$ is bounded. Also, since $\bold{f}_n$ is a point cloud mapping within a fixed rectangle, the vectors $U_n$ and $V_n$ are also bounded. Therefore, $E_n$ converges to the zero vector since $\gamma_{n}$ converges to zero. Moreover, from Equation (\ref{pf:tmap1pc-bcn-conv-bc}) and the assumption that $\tilde{\bm{\mu}}$ is well-defined, we conclude that $|D_2(x)[U_{n+1}^T,V_{n+1}^T]^T|$ is bounded from below when $n$ is large, and the lower bound is a positive number. Thus, the right hand side in Equation (\ref{pf:tmap1pc-sigma-to-mu}) is well-defined, and it converges to $\lim_n\tilde{\bm{\mu}}_{n+1}(x) = \tilde{\bm{\mu}}(x)$. This shows that $\bm{\sigma}_n$ converges to $\tilde{\bm{\mu}}$.

Since $|\bm{\sigma}_n(x)|=k_n <1$ is a constant for any $x\in P$, $k:=\lim_n k_n= |\tilde{\bm{\mu}}(x)|$ is also a constant.

Let $p_i$ be an arbitrary interior point. It remains to prove that $(L_i\tilde{\bm{\mu}})\overline{\tilde{\bm{\mu}}(p_i)} \in \mathbb{R}$. The statement is trivial if $L_i\tilde{\bm{\mu}} =0$. Hence, we only consider the case when $L_i\tilde{\bm{\mu}} \neq 0$ and $k\neq 0$. Let $\tilde{\bm{\mu}}_n=\tilde{\bm{\mu}}_{n,1} + i\tilde{\bm{\mu}}_{n,2}$, $\tilde{\bm{\mu}}=\tilde{\bm{\rho}}_{1} + i\tilde{\bm{\rho}}_{2}$, and $\bm{\sigma}_n = \bm{\sigma}_{n,1} + i \bm{\sigma}_{n,2}$. Without loss of generality, let $\tilde{\bm{\mu}}_{2}(p_i)\neq 0$. Then, when $n$ is large enough,
\begin{equation}
\frac{\tilde{\bm{\rho}}_{1}(p_i)}{\tilde{\bm{\rho}}_{2}(p_i)}
= \lim_n \frac{\bm{\sigma}_{n,1}(p_i)}{\bm{\sigma}_{n,2}(p_i)}
= \lim_n \frac{L_i(\tilde{\bm{\mu}}_{n,1} / |\tilde{\bm{\mu}}_n|)}{L_i(\tilde{\bm{\mu}}_{n,2} / |\tilde{\bm{\mu}}_n|)}
= \frac{L_i(\tilde{\bm{\rho}}_{1}/k)}{L_i(\tilde{\bm{\rho}}_{2}/k)}
= \frac{L_i\tilde{\bm{\rho}}_{1}}{L_i\tilde{\bm{\rho}}_{2}}
\end{equation}
Therefore, $(L_i\tilde{\bm{\rho}}_{1}) \tilde{\bm{\rho}}_{2}(p_i) = (L_i\tilde{\bm{\rho}}_{2}) \tilde{\bm{\rho}}_{1}(p_i)$ and the result follows.
\end{proof}
\bigbreak

As shown above, if Algorithm \ref{alg:tmap} converges, then $\bold{f}_n$ converges to the PCT-map we desired, and $\tilde{\bm{\mu}}_n$, $\bm{\sigma}_n$ converge to its discrete diffuse PCBC. Moreover, the limit PCBC has a constant norm. This satisfies the requirement about the norm of the Beltrami coefficients of PCT-maps in Definition \ref{def:pctm}. Furthermore, the other requirement $(L_i\tilde{\bm{\mu}})\overline{\tilde{\bm{\mu}}(p_i)} \in \mathbb{R}$ in Definition \ref{def:pctm} is also fulfilled. As a remark, since Proposition \ref{prop:pctm2tm} is valid when the fill distance of a sequence of point clouds converges to zero, it follows that our algorithm produces a more accurate result as the density of the point cloud increases. A current limitation of the above proposition is the key assumption on the convergence of Algorithm \ref{alg:tmap}. One of our future works will be proving the convergence of Algorithm \ref{alg:tmap}.

In conclusion, we have proposed an effective algorithm for computing landmark-matching PCT-maps on disk-type point clouds with mathematical guarantees. Experimental results of the computation of PCT-maps are presented in Section \ref{experiment}.


\subsection{Shape analysis of disk-like point clouds via Teichm\"uller metric}\label{tempo_subsection4}

With our proposed Teichm\"uller parameterization algorithm, we can easily compute mappings between feature-endowed point clouds. Furthermore, using the Teichm\"uller metric induced by Teichm\"uller mappings, we can effectively classify different point clouds. In this section, we present our proposed dissimilarity metric based on the Teichm\"uller parameterizations for classifying different point clouds with prescribed landmark constraints.

First, we describe our proposed registration algorithm in the continuous setting. Suppose $\Omega_1$ and $\Omega_2$ are two disk-type Riemann surfaces, and the landmark correspondences are given by
\begin{equation}
 p_i \leftrightarrow q_i
\end{equation}
for $i = 1,2,\cdots,n$. We aim to compute the Teichm\"uller mapping $f: \Omega_1 \to \Omega_2$ such that $f(p_i) = q_i$ for all $i=1,2,\cdots,n$.

We start by computing the conformal parameterizations $g_1: \Omega_1 \to \tilde{R_1}$, $g_2: \Omega_2 \to \tilde{R_2}$. Then, we compute the Teichm\"uller parameterization $h: \tilde{R_1} \to \tilde{R_2}$ between the two planar domains, with the following landmark correspondences
\begin{equation}
 h(g_1(p_i)) = g_2(q_i)
\end{equation}
for $i = 1,2,\cdots,n$. Finally, a mapping $f: \Omega_1 \to \Omega_2$ can be obtained by a composition of the abovementioned mappings. More specifically, we define
\begin{equation}
f=g_2^{-1} \circ h \circ g_1.
\end{equation}
Obviously, $f$ satisfies
\begin{equation}
 f(p_i) = q_i.
\end{equation}
Moreover, $f$ is a Teichm\"uller mapping. In other words, the Beltrami coefficient associated with $f$ is with constant norm, and its argument part is the conjugate of the argument part of a nonzero holomorphic mapping.
This phenomenon can be explained by the composition property of quasi-conformal mappings. Note that since $h$ is a Teichm\"uller mapping, the Beltrami coefficient $\mu_h$ is with constant norm. Also, since $g_1,g_2$ are conformal mappings, the Beltrami coefficient $\mu_{g_1},\mu_{g_2}$ is zero. By Theorem \ref{thm:composition},
\begin{equation}
\mu_f = \mu_{g_2^{-1} \circ h \circ g_1}
= (\mu_h \circ g_1) \frac{\overline{(\partial_z g_1)}}{\partial_z g_1}
= k \frac{\overline{(\partial_z g_1)^2 (\varphi \circ g_1)}}{|(\partial_z g_1)^2 (\varphi \circ g_1)|}
\end{equation}
where $k=|\mu_h|$ is a constant, and $\varphi$ is holomorphic. Moreover, since $g_1$ is conformal, $(\partial_z g_1)^2 (\varphi \circ g_1)$ is holomorphic. Therefore, the registration mapping $f$ is a landmark-matching Teichm\"uller mapping.

In the discrete case, the computations of the conformal and the Teichm\"uller parameterizations of point clouds can be respectively achieved by Algorithm \ref{alg:conformal_map} and Algorithm \ref{alg:tmap}. The entire registration algorithm is summarized in Algorithm \ref{alg:registration}.

\begin{algorithm}[H] \label{alg:registration}
 \KwIn{Two disk-type point clouds $P_1$ and $P_2$, with landmark correspondences $p_i \leftrightarrow q_i$, $p_i\in P_1$, $q_i\in P_2$, $i=1,2,...,n$}
 \KwOut{The Teichm\"uller mapping $\bold{f}:P_1 \to P_2$ with $\bold{f}(p_i)=q_i$ for all $i$}
 Using Algorithm \ref{alg:conformal_map}, compute the conformal parameterizations $\bold{g}_t$ from $P_t$ to a rectangular domain $\tilde{R_t}$, where $t=1,2$\;
 Using Algorithm \ref{alg:tmap}, compute the Teichm\"uller parameterization $\bold{h}$ from $\tilde{R_1}$ to the domain $\tilde{R_2}$, with the landmark constraints $\bold{h}(\bold{g}_1(p_i)) = \bold{g}_2(q_i)$ for all $i$\;
 Obtain the registration mapping $\bold{f}=\bold{g}_2^{-1} \circ \bold{h} \circ \bold{g}_1$.
 \caption{Landmark-matching Teichm\"uller registration between two point clouds with disk topology}
\end{algorithm}

Furthermore, with the aid of our landmark-constrained Teichm\"uller registration scheme, a metric shape space can be built for point cloud classification and shape analysis. More explicitly, a metric called the Teichm\"uller metric is naturally induced by Teichm\"uller mappings. Mathematically, the Teichm\"uller metric is defined as follows.
\begin{defn}[Teichm\"uller metric]
For every $i$, let $M_i$ be a Riemann surface with feature landmarks $\{p_{ik}\}_{k=1}^n$, $\Omega$ be the template surface with the same topology as $M_i$ with corresponding landmarks $\{q_{k}\}_{k=1}^n$. Suppose each surface $M_i$ is parameterized onto $\Omega$ by a landmark-matching quasi-conformal homeomorphism $f_i: M_i \to \Omega$. The \emph{Teichm\"uller metric} between $(f_i,M_i)$ and $(f_j,M_j)$ is defined as
\begin{equation}\label{eqt:t-metric}
 d_T((f_i, M_i), (f_j, M_j)) = \inf_{\varphi} \frac{1}{2} \log K(\varphi),
\end{equation}
where $\varphi:M_i \to M_j$ varies over all quasi-conformal mappings with $\{p_{ik}\}_{k=1}^n$ corresponds to $\{p_{jk}\}_{k=1}^n$, which is homotopic to $f_j^{-1} \circ f_i$, and $K$ is the maximal quasi-conformal dilation.\\
\end{defn}

According to the composition property of quasi-conformal mappings, neither the left nor the right composition with a conformal mapping affect the maximal quasi-conformal dilation $K$. Therefore, if we consider any $\varphi = g_j^{-1}\circ h \circ g_i$, where $g_l$ is a conformal mapping from $\mathcal{M}_l$ to the rectangle $\tilde{R_l}$, it follows that
\begin{equation}
K(\varphi) = K(h).
\end{equation}
Thus, the Teichm\"uller metric is uniquely determined by the maximal quasi-conformal dilation of the extremal mapping between the two rectangular domains.

By Theorem \ref{landmarkteichmullerdisk}, with a suitable boundary condition, the Teichm\"uller mapping between two unit disks is extremal and unique. A similar result also holds for the Teichm\"uller mapping between two rectangular domains. Therefore, the infimum in Equation (\ref{eqt:t-metric}) is achieved by the unique Teichm\"uller mapping between $\tilde{R_i}$ and $\tilde{R_j}$. To be more specific, we have
\begin{equation}
 d_T((f_i, M_i), (f_j, M_j)) = \frac{1}{2} \log K(h)
\end{equation}
where $h:\tilde{R_i} \to \tilde{R_j}$ is a landmark-constrained Teichm\"uller mapping with $g_j^{-1}\circ h\circ g_i$ homotopic to $f_j^{-1} \circ f_i$. Hence, the difference of two feature-endowed point clouds can be evaluated in terms of the quasi-conformal dilation of the Teichm\"uller mapping between them. This motivates us the following method for building up a dissimilarity metric of feature-endowed point clouds.

Let $\{P_t\}$ be a collection of feature-endowed point clouds. To perform an accurate classification on $\{P_t\}$, we take every pair of point cloud surface $P_i, P_j$ and compute the landmark-matching Teichm\"uller mappings $\bold{f}_{ij}: P_i \to P_j$ for all $i,j$. Then, we evaluate the associated discrete diffuse PCBCs $\bm{\mu}_{ij}$ and denote the Teichm\"uller distance of $P_i$ and $P_j$ by
\begin{equation}
d_{ij} = \frac{1}{2} \log \frac{1+\|\bm{\mu}_{ij}\|_{\infty}}{1-\|\bm{\mu}_{ij}\|_{\infty}}.
\end{equation}
The value $d_{ij}$ records the difference between the point clouds $P_i$ and $P_j$. Specifically, if $P_i$ and $P_j$ are with similar shapes, then $d_{ij}$ is small. By analyzing the distance matrix $(d_{ij})$ formed, we can effectively classify different feature-endowed point clouds. For instance, the multidimensional scaling (MDS) method can be applied for the classification here.

In the discrete case, numerical errors are unavoidable. To alleviate the effect of very minor outlying Beltrami coefficients, we replace $\|\mu\|_{\infty}$ by the average of $|\mu|$ in the calculation of the Teichm\"uller distances. The validity of this modification is justified by the negligible variance of $|\mu|$. It is also noteworthy that in practice, given two point clouds $P_i$ and $P_j$, the Teichm\"uller distance $d_{ij}$ (induced by the Teichm\"uller mapping $f_{ij}: P_i \to P_j$) may be slightly different from $d_{ji}$ (induced by the Teichm\"uller mapping $f_{ji}: P_j \to P_i$) because of the numerical errors and the stopping criterion of the algorithm. Therefore, for symmetry, we replace the distance matrix $D$ by $(D+D^T)/2$. Our algorithm for the dissimilarity metric is summarized in Algorithm \ref{alg:metric}. Experimental results are provided in Section \ref{experiment} to illustrate the effectiveness of our proposed method.

\begin{algorithm}[H] \label{alg:metric}
 \KwIn{A set of point clouds $\{P_t\}_{t=1}^M$ with landmark correspondences}
 \KwOut{The distance matrix $D = (d_{ij})$, $i,j = 1,2,\cdots, M$}
 Using Algorithm \ref{alg:registration}, compute the Teichm\"uller mapping $\bold{f}_{ij}$ from each point cloud $P_i$ to another point cloud $P_j$ with the landmark correspondences\;
 Compute the Teichm\"uller distance $d_{ij}$ between $P_i$ and $P_j$ by $d_{ij} = \frac{1}{2}\log (K(\bold{f}_{ij}))$ and form the matrix $D = (d_{ij})$\;
 Update $D$ by $(D+D^T)/2$\;
 \caption{Building a distance matrix using the Teichm\"uller metric}
\end{algorithm}

\section{Experimental results} \label{experiment}
In this section, we demonstrate the effectiveness of our proposed TEMPO method by various examples. Our proposed algorithms are implemented in MATLAB. The point clouds are adapted from the AIM@SHAPE shape repository \cite{aim@shape} and the TF3DM repository \cite{tf3dm}. The Laboratory for Computational Longitudinal Neuroimaging (LCLN) shape database \cite{lcln}, the human face dataset \cite{Bronstein07}, the Spacetime Face Data \cite{Zhang04}, and some additional facial point clouds sampled using Kinect are used in the first shape analysis experiment. The database in \cite{Beumier00} is used in the second shape analysis experiment. The sparse linear systems are solved using the built-in backslash operator (\textbackslash) in MATLAB. All experiments are performed on a PC with an Intel(R) Core(TM) i7-4770 CPU @3.40 GHz processor and 16.00 GB RAM.

\subsection{The performance of our proposed approximation schemes}
In this subsection, we evaluate the performance of our proposed approximation schemes for computing quasi-conformal (including conformal) mappings on point clouds with disk topology.

We first compare our proposed scheme with the local mesh method \cite{Lai13} and the moving least square method with special weight \cite{Liang12,Liang13} for computing conformal mappings. In each experiment, we first generate a random point cloud on a planar rectangular domain. Then, we transform the planar point cloud to a 3D point cloud using a conformal mapping with an explicit formula. On the transformed point cloud, we approximate the Laplace-Beltrami operator using the aforementioned schemes. Using the approximated Laplace-Beltrami operator, we solve the Laplace equation to map the point cloud back onto the rectangular domain. The resulting position errors show the accuracies of the approximation schemes. An example of such a conformal mapping with an explicit formula is given in Figure \ref{fig:laplacian}. Table \ref{table:laplacian} lists the statistics of several experiments with this conformal map. It can be observed that our proposed combined scheme provides better approximations for the conformal mappings on point clouds with disk topology.

\begin{figure}[t]
\centering
\includegraphics[width=0.45\textwidth]{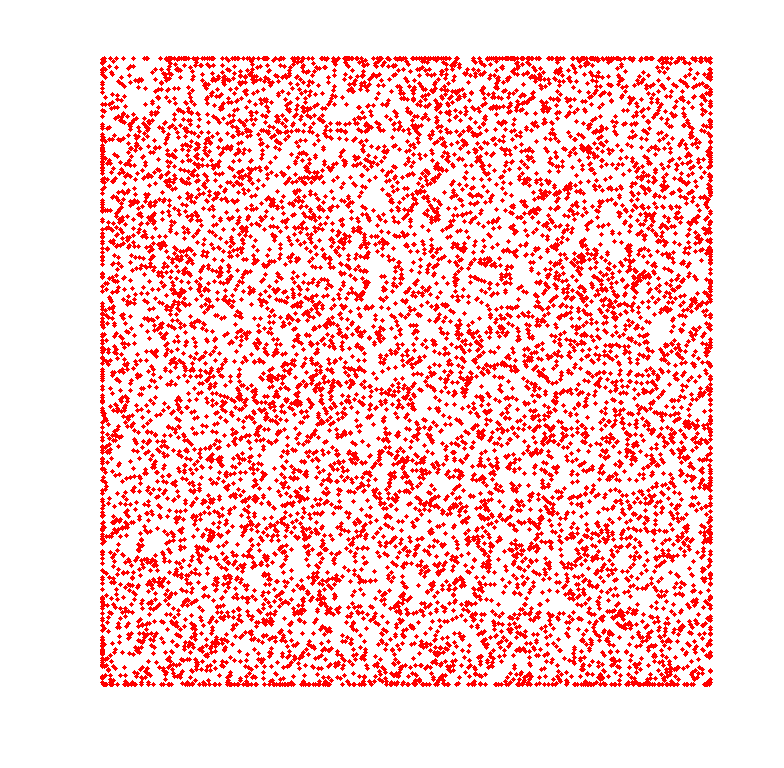}
\includegraphics[width=0.45\textwidth]{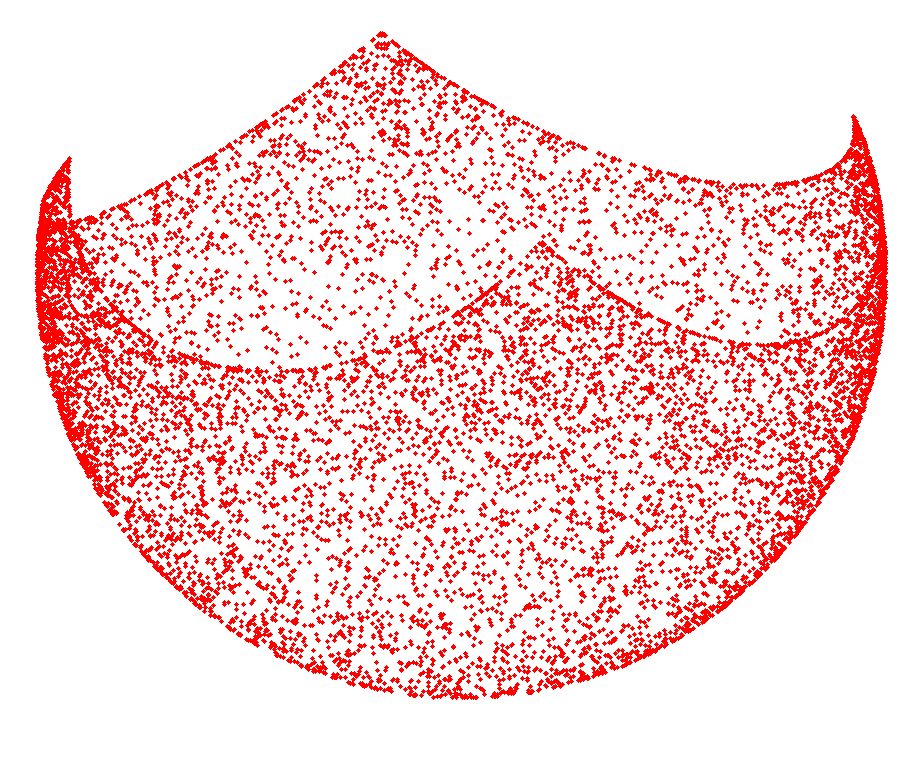}

\caption{An example (the stereographic projection) used in evaluating different numerical schemes for computing conformal mappings on point clouds. The conformal mapping is given by $\displaystyle f(x,y) = \left(\frac{2x}{x^2+y^2+1}, \frac{2y}{x^2+y^2+1},\frac{x^2+y^2-1}{x^2+y^2+1}\right)$.}
\label{fig:laplacian}
\end{figure}

\begin{table}[t]
\begin{tabular}{|c|C{22mm}|C{22mm}|C{22mm}|}
\multicolumn{4}{c}{Example 1}\\ \hline 
Method & maximum position error & average 1-norm error & average 2-norm error\\
\hline
Local Mesh \cite{Lai13}& 0.048 & 0.017 & 0.00031\\
MLS with special weight \cite{Liang12,Liang13}& 1.821 & 0.275 & 0.13026\\
Our proposed method & 0.048 & 0.017 & 0.00031\\
\hline
\end{tabular}
\medskip\\
\begin{tabular}{|c|C{22mm}|C{22mm}|C{22mm}|}
\multicolumn{4}{c}{Example 2}\\ \hline 
Method & maximum position error & average 1-norm error & average 2-norm error\\
\hline
Local Mesh \cite{Lai13}& 0.022 & 0.0043 & 0.000020\\
MLS with special weight \cite{Liang12,Liang13}& 0.853 & 0.2529 & 0.071753\\
Our proposed method & 0.011 & 0.0035 & 0.000014\\
\hline
\end{tabular}
\medskip\\
\begin{tabular}{|c|C{22mm}|C{22mm}|C{22mm}|}
\multicolumn{4}{c}{Example 3}\\ \hline 
Method & maximum position error & average 1-norm error & average 2-norm error\\
\hline
Local Mesh \cite{Lai13}& 0.011 & 0.0041 & 0.000017\\
MLS with special weight \cite{Liang12,Liang13}& 0.182 & 0.0659 & 0.004212\\
Our proposed method & 0.010 & 0.0035 & 0.000013\\
\hline
\end{tabular}
\medskip\\
\begin{tabular}{|c|C{22mm}|C{22mm}|C{22mm}|}
\multicolumn{4}{c}{Example 4 (Noisy)}\\ \hline 
Method & maximum position error & average 1-norm error & average 2-norm error\\
\hline
Local Mesh \cite{Lai13}& 0.050 & 0.0150 & 0.000202\\
MLS with special weight \cite{Liang12,Liang13}& 0.379 & 0.0881 & 0.009463\\
Our proposed method & 0.026 & 0.0083 & 0.000063\\
\hline
\end{tabular}
\medskip\\
\begin{tabular}{|c|C{22mm}|C{22mm}|C{22mm}|}
\multicolumn{4}{c}{Example 5 (Noisy)}\\ \hline 
Method & maximum position error & average 1-norm error & average 2-norm error\\
\hline
Local Mesh \cite{Lai13}& 0.090 & 0.0128 & 0.000238\\
MLS with special weight \cite{Liang12,Liang13}& 0.501 & 0.0574 & 0.004941\\
Our proposed method & 0.049 & 0.0097 & 0.000105\\
\hline
\end{tabular}
\caption{The performances of three numerical schemes for computing conformal maps on point clouds. The stereographic projection with a randomly generated point cloud is considered in each experiment.}
\label{table:laplacian}
\end{table}

Then, we compare the mentioned approximation schemes for computing quasi-conformal mappings with prescribed Beltrami coefficients. This time, in each experiment, we transform a randomly generated point cloud using a quasi-conformal mapping with an explicit formula. An example is given in Figure \ref{fig:generalized_laplacian}. On the transformed point cloud, we approximate the generalized Laplacian using different schemes and solve the generalized Laplace equation to map the point cloud back onto the rectangular domain. Table \ref{table:generalized_laplacian} lists the statistics of several experiments. Again, our combined approximation scheme produces results with higher accuracy.

\begin{figure}[t]
\centering
\includegraphics[width=0.45\textwidth]{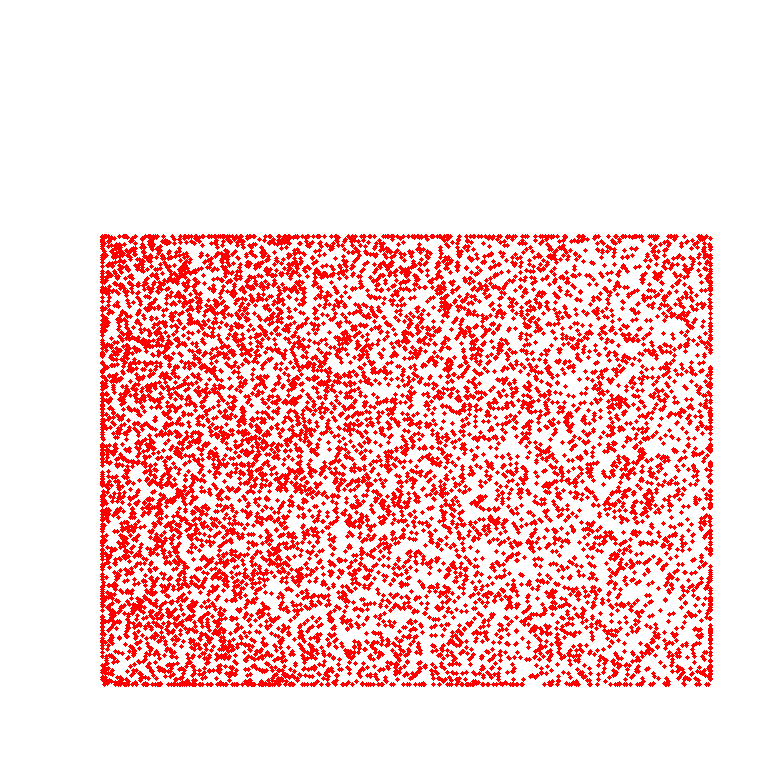}
\includegraphics[width=0.45\textwidth]{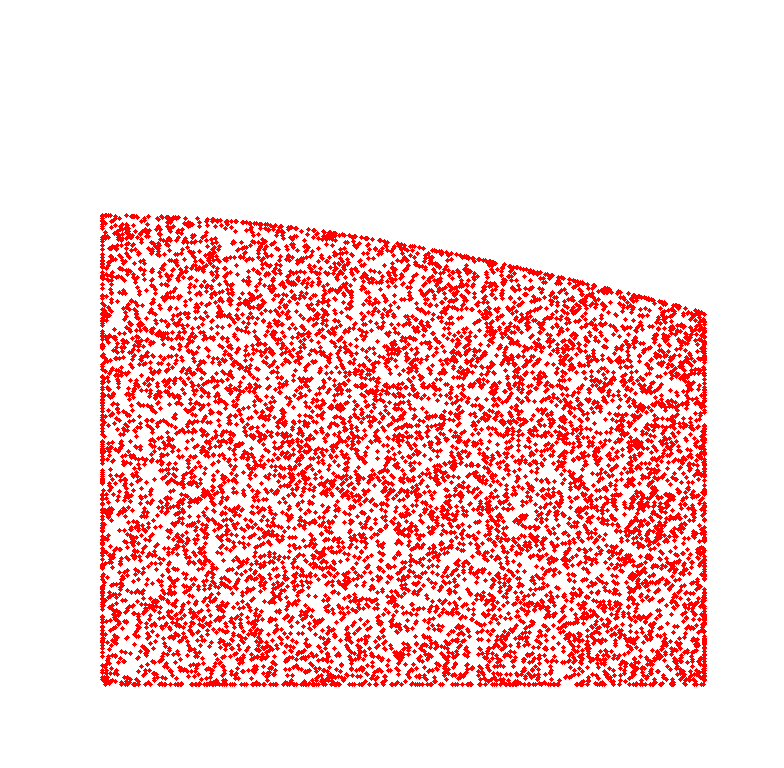}
\caption{An example used in evaluating different schemes for computing quasi-conformal mappings on point clouds. The quasi-conformal mapping is given by $\displaystyle f(x,y)= \left(\log(x+1),\arcsin\left(\frac{y}{2+(\log(x+1))^2}\right)\right)$.}
\label{fig:generalized_laplacian}
\end{figure}

\begin{table}[t]
\begin{tabular}{|c|C{22mm}|C{22mm}|C{22mm}|}
\multicolumn{4}{c}{Example 1}\\ \hline 
Method & maximum position error & average 1-norm error & average 2-norm error\\
\hline
Local Mesh \cite{Lai13}& 0.0756 & 0.01150 & 0.00014248\\
MLS with special weight \cite{Liang12,Liang13}& 0.0940 & 0.02524 & 0.00089952\\
Our proposed method & 0.0252 & 0.00836 & 0.00006811\\
\hline
\end{tabular}
\medskip\\
\begin{tabular}{|c|C{22mm}|C{22mm}|C{22mm}|}
\multicolumn{4}{c}{Example 2}\\ \hline 
Method & maximum position error & average 1-norm error & average 2-norm error\\
\hline
Local Mesh \cite{Lai13}& 0.0467 & 0.01011 & 0.00010035\\
MLS with special weight \cite{Liang12,Liang13}& 0.0212 & 0.01219 & 0.00012290\\
Our proposed method & 0.0183 & 0.00750 & 0.00005012\\
\hline
\end{tabular}
\medskip\\
\begin{tabular}{|c|C{22mm}|C{22mm}|C{22mm}|}
\multicolumn{4}{c}{Example 3}\\ \hline 
Method & maximum position error & average 1-norm error & average 2-norm error\\
\hline
Local Mesh \cite{Lai13}& 0.0172 & 0.00813 & 0.00005617\\
MLS with special weight \cite{Liang12,Liang13}& 0.6373 & 0.11550 & 0.02536965\\
Our proposed method & 0.0172 & 0.00796 & 0.00005446\\
\hline
\end{tabular}
\medskip\\
\begin{tabular}{|c|C{22mm}|C{22mm}|C{22mm}|}
\multicolumn{4}{c}{Example 4 (Noisy)}\\ \hline 
Method & maximum position error & average 1-norm error & average 2-norm error\\
\hline
Local Mesh \cite{Lai13}& 0.0181 & 0.00685 & 0.00004237\\
MLS with special weight \cite{Liang12,Liang13}& 0.0565 & 0.01594 & 0.00034875\\
Our proposed method & 0.0177 & 0.00652 & 0.00003953\\
\hline
\end{tabular}
\medskip\\
\begin{tabular}{|c|C{22mm}|C{22mm}|C{22mm}|}
\multicolumn{4}{c}{Example 5 (Noisy)}\\ \hline 
Method & maximum position error & average 1-norm error & average 2-norm error\\
\hline
Local Mesh \cite{Lai13}& 0.0170 & 0.00769 & 0.00005009\\
MLS with special weight \cite{Liang12,Liang13}& 0.0429 & 0.01948 & 0.00034534\\
Our proposed method & 0.0170 & 0.00762 & 0.00004964\\
\hline
\end{tabular}
\caption{The performances of different schemes for computing quasi-conformal mappings on point clouds.}
\label{table:generalized_laplacian}
\end{table}

\subsection{Landmark constrained Teichm\"uller parameterizations}

\begin{table}[t]
\centering
\begin{tabular}{|c|c|c|c|c|c|c|c|}
\multicolumn{8}{c}{Result with decreasing $\gamma$}\\ \hline
$\gamma$ 	   & 0 	 & 1 	  & 50     & 100    & 500 & 1000     & $\infty$ \\ \hline 
Var($|\bm{\mu}|$) & NaN & 0.0018  & 0.0024 & 0.0020  & 0.0105 & 0.0085  & 0.0050   \\ \hline 
\end{tabular}
\medskip\\
\begin{tabular}{|c|c|c|c|c|c|c|c|}
\multicolumn{8}{c}{Result with constant $\gamma$}\\ \hline
$\gamma$ 	   & 0 	 & 1 	 & 50     & 100	 & 500 & 1000   & $\infty$ \\ \hline 
Var($|\bm{\mu}|$) & NaN & 0.0018 & 0.0022 & 0.0017 & 0.0031 & 0.0022  & 0.0050   \\ \hline 
\end{tabular}
\caption{The results obtained by different choices of $\gamma$ in Algorithm \ref{alg:tmap}. The row of $\gamma$ refers to the initial values.}
\label{tab:gamma}
\end{table}

\begin{figure}[t]
\centering
\includegraphics[width=0.36\textwidth]{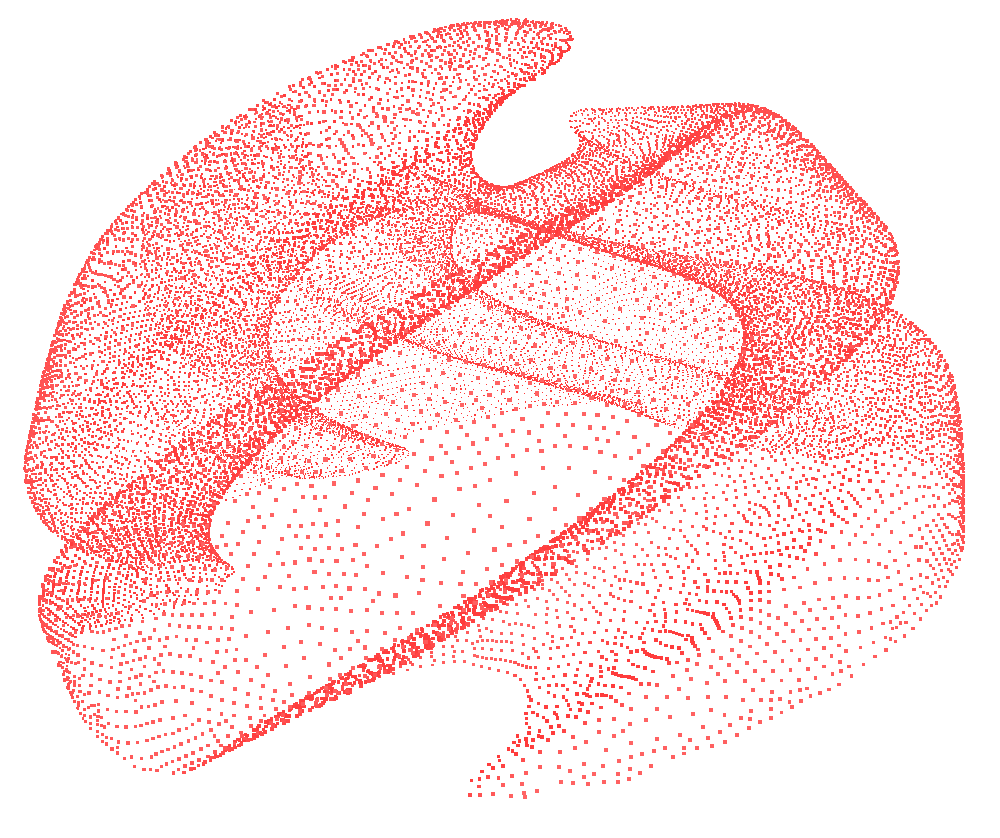}
\includegraphics[width=0.31\textwidth]{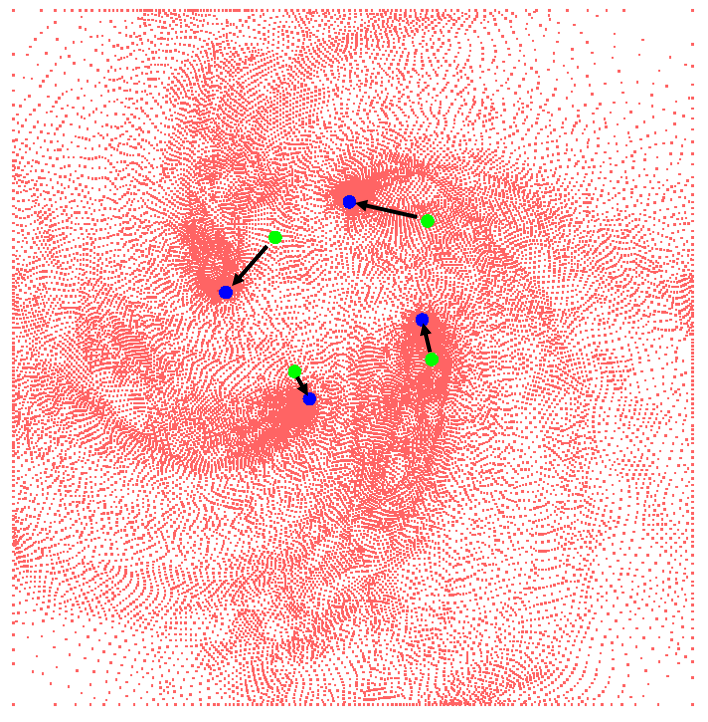}
\includegraphics[width=0.31\textwidth]{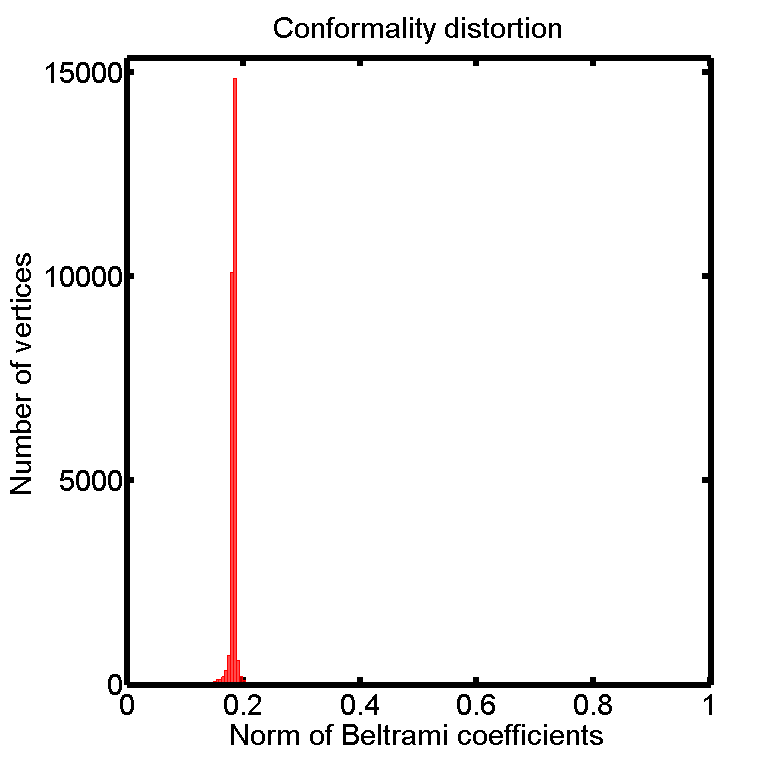}
\caption{Teichm\"uller parameterization of a spiral point cloud. Left: The input point cloud. Middle: The parameterization result. The landmark constraints are represented by the green and blue points. The green points represent the original locations and the blue points are the target locations. Right: The histogram of the norm of the Beltrami coefficients. It can be observed that each histogram highly concentrates at one value. This shows that our parameterization achieves uniform conformality distortion.}
\label{fig:shaped_ball}
\end{figure}

\begin{figure}[t]
\centering
\includegraphics[width=0.33\textwidth]{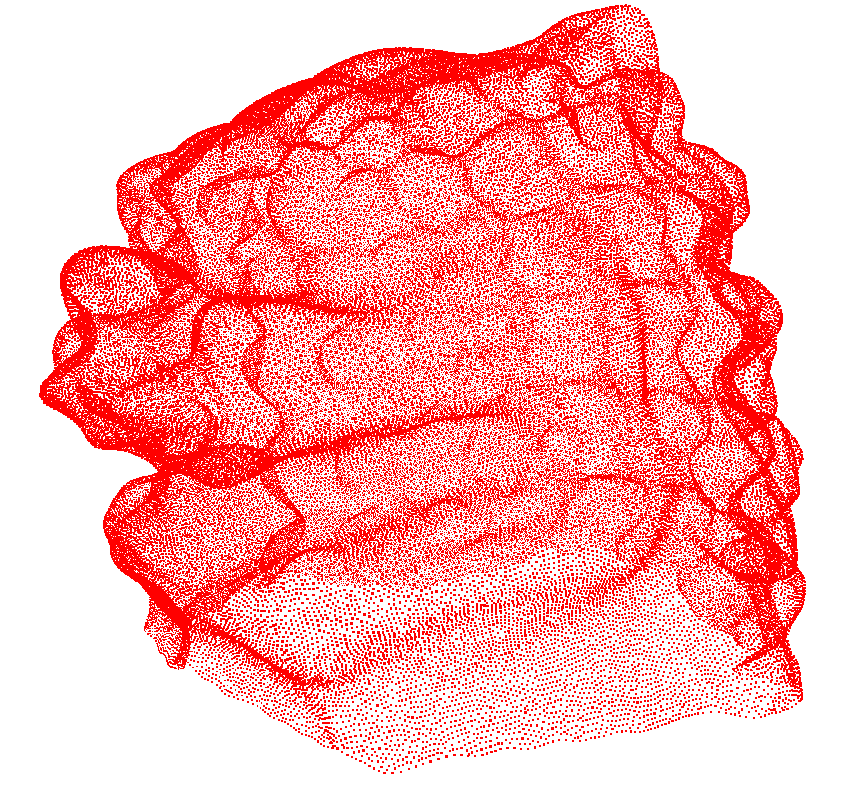}
\includegraphics[width=0.32\textwidth]{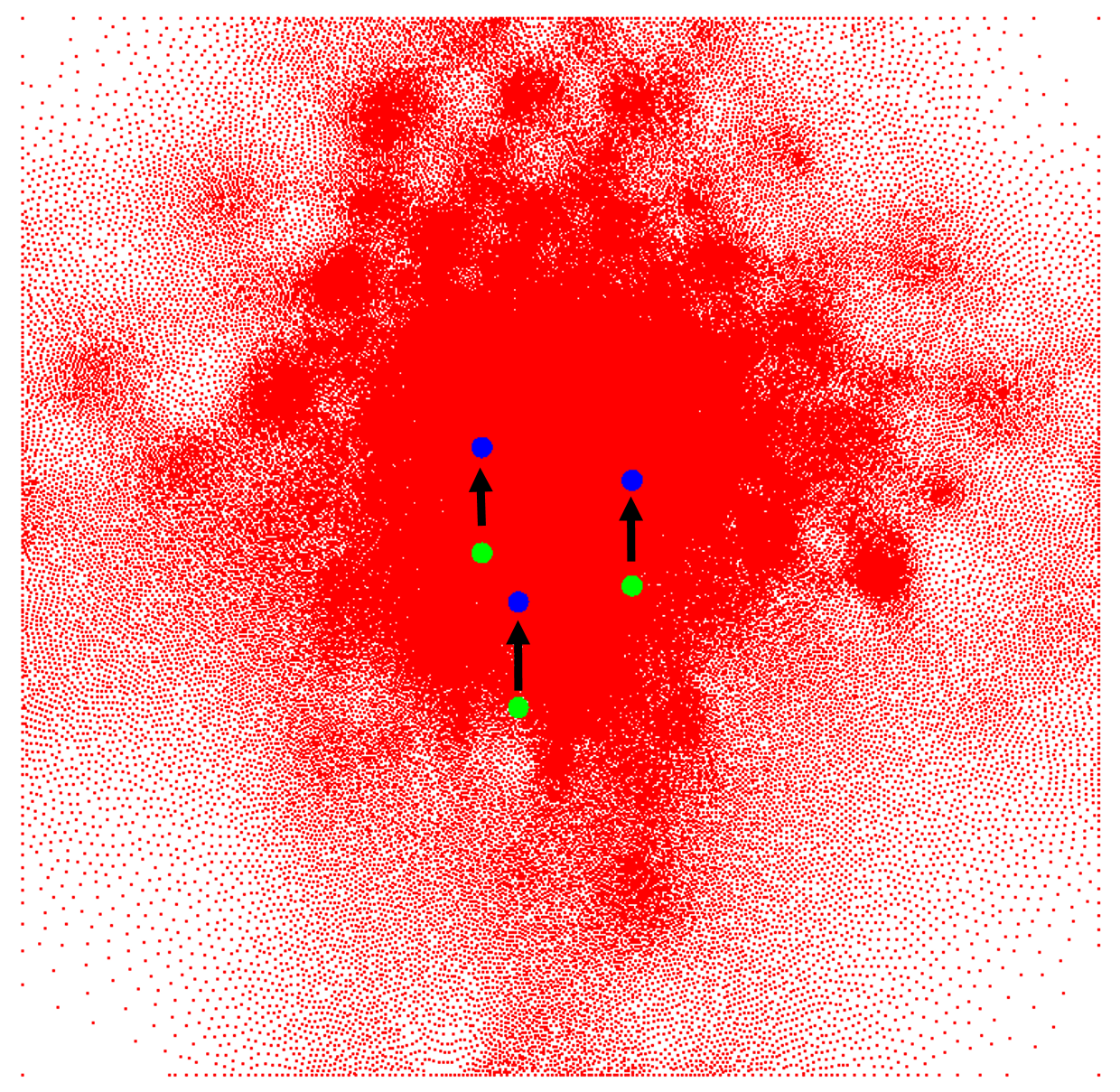}
\includegraphics[width=0.32\textwidth]{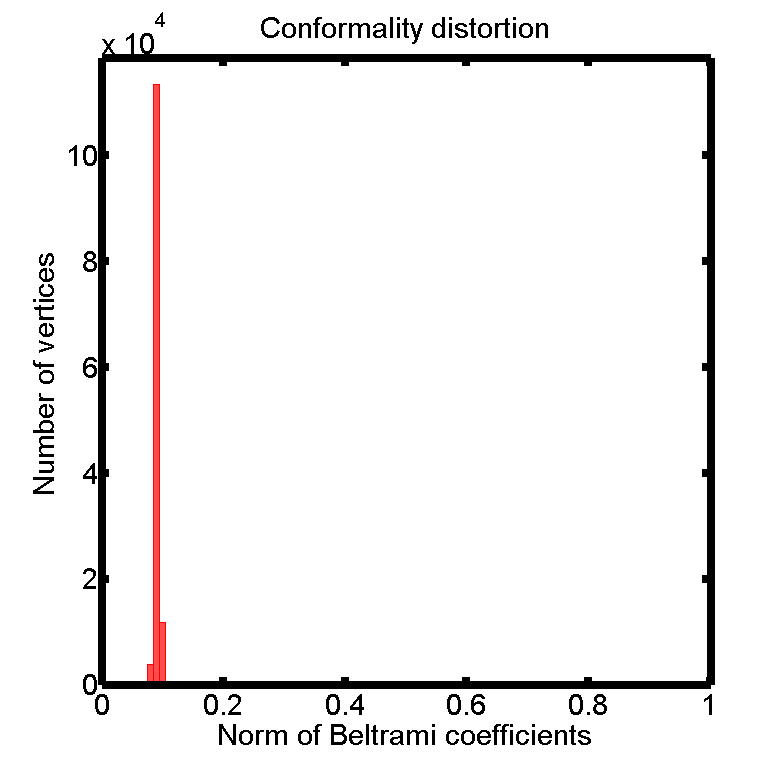}
\caption{Teichm\"uller parameterization of a lion point cloud. Left: The input point cloud. Middle: The parameterization result. The landmark constraints are represented by the green and blue points. The green points represent the original locations and the blue points are the target locations. Right: The histogram of the norm of the Beltrami coefficients. }
\label{fig:lion}
\end{figure}

\begin{figure}[t]
\centering
\includegraphics[width=0.25\textwidth]{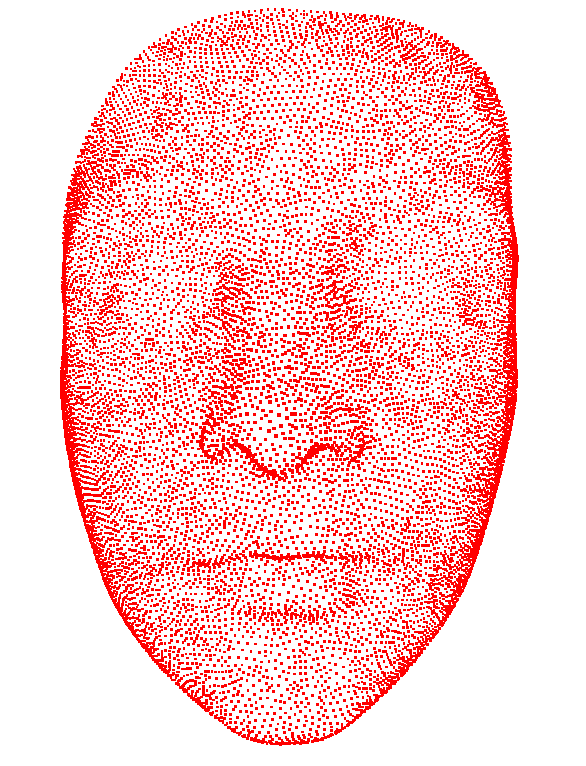}
\includegraphics[width=0.35\textwidth]{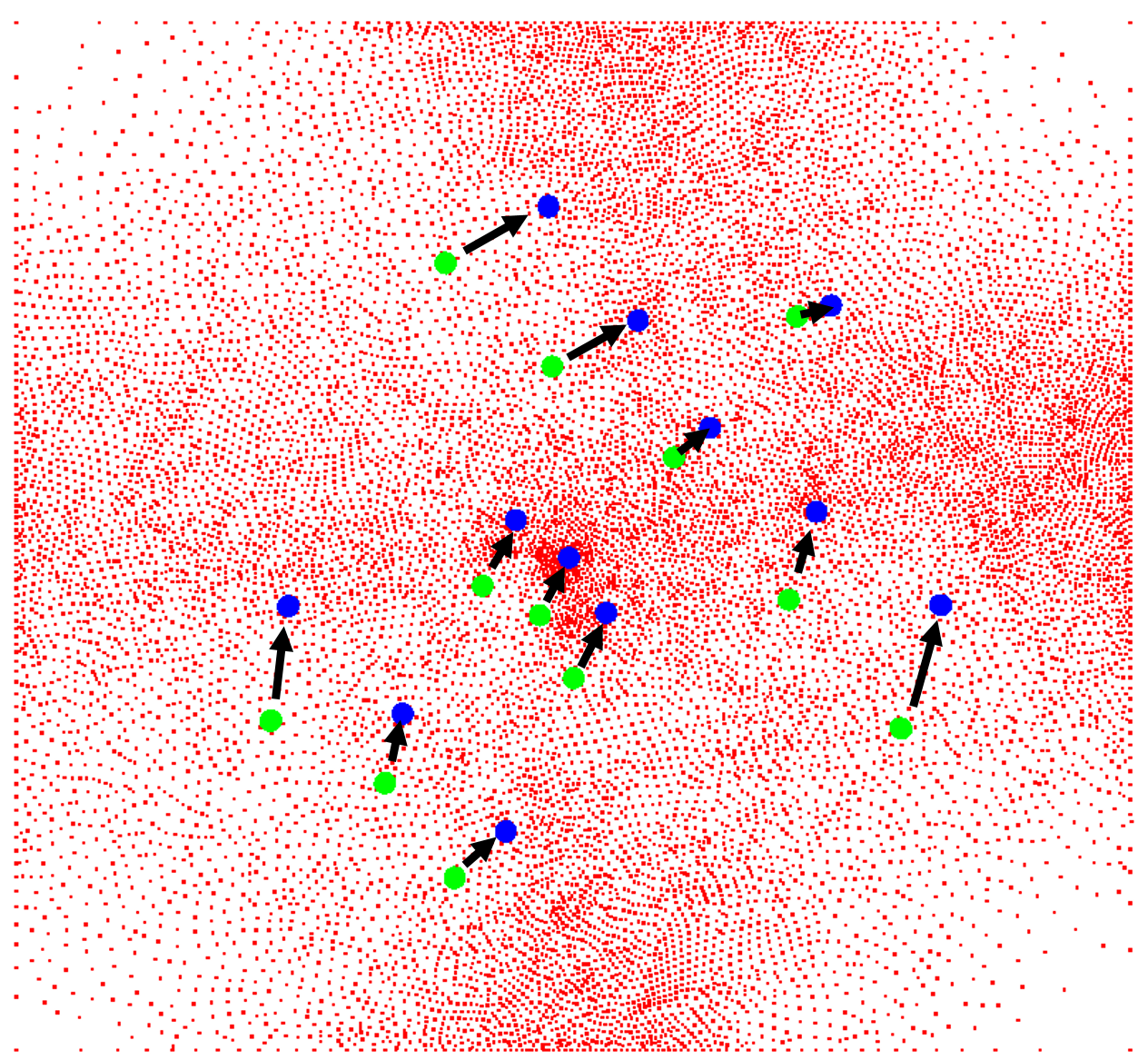}
\includegraphics[width=0.35\textwidth]{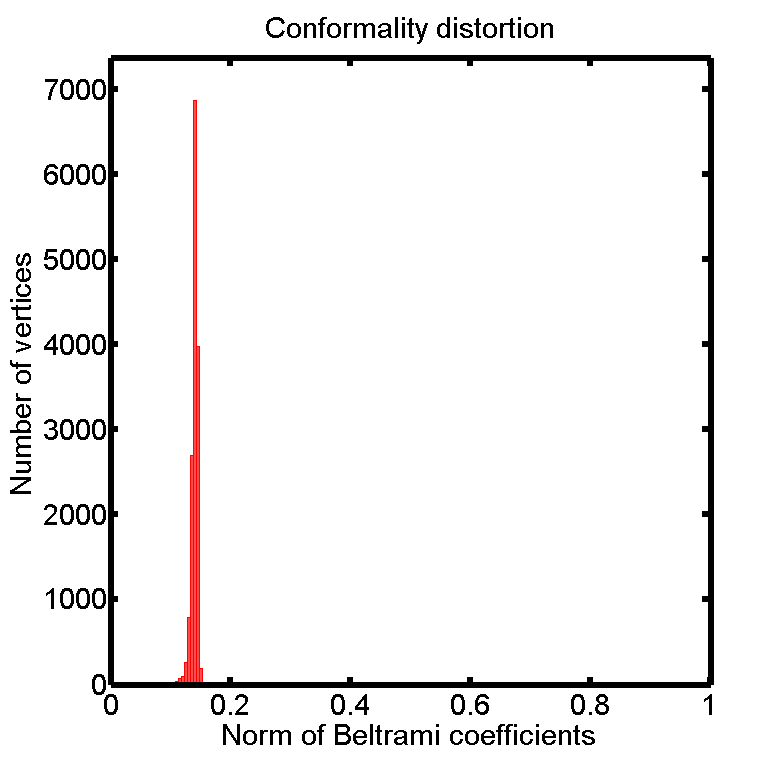}
\caption{Teichm\"uller parameterization of Lucy point cloud. Left: The input point cloud. Middle: The parameterization result. The landmark constraints are represented by the green and blue points. The green points represent the original locations and the blue points are the target locations. Right: The histogram of the norm of the Beltrami coefficients. }
\label{fig:lucy}
\end{figure}

\begin{figure}[t]
\centering
\includegraphics[width=0.25\textwidth]{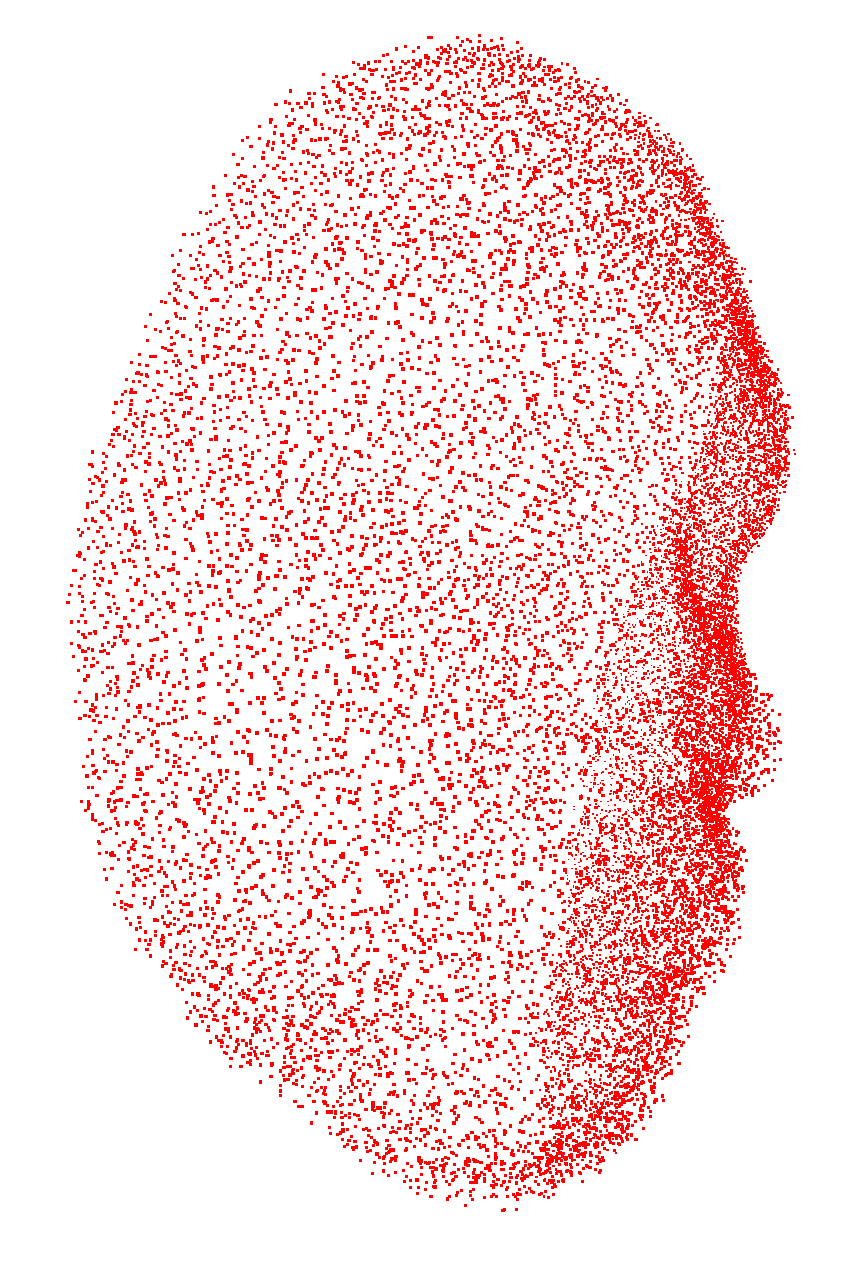}
\includegraphics[width=0.35\textwidth]{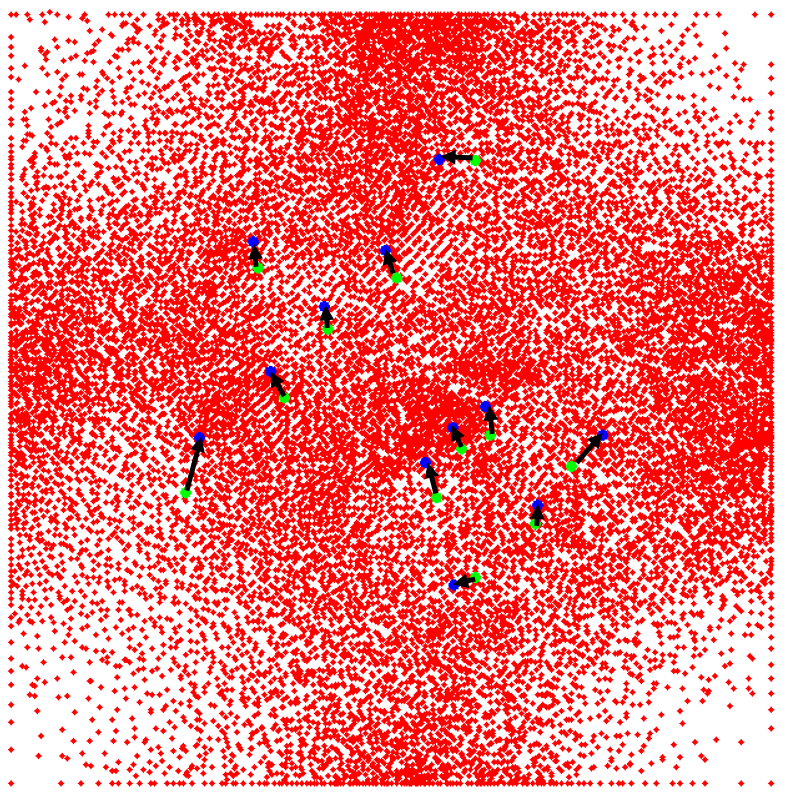}
\includegraphics[width=0.35\textwidth]{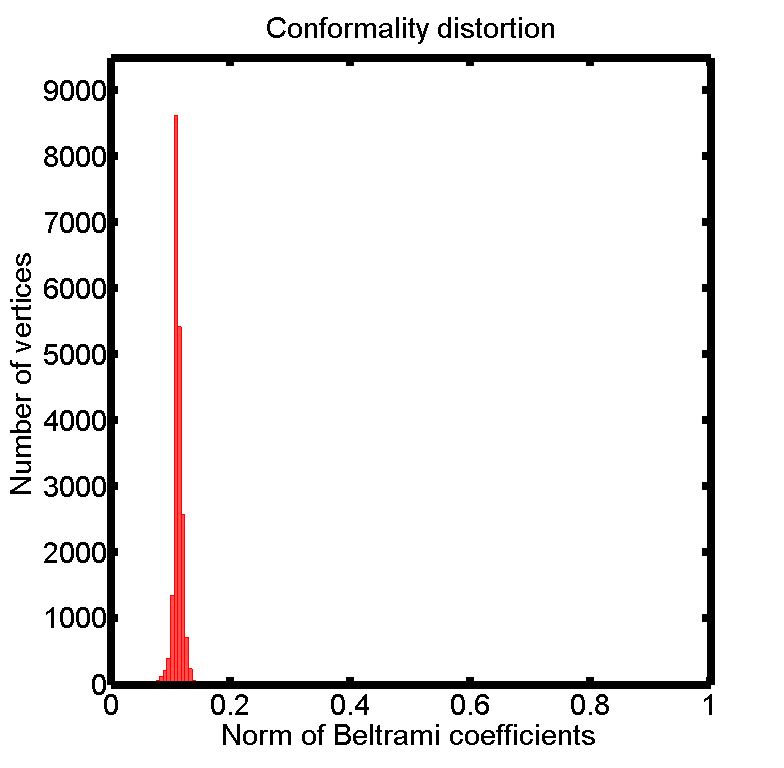}
\caption{Teichm\"uller parameterization of a noisy facial point cloud. Left: The input noisy point cloud. Middle: The parameterization result. The landmark constraints are represented by the green and blue points. The green points represent the original locations and the blue points are the target locations. Right: The histogram of the norm of the Beltrami coefficients. }
\label{fig:noisy_face}
\end{figure}

\begin{table}[t]
\centering
\begin{tabular}{|c|c|c|c|c|}
\hline
Point clouds & \# of points & Time (m) & Mean($|\bm{\mu}|$)& Var($|\bm{\mu}|$) \\ \hline
Cesar & 9755 & 0.3081 & 0.2402 & 6.1474e-04 \\\hline
Muscle guy & 14240 & 0.3241 & 0.2934 & 9.8858e-04 \\\hline
Obese man & 14415 & 0.3568 & 0.2508 & 7.7045e-04 \\\hline
Lucy & 15167 & 0.6486 & 0.1523 & 1.4981e-04 \\\hline
Spiral & 28031 & 1.1648 & 0.1986 & 8.0233e-04 \\\hline
Human face (neutral) & 31350 & 0.6517 & 0.0713 & 3.3564e-05 \\\hline 
Human face (angry) & 31468 & 0.7428 & 0.0796 & 3.8812e-05 \\\hline 
Human face (sad) & 31543 & 0.6117 & 0.0693 & 5.6557e-05 \\\hline 
Human face (happy) & 31878 & 0.8536 & 0.0909 & 7.7383e-05 \\\hline 
Bumps & 114803 & 2.5353 & 0.1070 & 3.4902e-05  \\\hline
Lion & 129957 & 2.5842 & 0.0893 & 2.8239e-05  \\\hline
Noisy face 1 & 20184 & 0.9172 & 0.1111 & 0.0140 \\\hline 
Noisy face 2 & 19743 & 0.5191 & 0.1077 & 0.0130 \\\hline 
Noisy face 3 & 20194 & 0.4905 & 0.0787 & 0.0157 \\\hline 
\end{tabular}
\caption{The performance of our landmark-constrained Teichm\"uller parameterization algorithm for various point clouds.}
\label{tab:param-open}
\end{table}


After demonstrating the advantage of our proposed approximation scheme for quasi-conformal mappings, we illustrate the effectiveness of our landmark constrained Teichm\"uller parameterization algorithm using numerous experiments.

Note that there is a balancing parameter $\gamma$ in the hybrid equation (\ref{eqt:hybrid}) of our TEMPO algorithm. In this section, we first compare Teichm\"uller parameterization results with different choices of $\gamma$. Two categories of the parameter $\gamma$ are considered in running the TEMPO algorithm:
\begin{enumerate}[(i)]
\item Decreasing values: $\gamma_n$ is gradually decreased during the iterations. 
\item Constant values: $\gamma$ is set to be a constant throughout the iterations.
\end{enumerate}

%
In the experiment, we consider computing the Teichm\"uller parameterization of a real 3D point cloud. Since we do not have any explicit formula as the ground truth solution, we evaluate the performance by computing the variance of $|\bm{\mu}|$. The Teichm\"uller parameterization is desired to be with a negligible variance of $|\bm{\mu}|$.

The parameterization results with different input $\gamma$ are recorded in Table \ref{tab:gamma}. For fairness, we perform the same number of iterations in Algorithm \ref{alg:tmap}. It is noteworthy that the cases $\gamma = 0$ or $\infty$ respectively correspond to solely solving Equation (\ref{eqt:firstorder}) or Equation (\ref{eqt:secondorder}). %

It can be observed that if $\gamma = 0$, the algorithm fails. On the other hand, the result associated with $\gamma = \infty$ is not as accurate as those associated with non-zero finite $\gamma$. Consequently, our hybrid scheme outperforms the two existing methods.

In addition, the results associated with different non-zero finite values of $\gamma$ are highly similar, regardless of whether $\gamma$ is decreasing or constant. This observation suggests that the assumption $\lim_n \gamma_n = 0$ in Proposition \ref{prop:tmap conv 1pc} can be relaxed in practice. Therefore, for simplicity, we can take a constant $\gamma$ in running the TEMPO algorithm.

In the following experiments, the balancing parameter $\gamma$ in the hybrid equation (\ref{eqt:hybrid}) is set to be $\gamma \equiv 0.5$. The stopping criterion is set to be $\epsilon = 10^{-6}$.

Figure \ref{fig:shaped_ball} shows the landmark-constrained Teichm\"uller parameterization result of a spiral point cloud. The green points and the blue points represent the prescribed landmark constraints on the planar domain. Even under the large landmark deformations, our parameterization result ensures a uniform conformality distortion. This indicates the Teichm\"uller property of our proposed algorithm. More examples are shown in Figure \ref{fig:lion} and Figure \ref{fig:lucy}. It can be easily observed that in all of our experiments, the norms of the Beltrami coefficients always accumulate at a specific value. This indicates that our parameterization results of point clouds closely resemble the continuous Teichm\"uller parameterizations of simply-connected open surfaces.

Noticing that real point clouds are usually with certain undesirable noises, we also test the performance of our TEMPO algorithm on noisy point clouds. Figure \ref{fig:noisy_face} shows a noisy facial point cloud with 5\% uniformly distributed random noise and the Teichm\"uller parameterization result. The histogram of the resulting Beltrami coefficient indicates that the resulting map is a Teichm\"uller extremal map. This demonstrates the effectiveness of our TEMPO algorithm on noisy point cloud data.

Extensive experiments of our proposed landmark constrained Teichm\"uller parameterization algorithm have been carried out using different point clouds with disk topology. The statistics of the experiments are listed in Table \ref{tab:param-open}. Our proposed method is highly efficient. The computations of the landmark-matching Teichm\"uller parameterizations complete within 1 minute on average for point clouds with moderate size. Even for point clouds with over 100K vertices, our algorithm only takes a few minutes. To quantitatively assess the performance of our proposed method, we compute the mean and the variance of the Beltrami coefficient $\bm{\mu}$ of every Teichm\"uller parameterization result. It is noteworthy that the variances of the Beltrami coefficients are close to zero in all experiments. In other words, the conformality distortions are uniform. It can also be observed in Table \ref{tab:param-open} that for denser point clouds, the variances of the results are in general smaller. All of the above results agree with the theories we have established. They reflect the accuracy of our proposed algorithm for Teichm\"uller parameterizations of feature-endowed point clouds.



\subsection{Shape analysis of point clouds using Teichm\"uller metric}
\begin{figure}[t]
\centering
\includegraphics[width=0.32\textwidth]{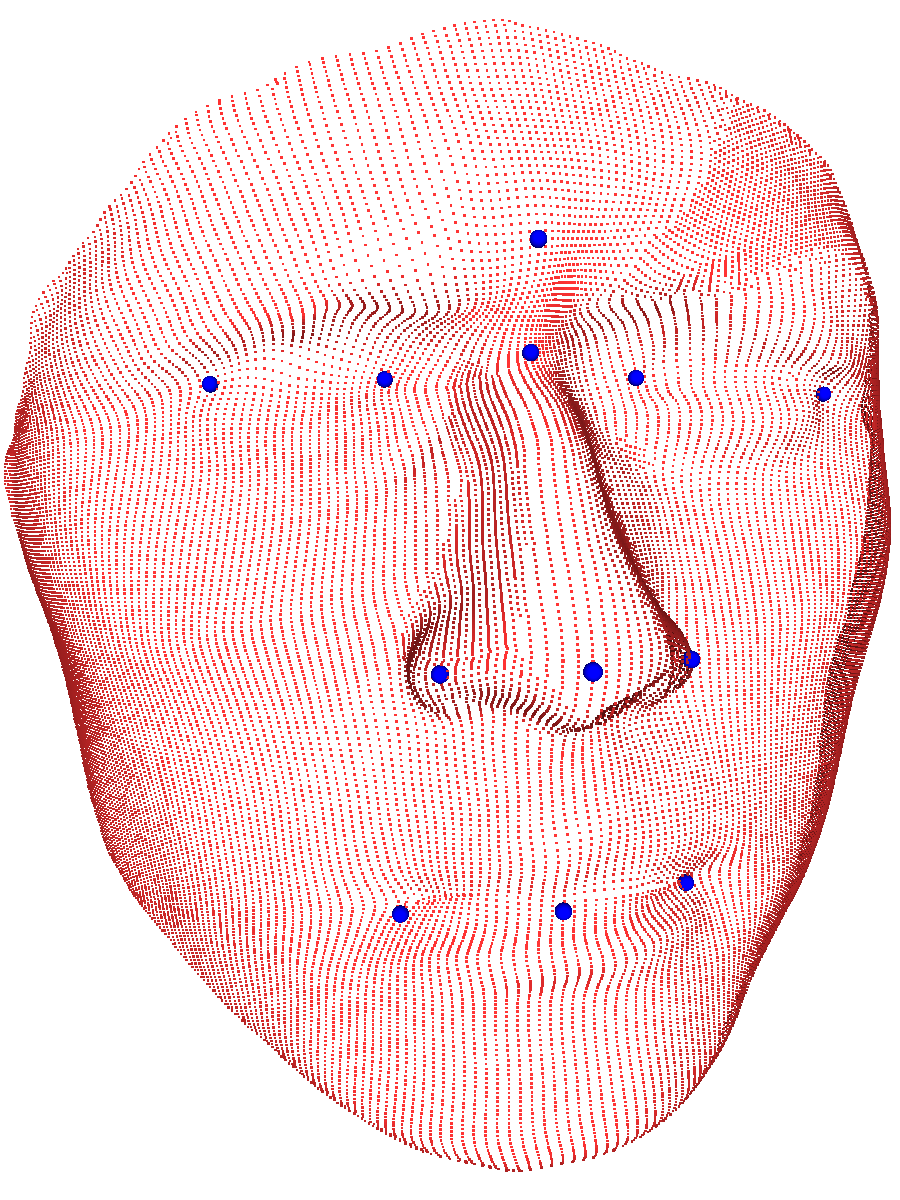}
\includegraphics[width=0.32\textwidth]{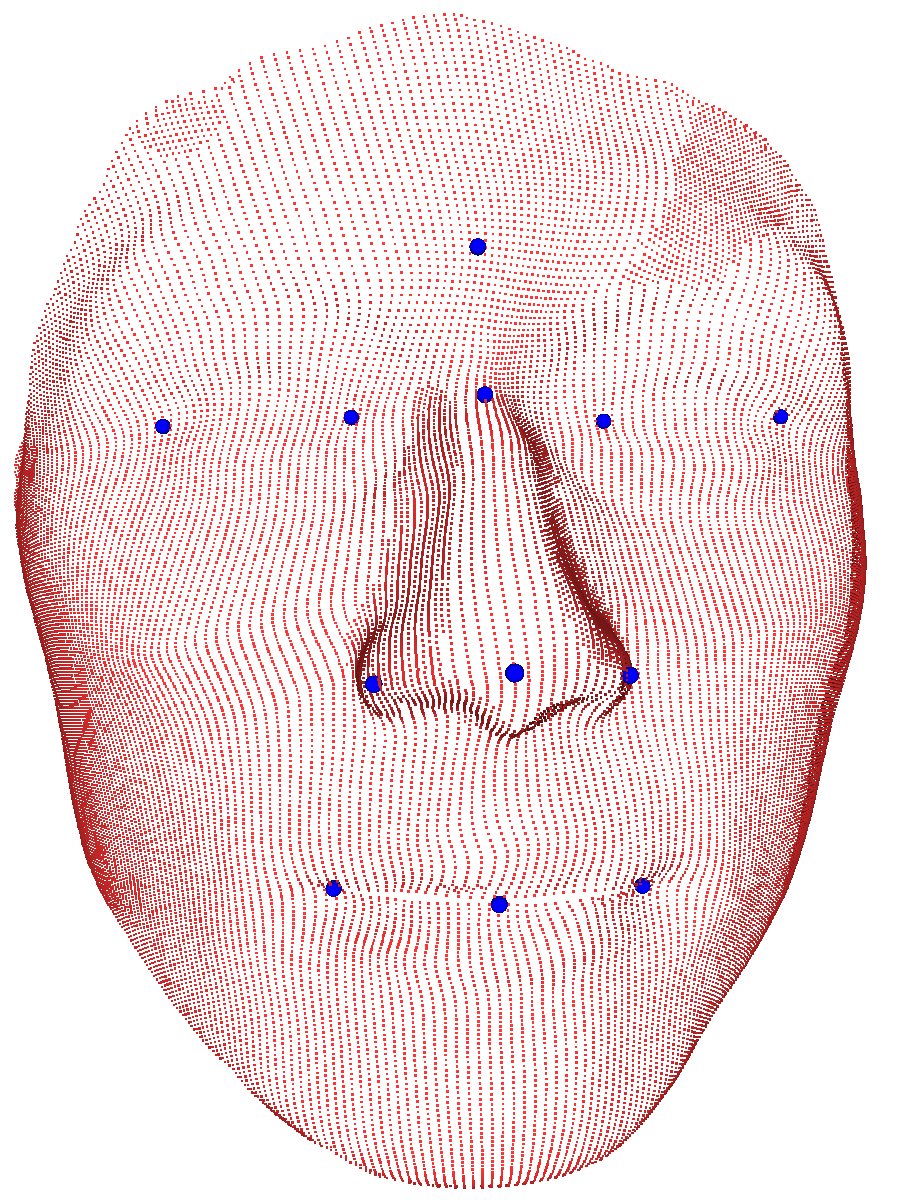}
\includegraphics[width=0.32\textwidth]{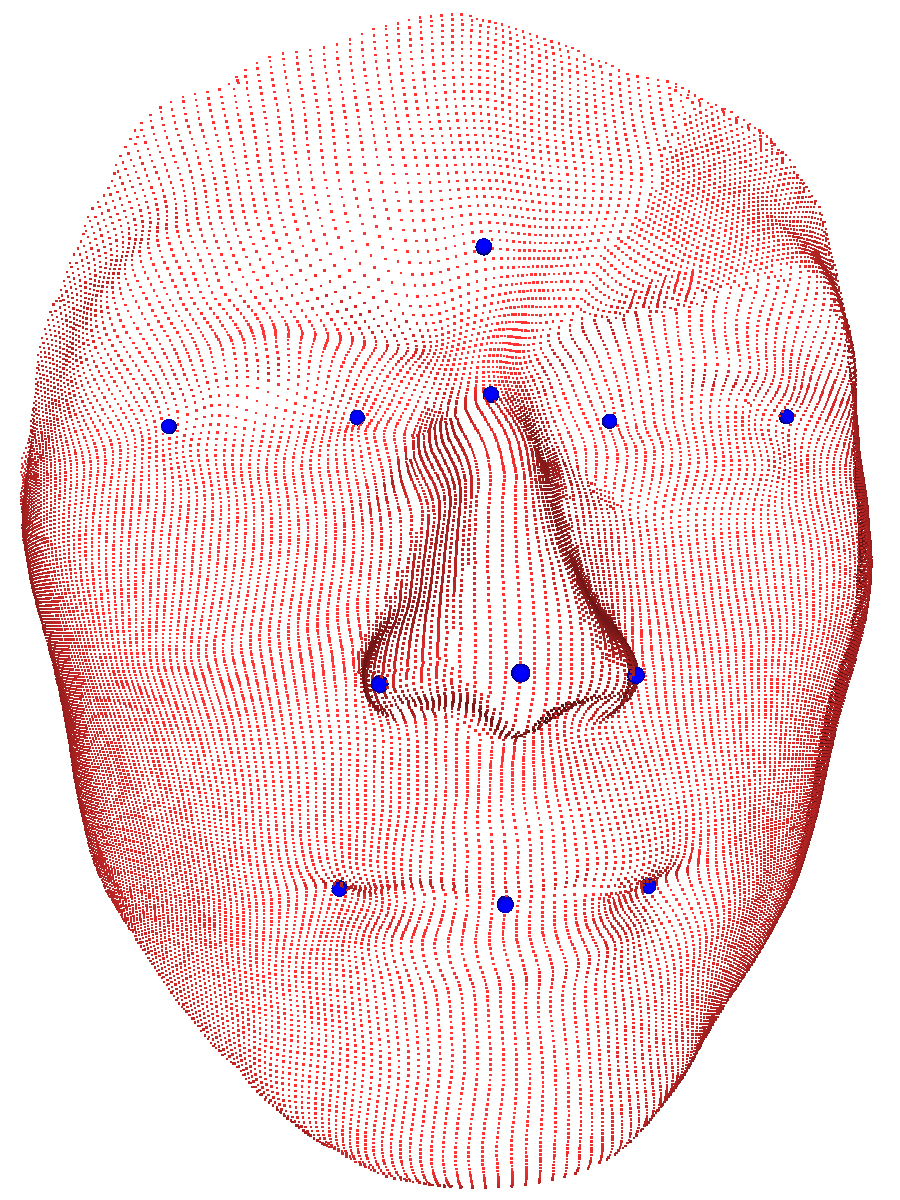}
\includegraphics[width=0.32\textwidth]{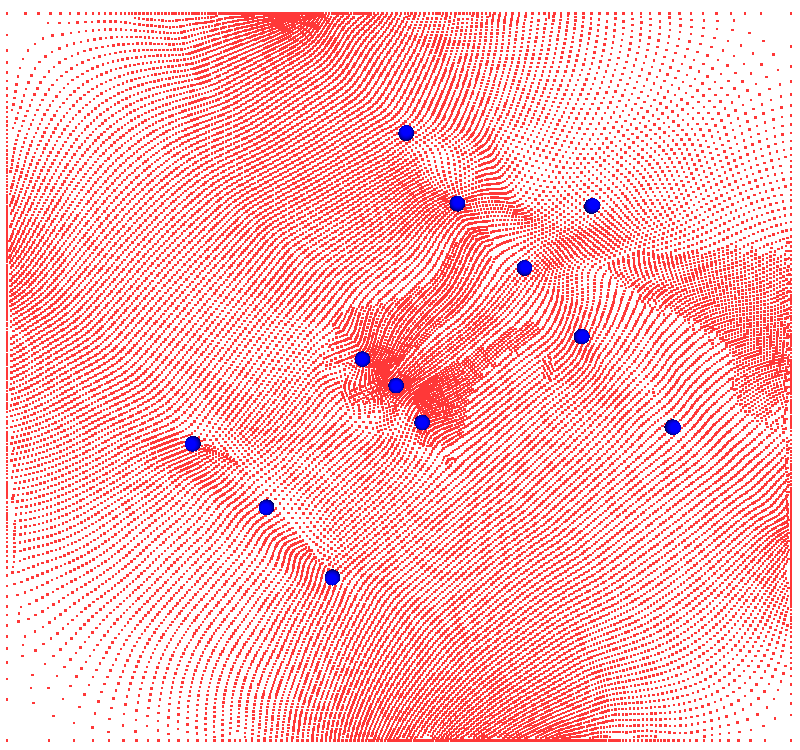}
\includegraphics[width=0.32\textwidth]{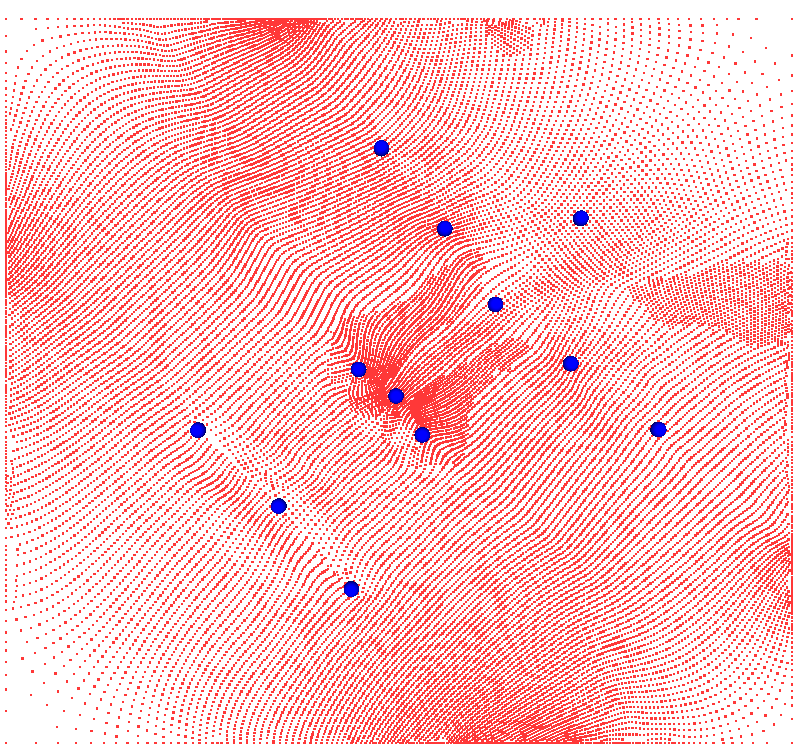}
\includegraphics[width=0.32\textwidth]{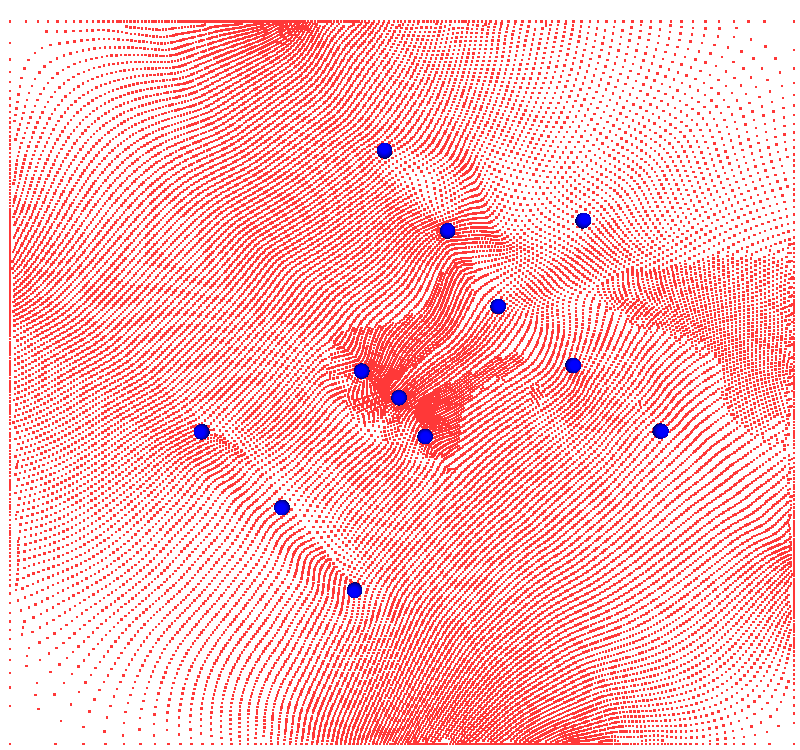}
\caption{Teichm\"uller registration between two human face point clouds with prescribed landmark constraints. Top left: The source point cloud. Top middle: The target point cloud. Top right: The registration result by our algorithm. Bottom: The conformal parameterizations of the source and target point clouds, and the landmark constrained Teichm\"uller parameterization. The blue points refer to the facial landmark constraints.}
\label{fig:face_registration}
\end{figure}

\begin{figure}[h!]
\centering
\includegraphics[width=0.11\textwidth]{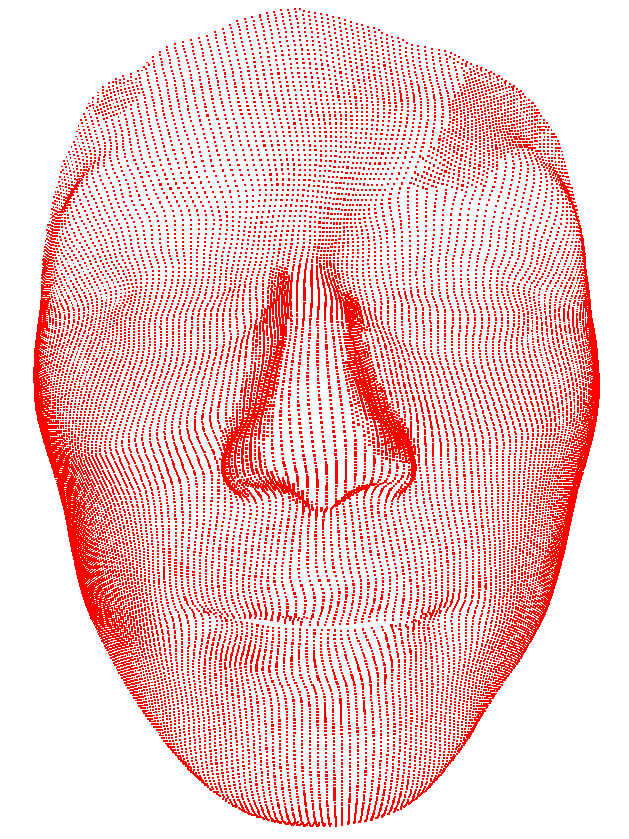}
\includegraphics[width=0.11\textwidth]{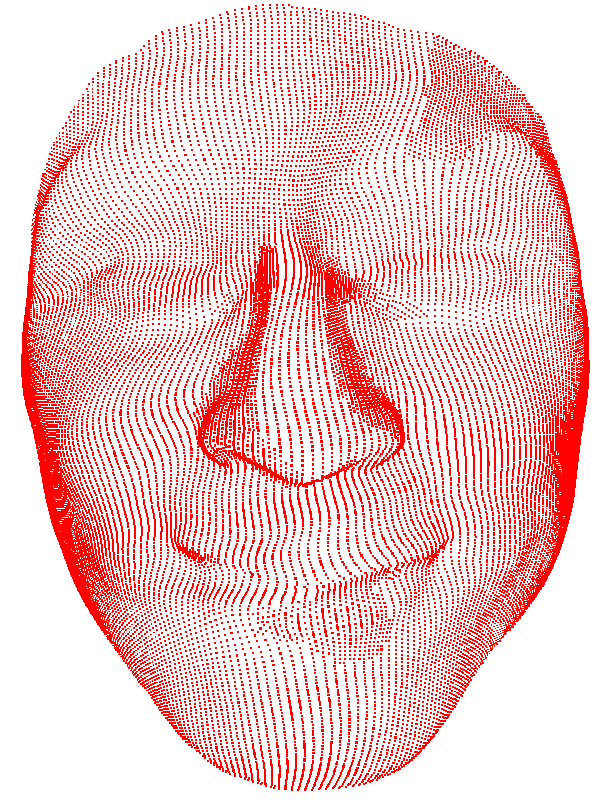}
\includegraphics[width=0.11\textwidth]{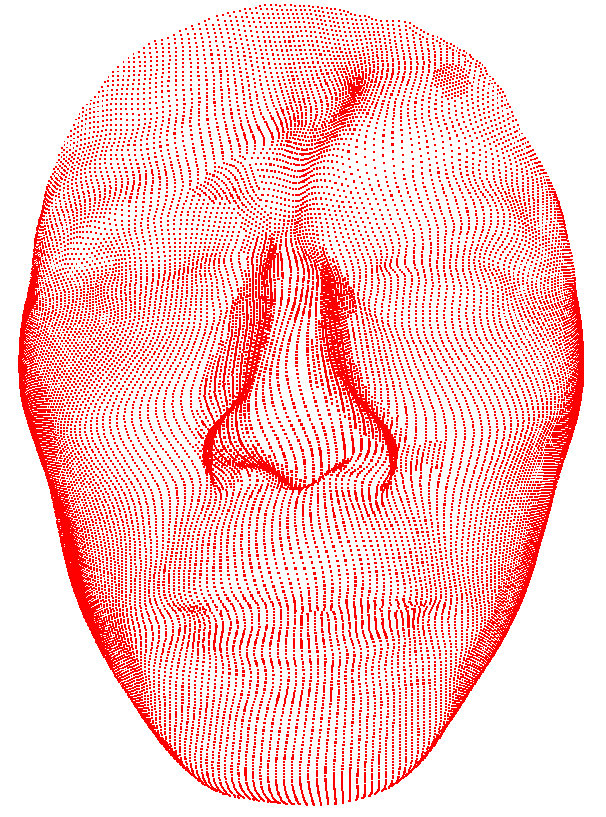}
\includegraphics[width=0.11\textwidth]{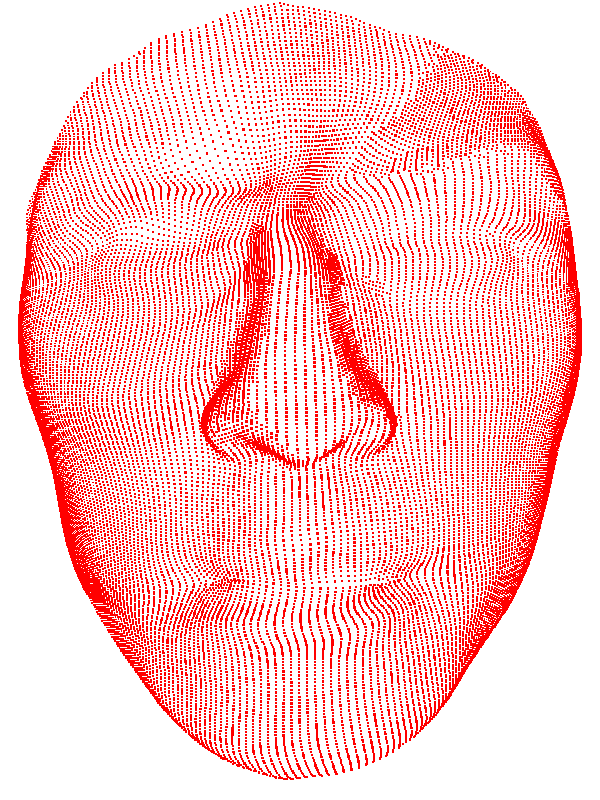}
\includegraphics[width=0.11\textwidth]{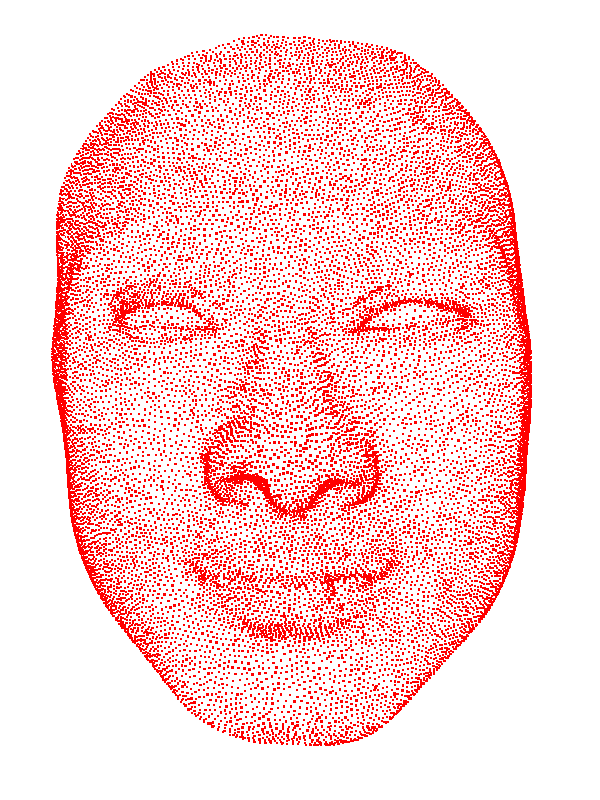}
\includegraphics[width=0.11\textwidth]{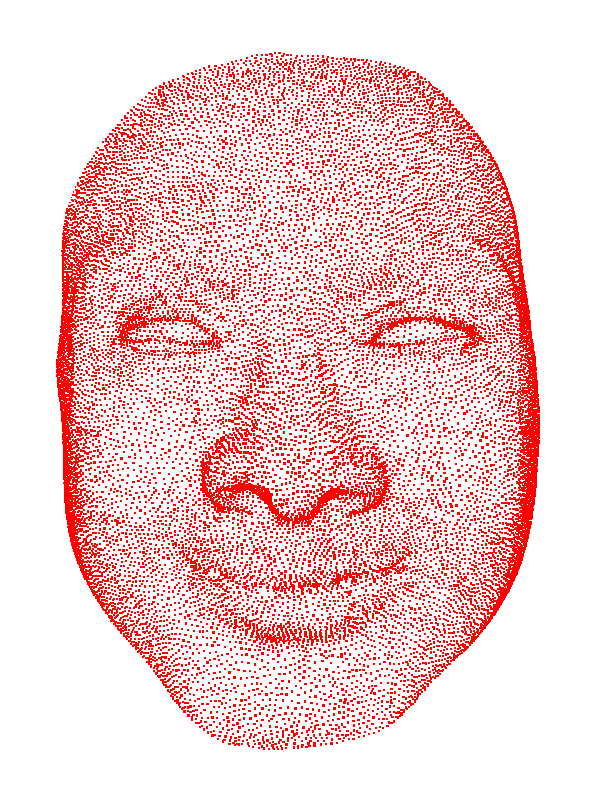}
\includegraphics[width=0.11\textwidth]{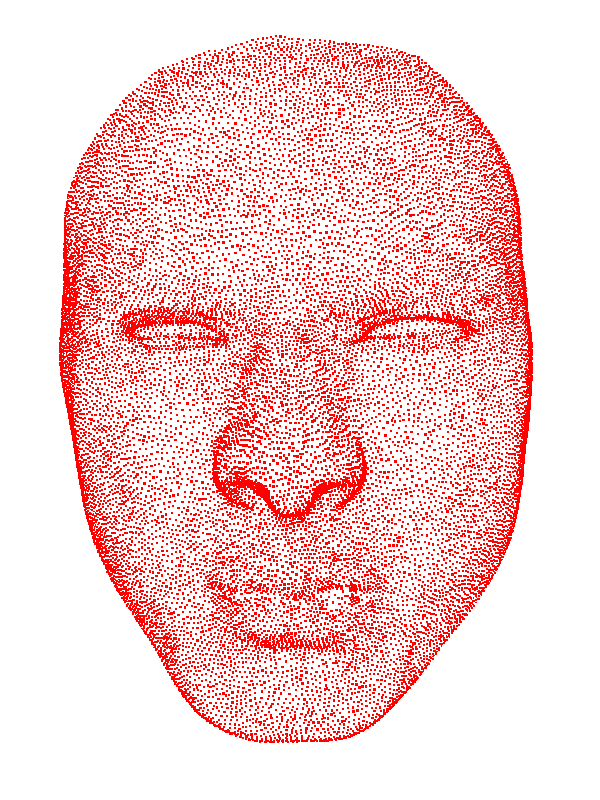}
\includegraphics[width=0.11\textwidth]{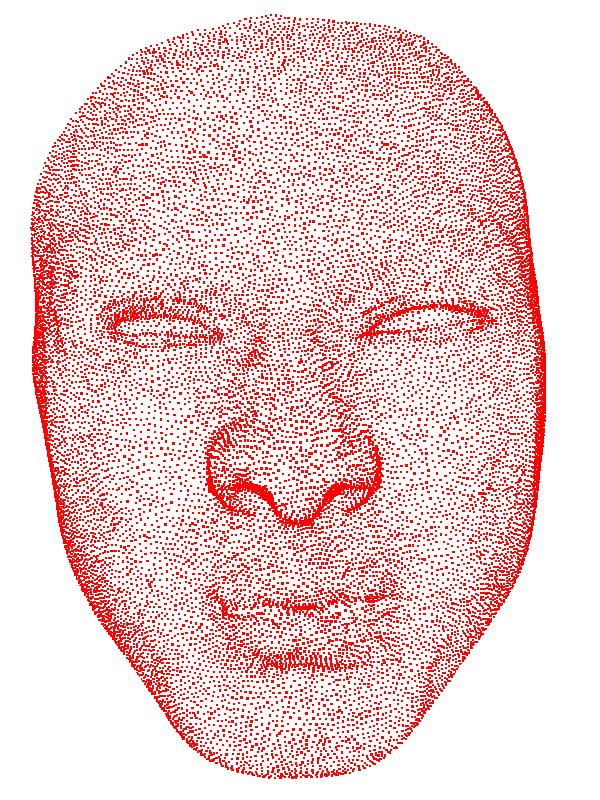}
\includegraphics[width=0.11\textwidth]{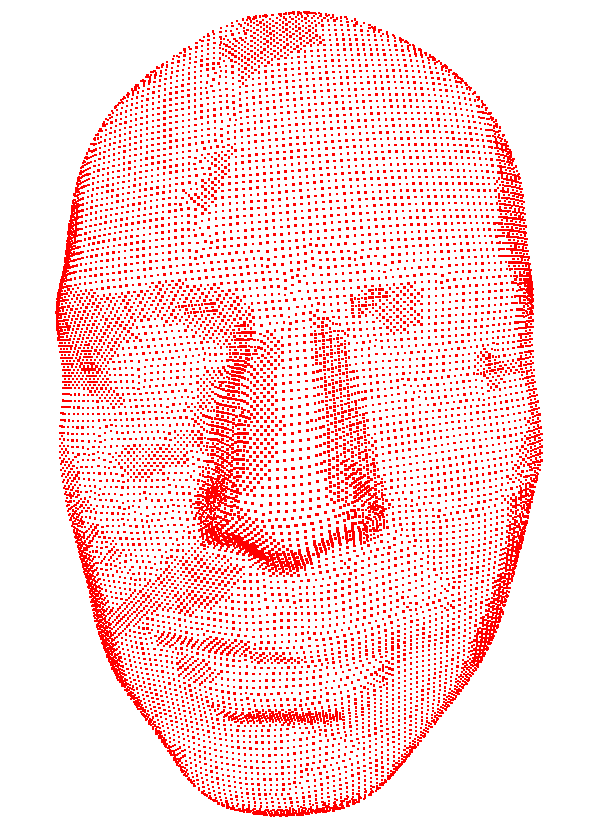}
\includegraphics[width=0.11\textwidth]{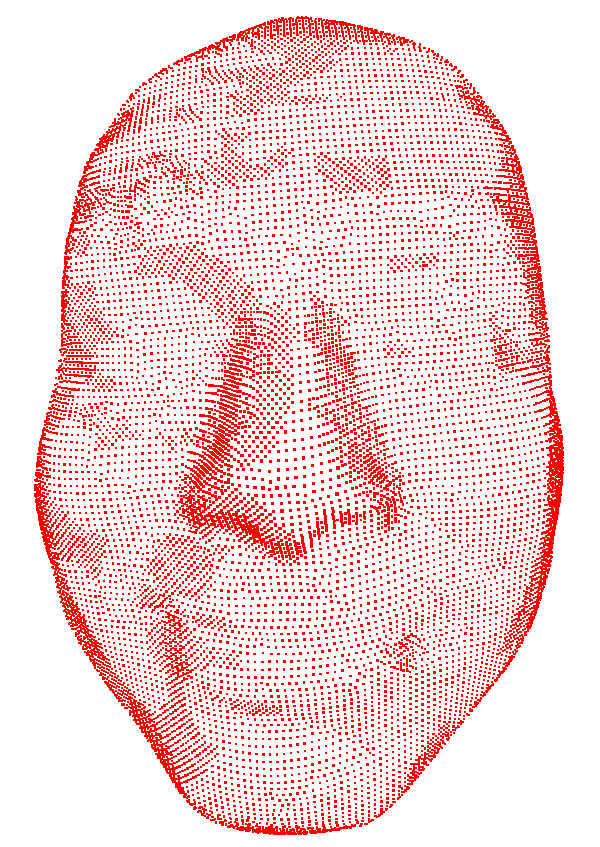}
\includegraphics[width=0.11\textwidth]{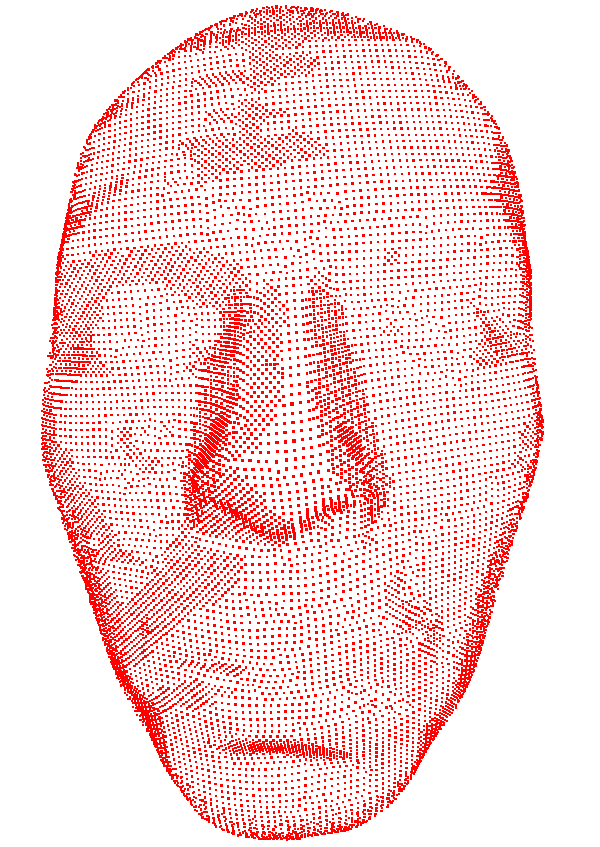}
\includegraphics[width=0.11\textwidth]{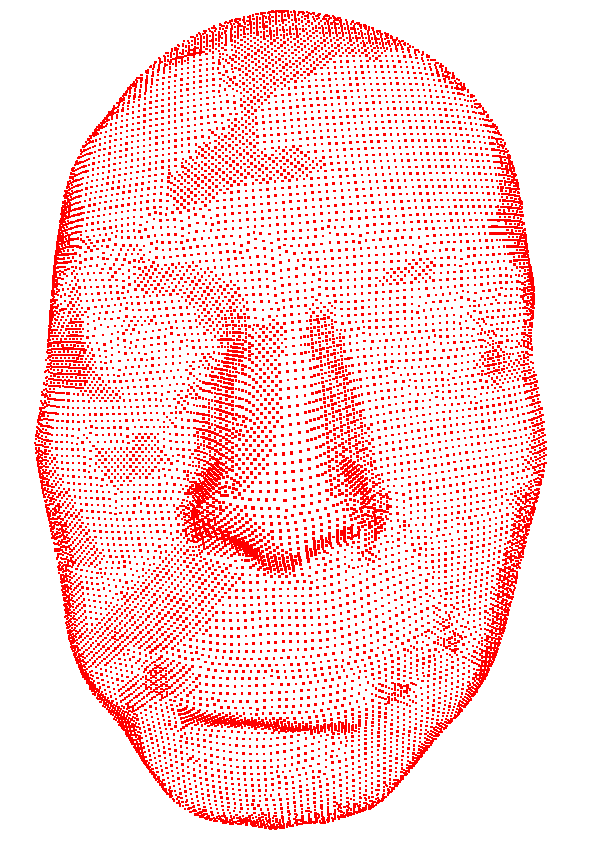}
\includegraphics[width=0.11\textwidth]{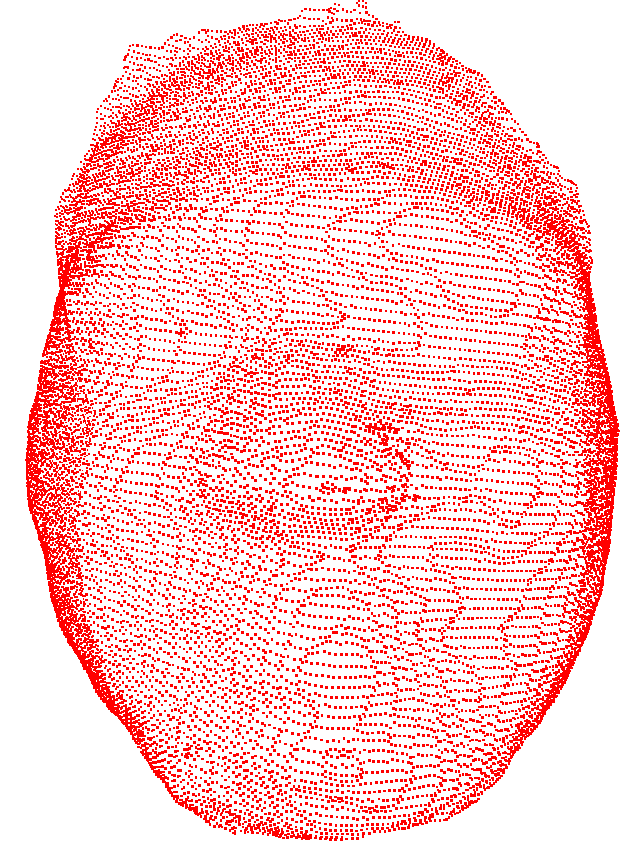}
\includegraphics[width=0.11\textwidth]{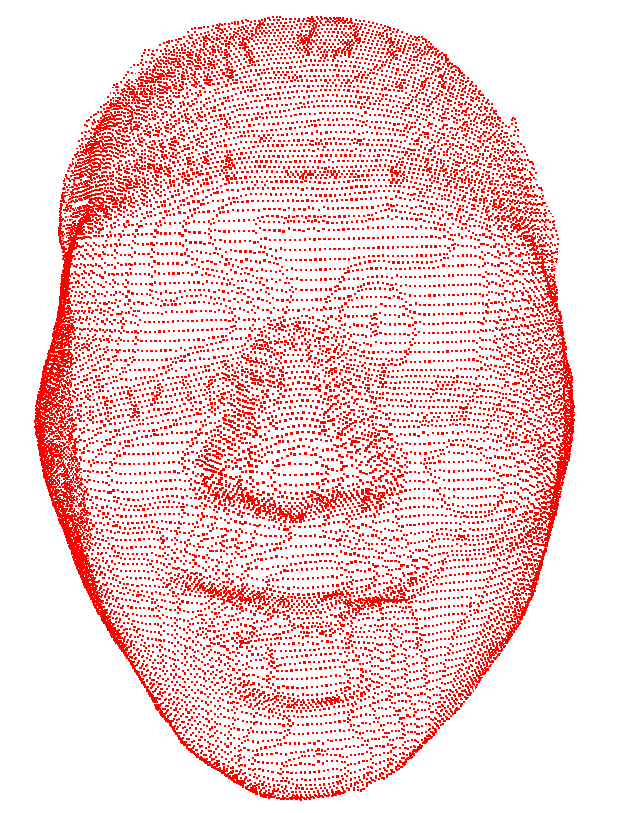}
\includegraphics[width=0.11\textwidth]{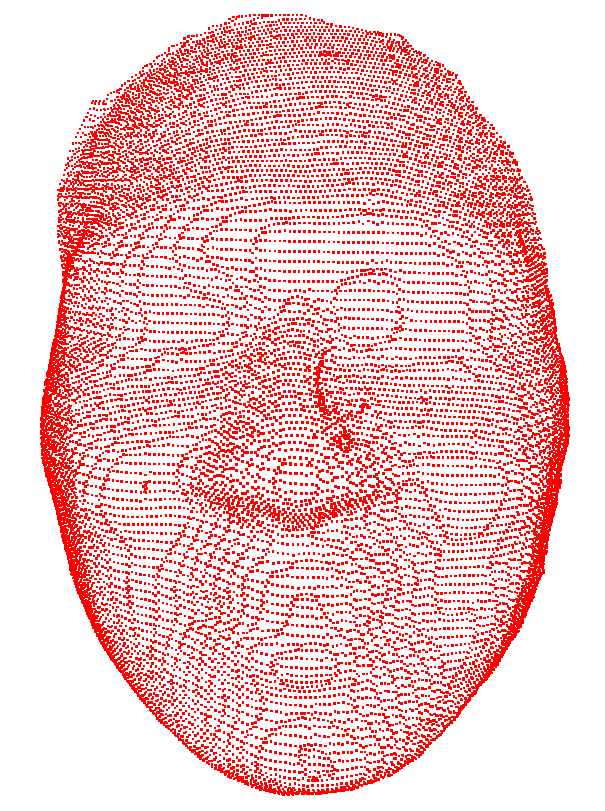}
\includegraphics[width=0.11\textwidth]{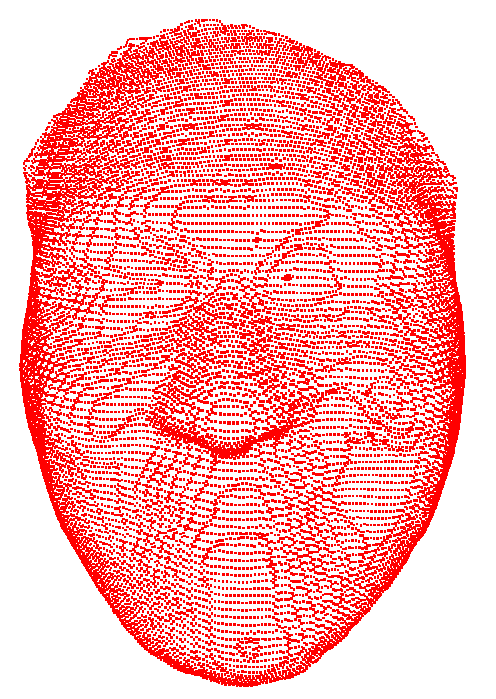}
\caption{The dataset of human face point clouds with different facial expressions used in our first experiment.}
\label{fig:face_dataset}
\end{figure}

\begin{figure}[t]
\centering
\includegraphics[width=0.45\textwidth]{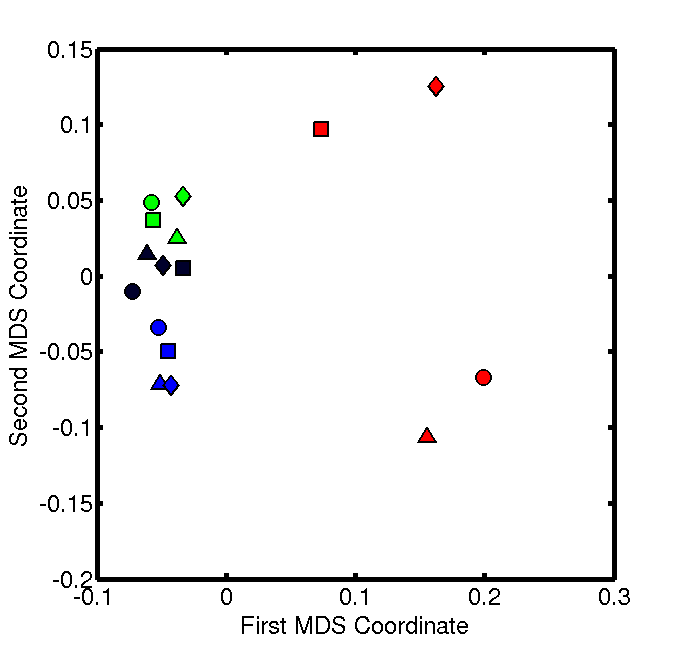}
\includegraphics[width=0.45\textwidth]{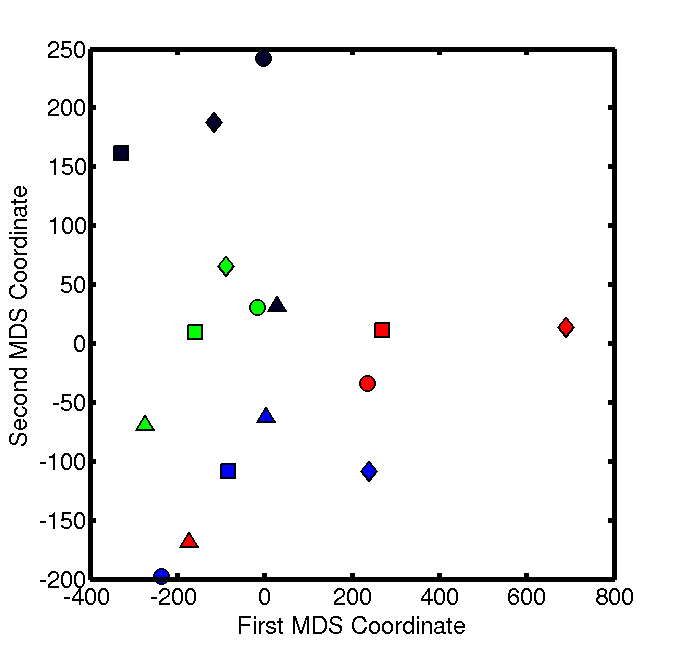}
\caption{The MDS results for our distance matrix (left) and the distance matrix in \cite{Irfanoglu04} (right). Each color represents one human, and each shape represents a type of facial expressions (circle: neutral, square: happy, triangle: sad, diamond: angry). }
\label{fig:mds}
\end{figure}

\begin{figure}[t]
\centering
\includegraphics[height=30mm]{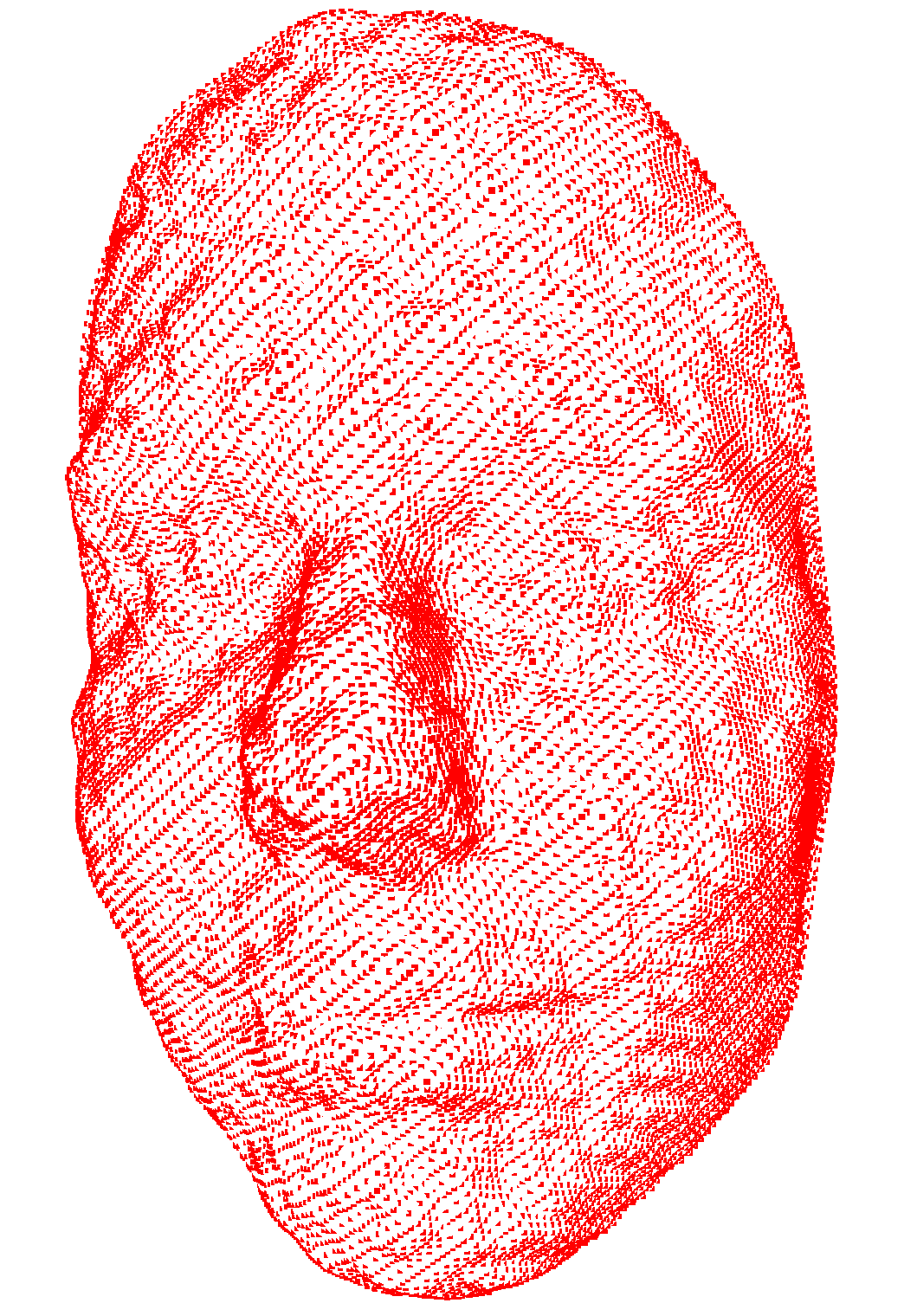}
\includegraphics[height=30mm]{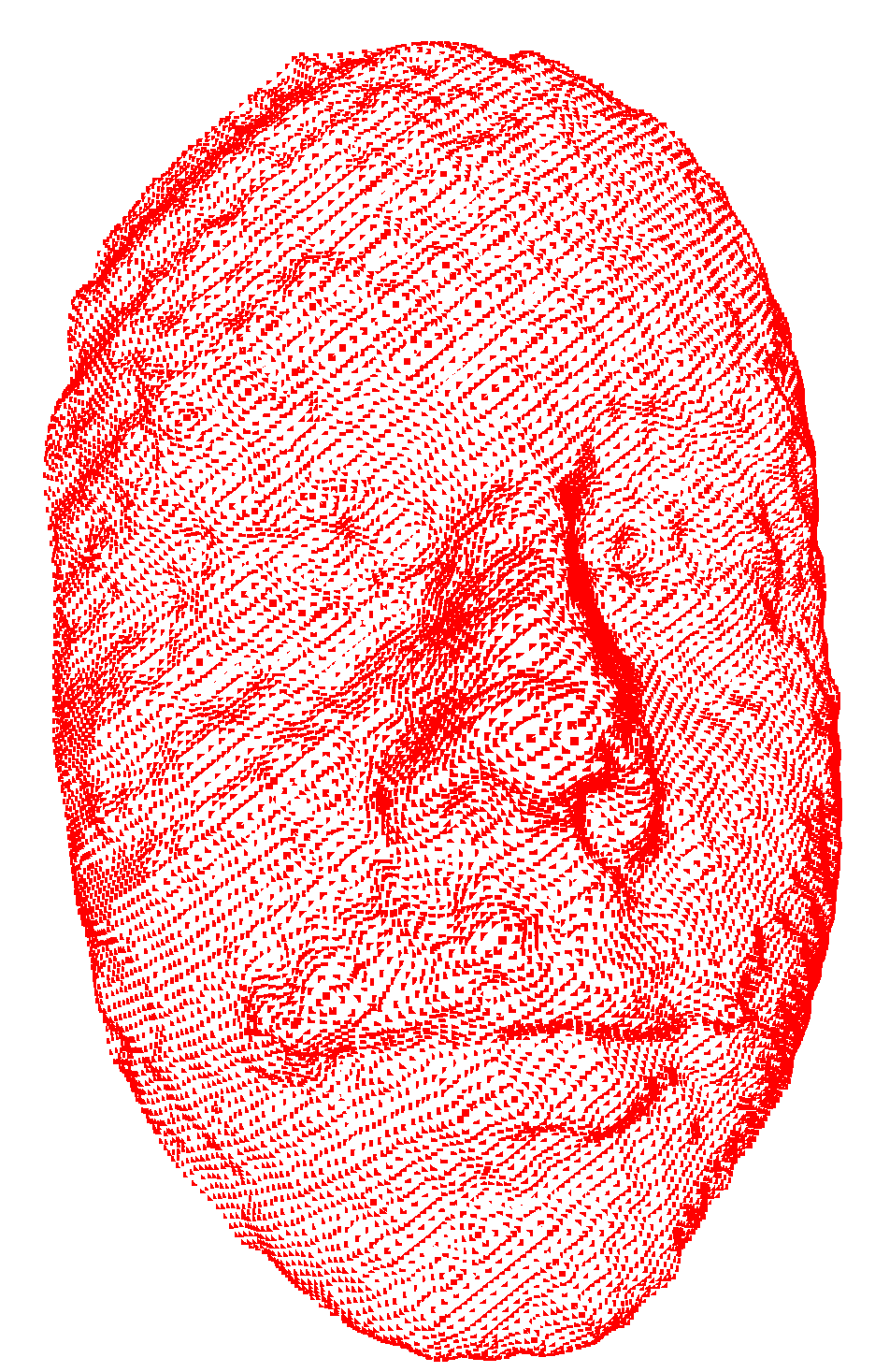}
\includegraphics[height=30mm]{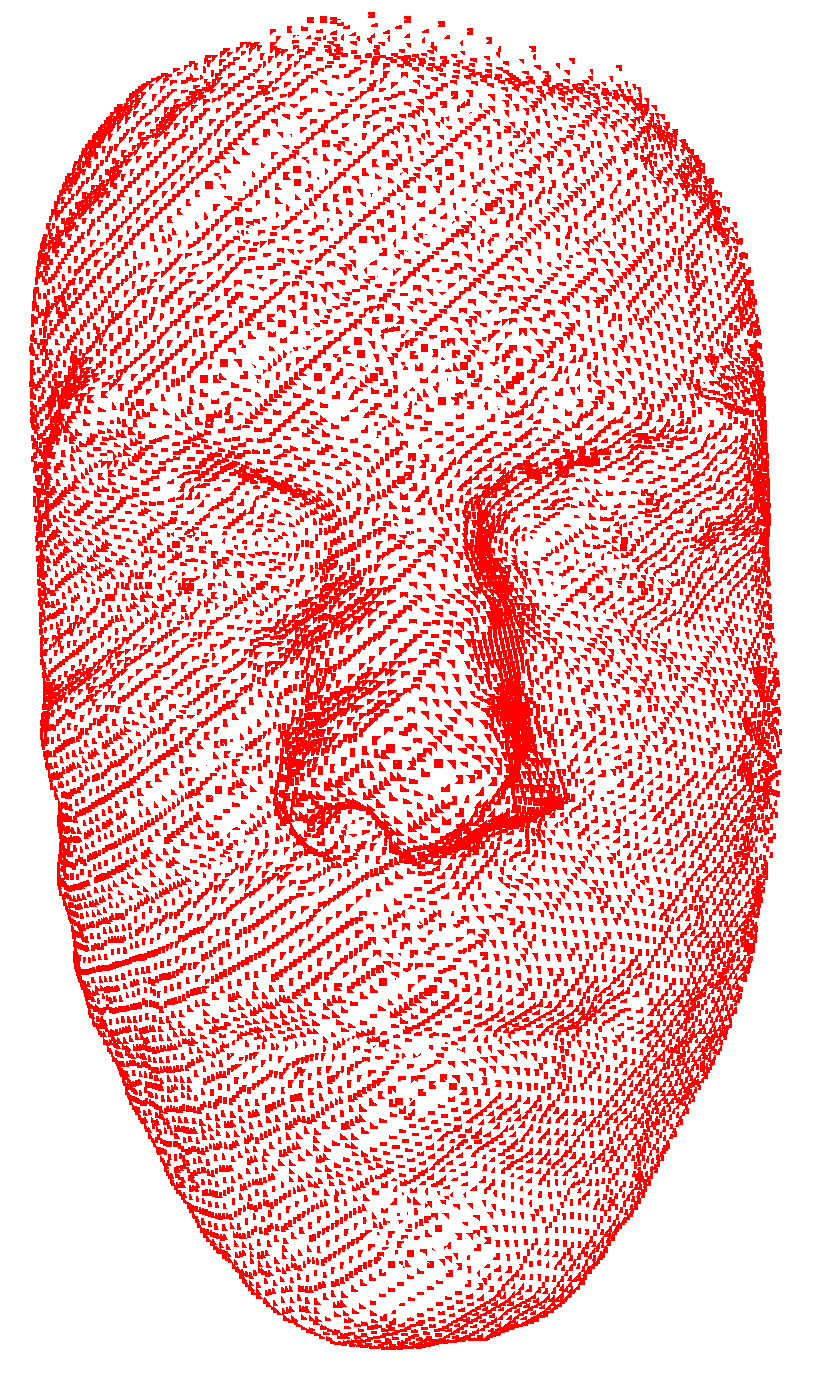}
\includegraphics[height=30mm]{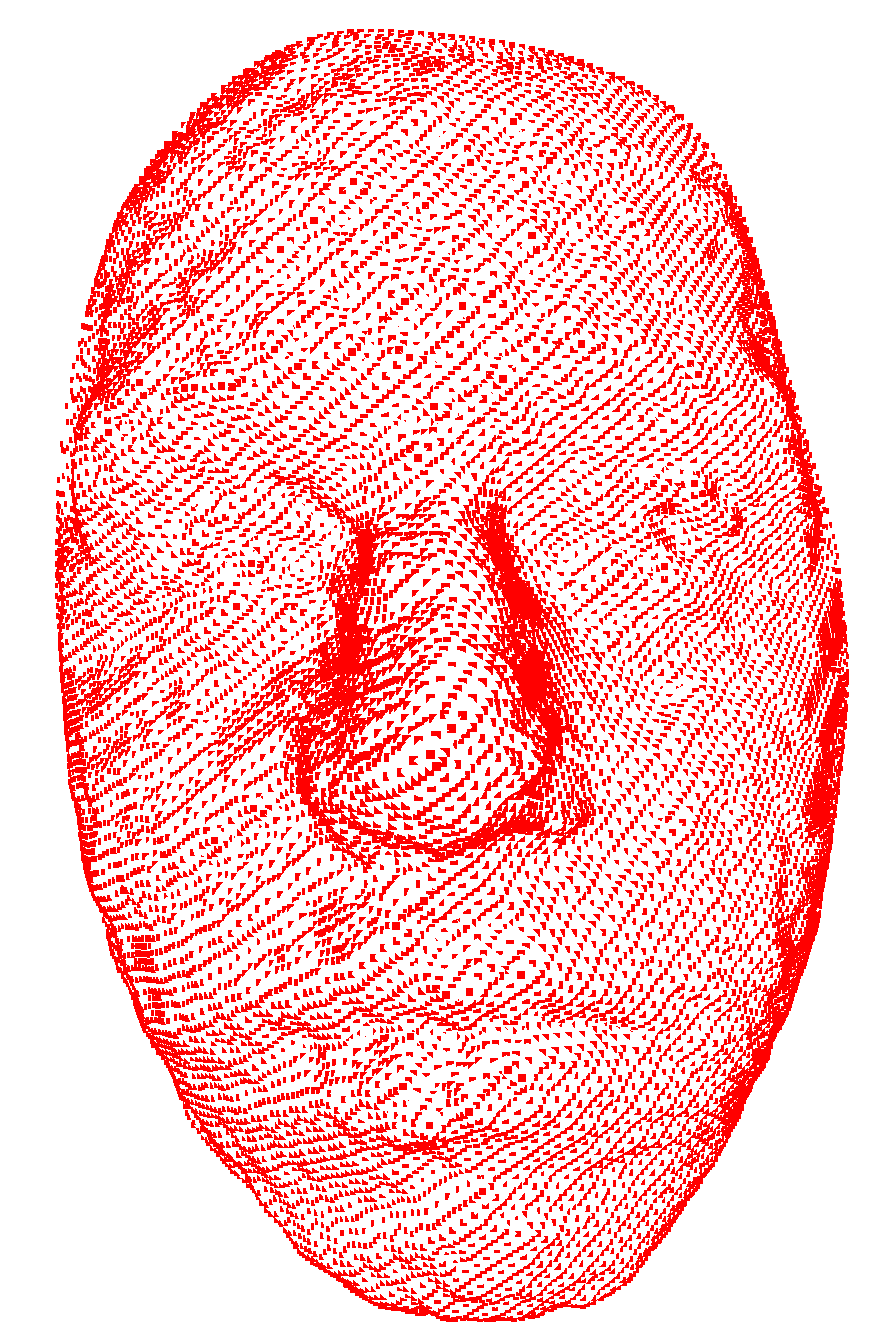}
\includegraphics[height=30mm]{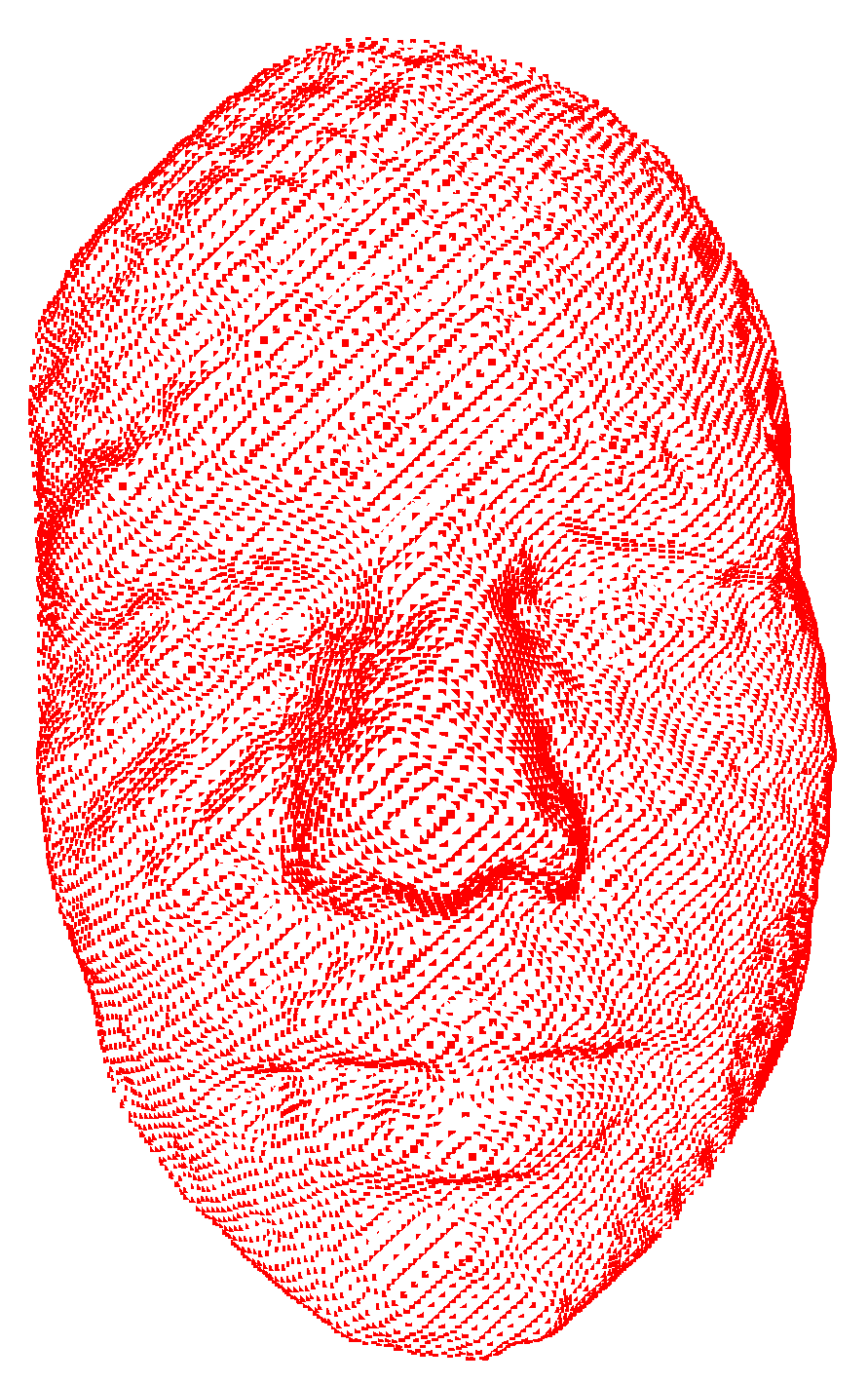}
\includegraphics[height=30mm]{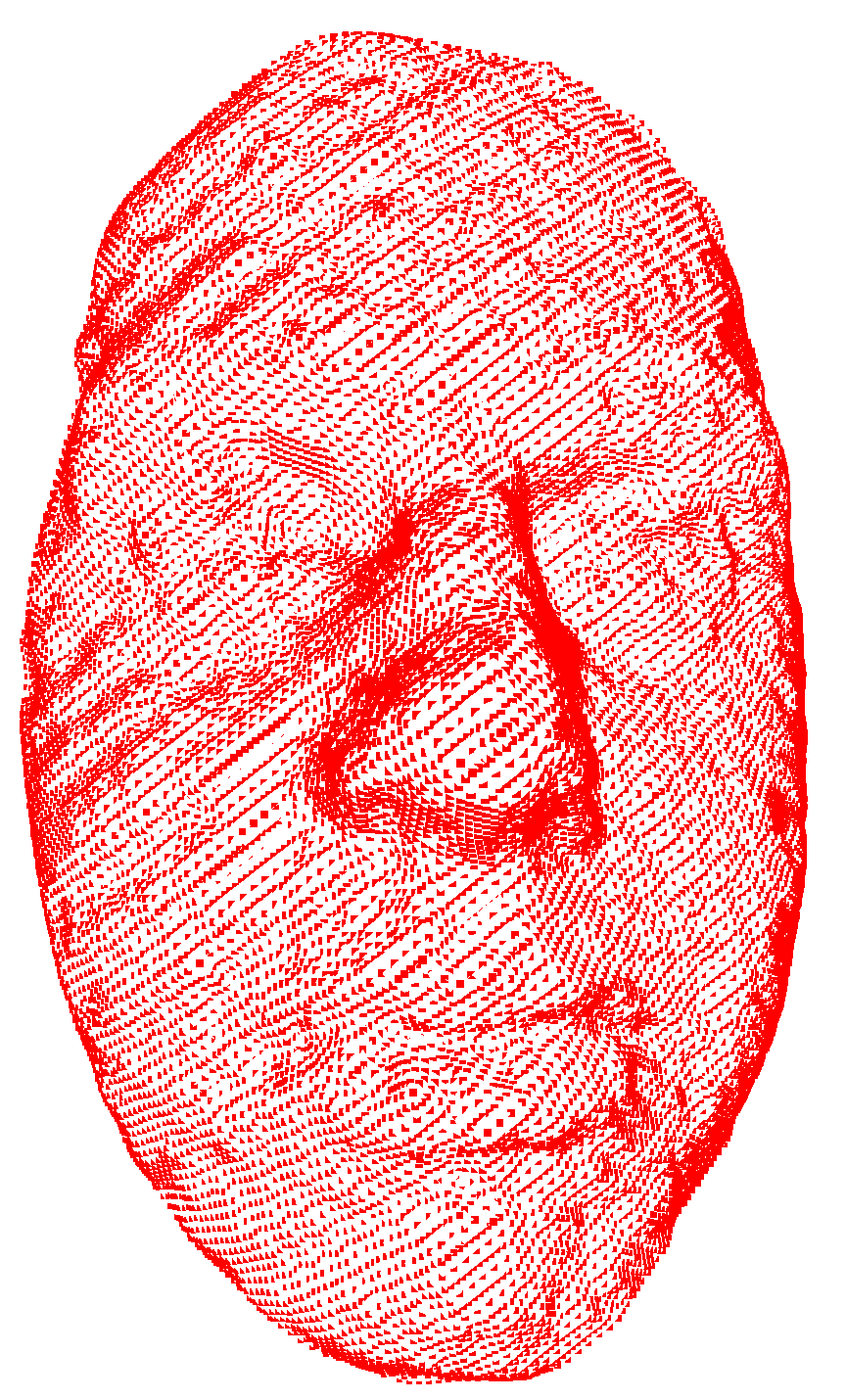}\\
\includegraphics[height=30mm]{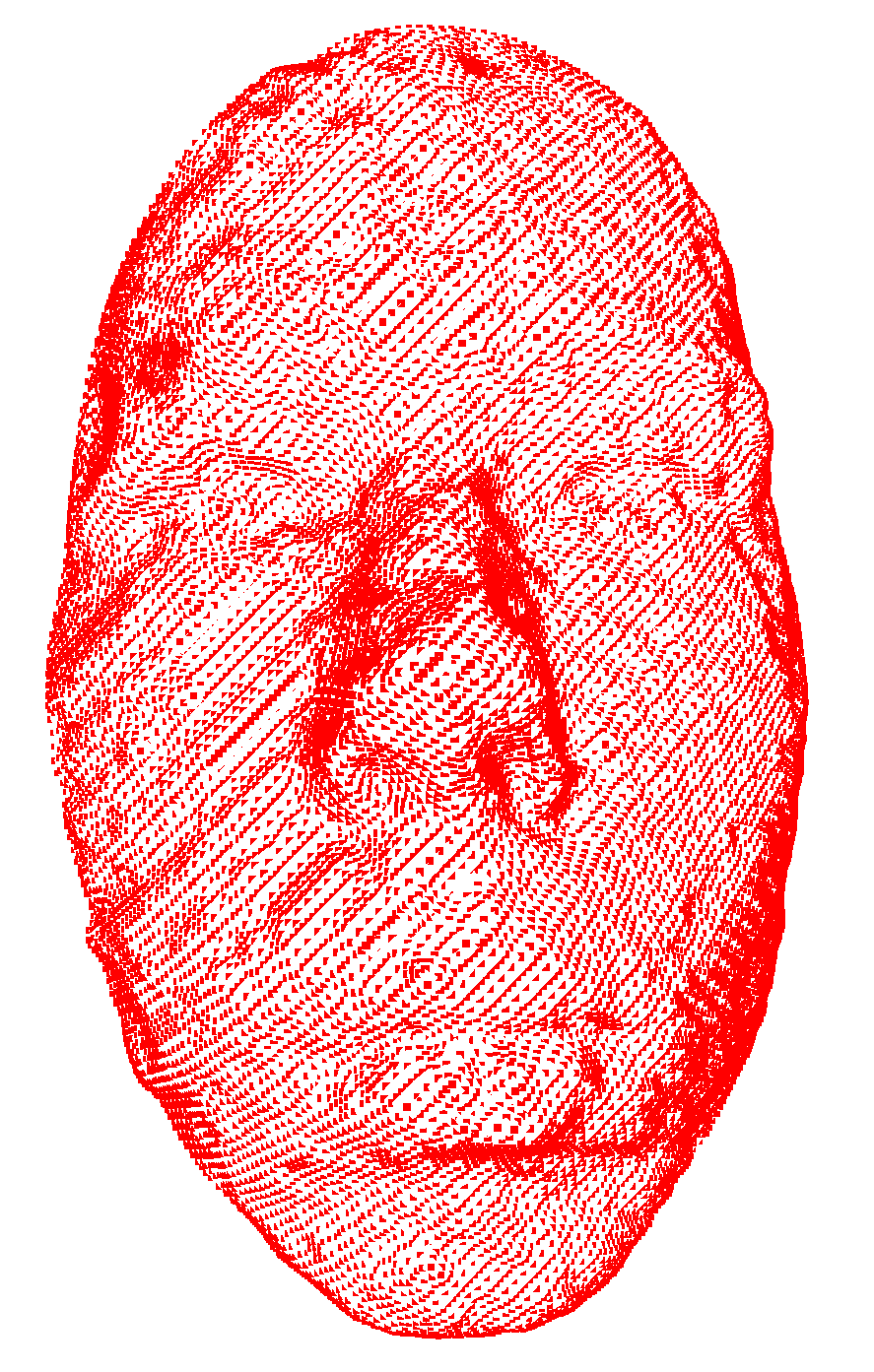}
\includegraphics[height=30mm]{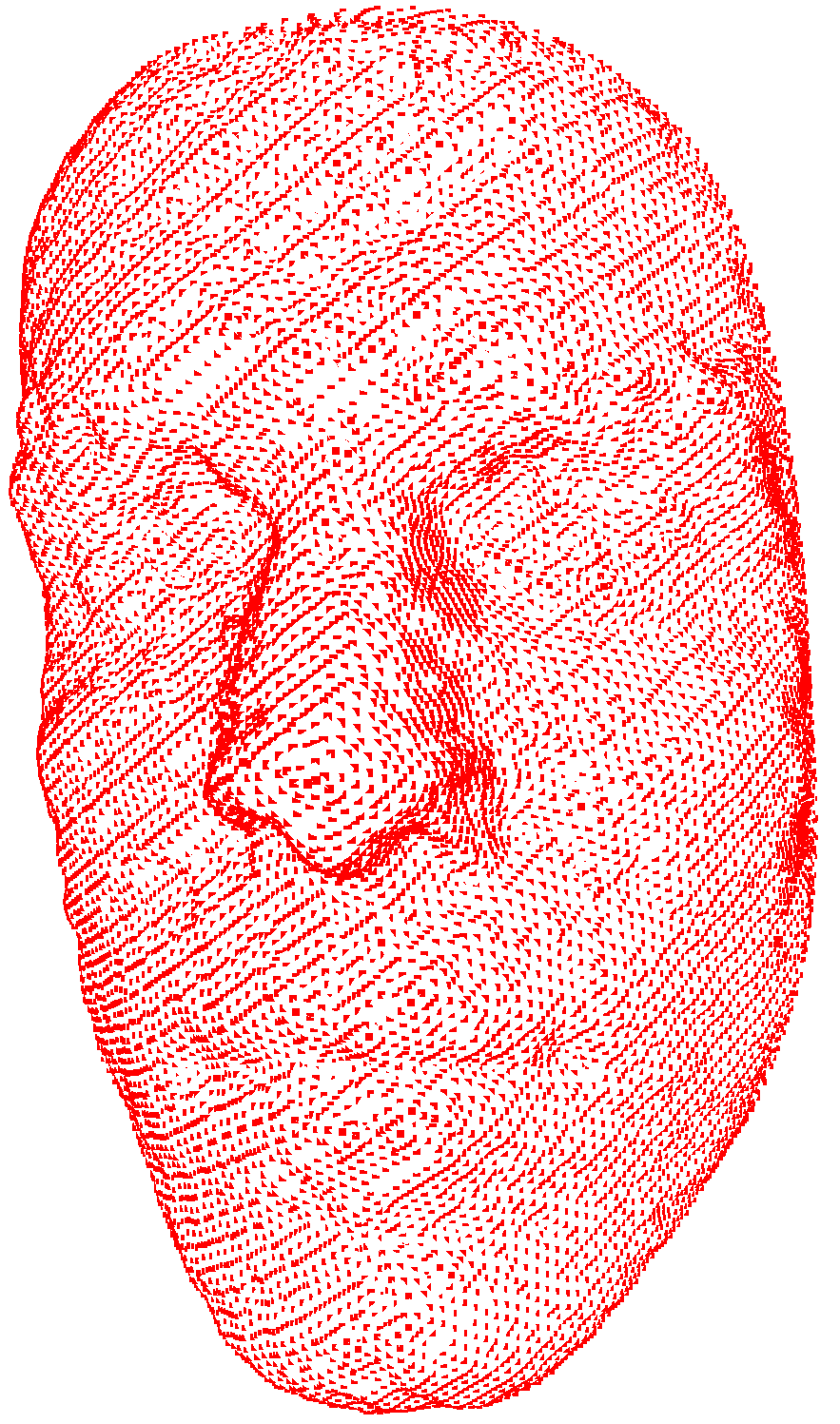}
\includegraphics[height=30mm]{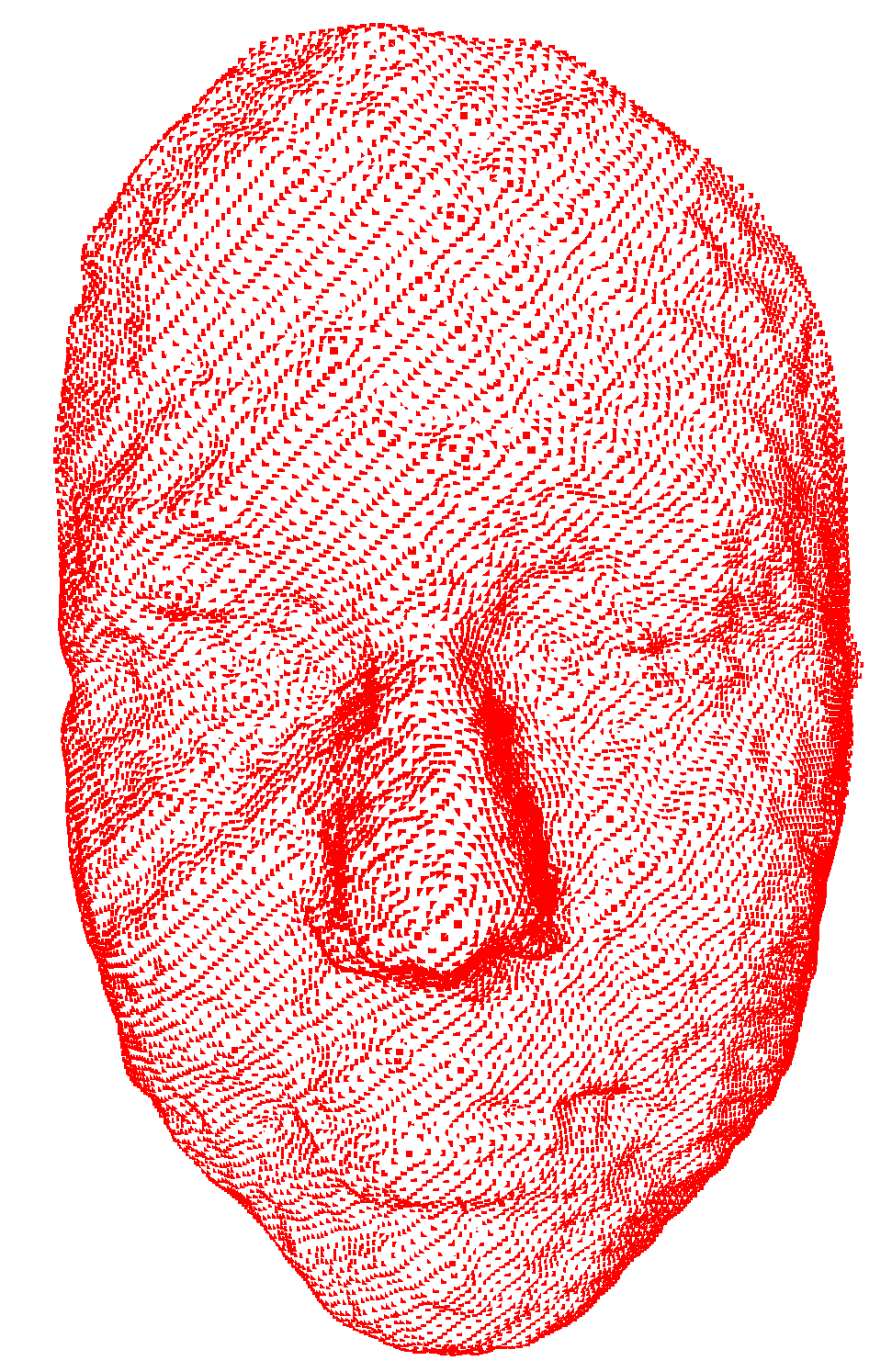}
\includegraphics[height=30mm]{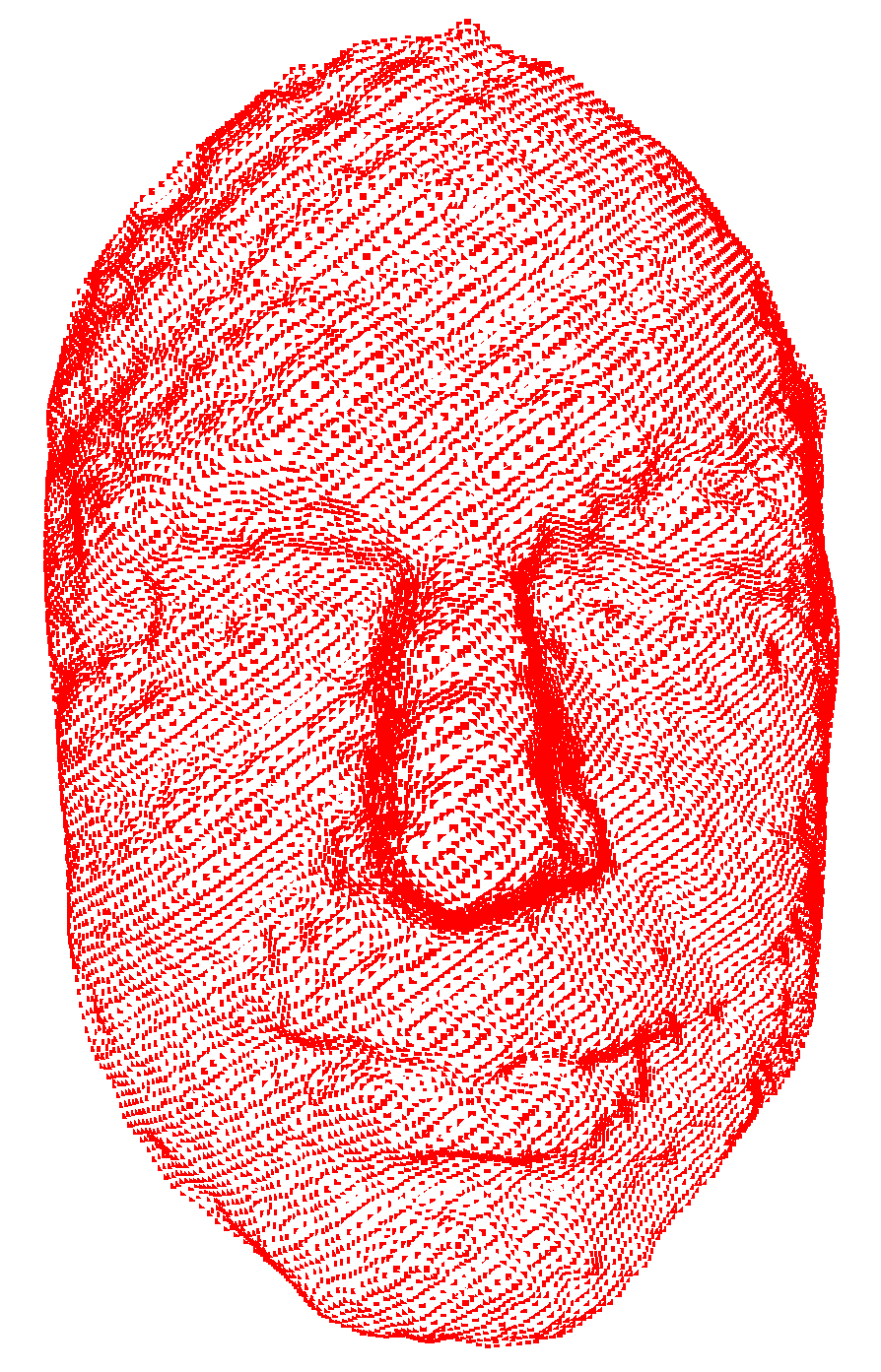}
\includegraphics[height=30mm]{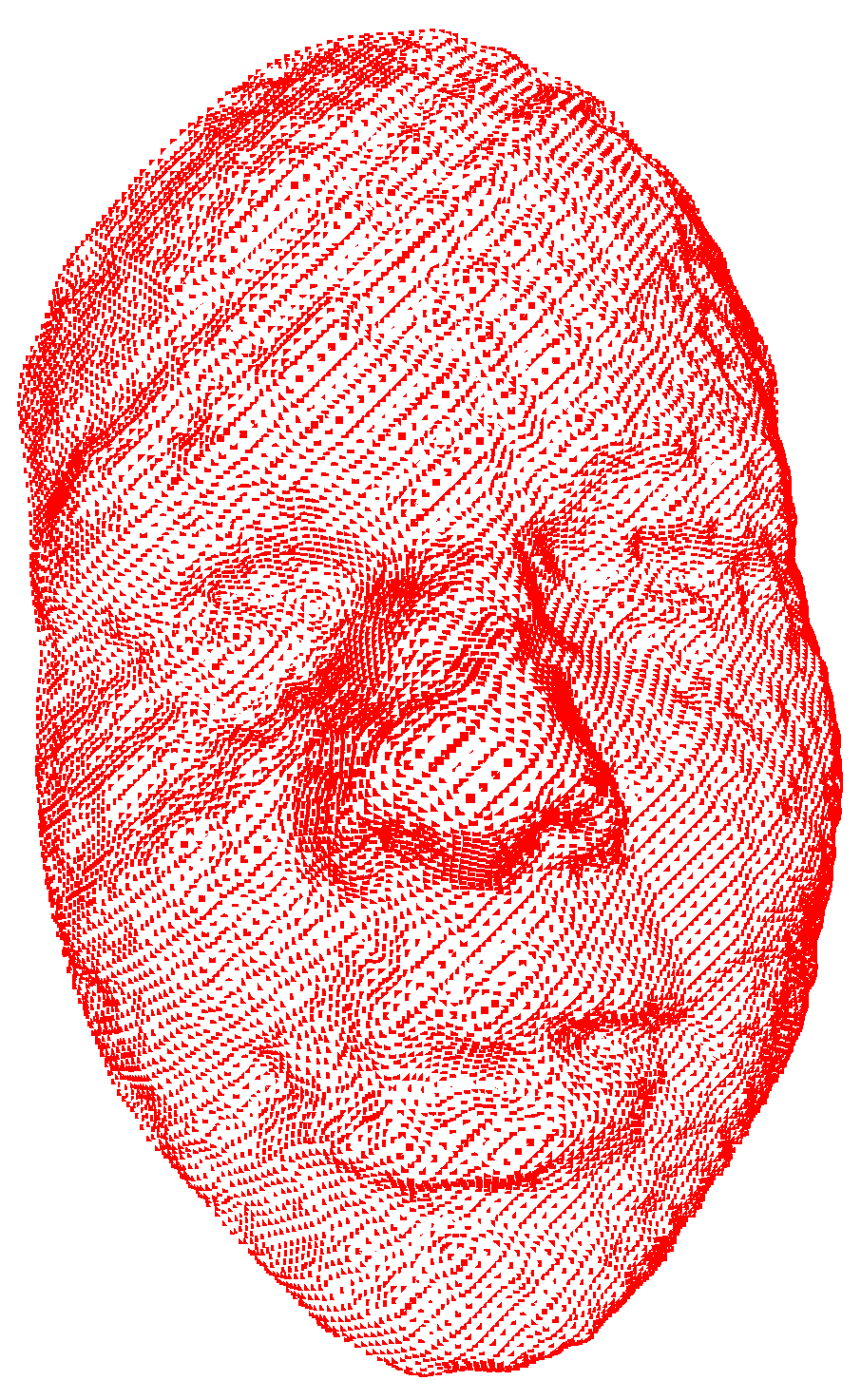}
\includegraphics[height=30mm]{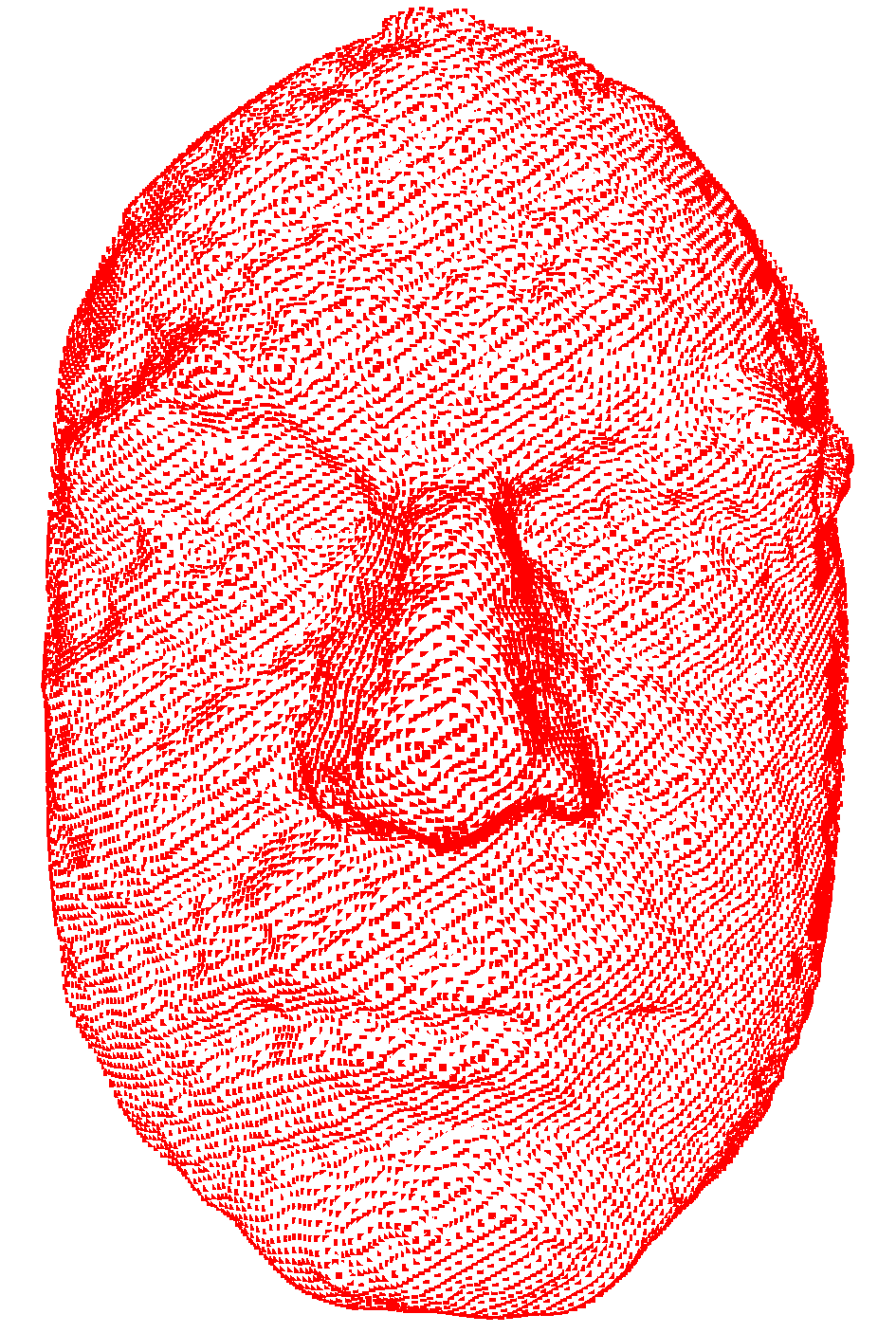}\\
\includegraphics[height=30mm]{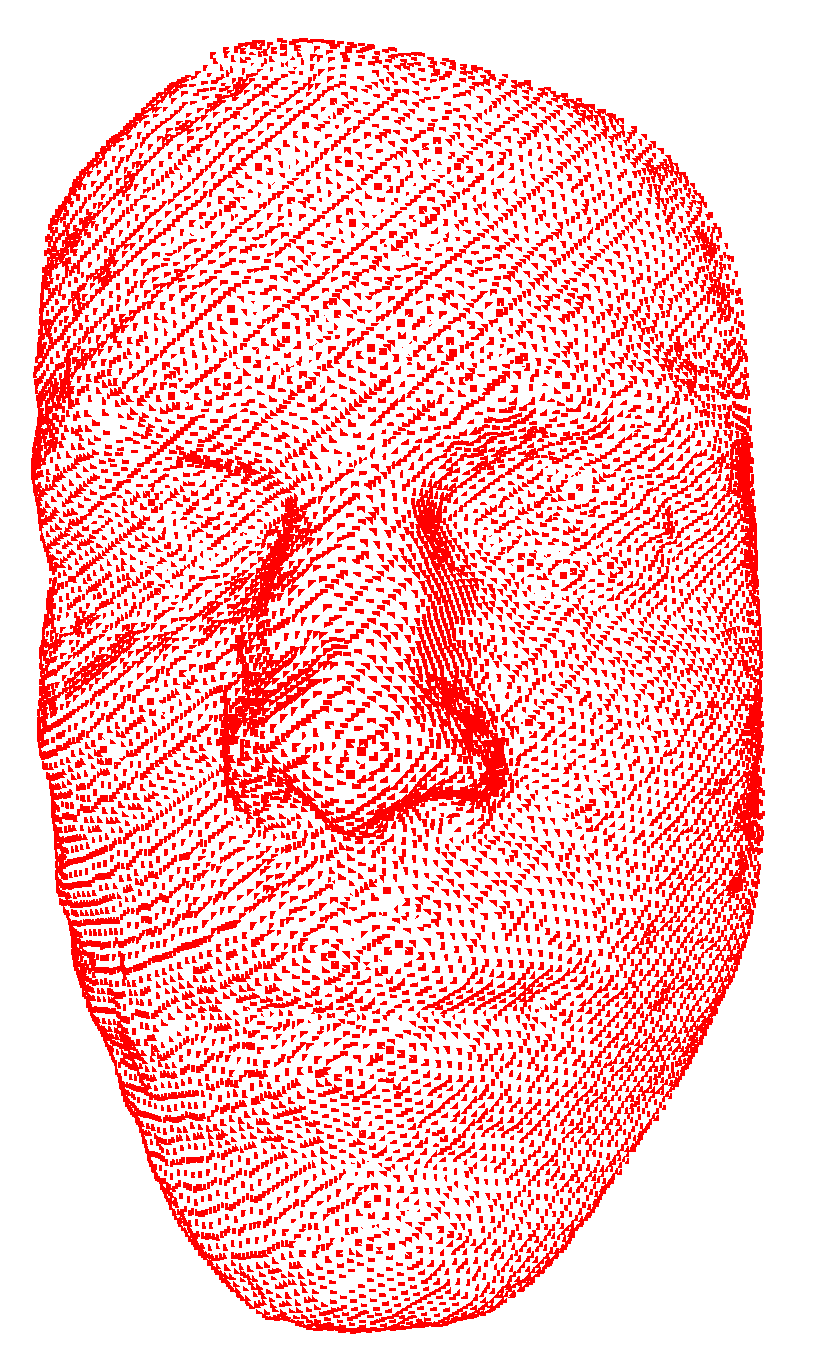}
\includegraphics[height=30mm]{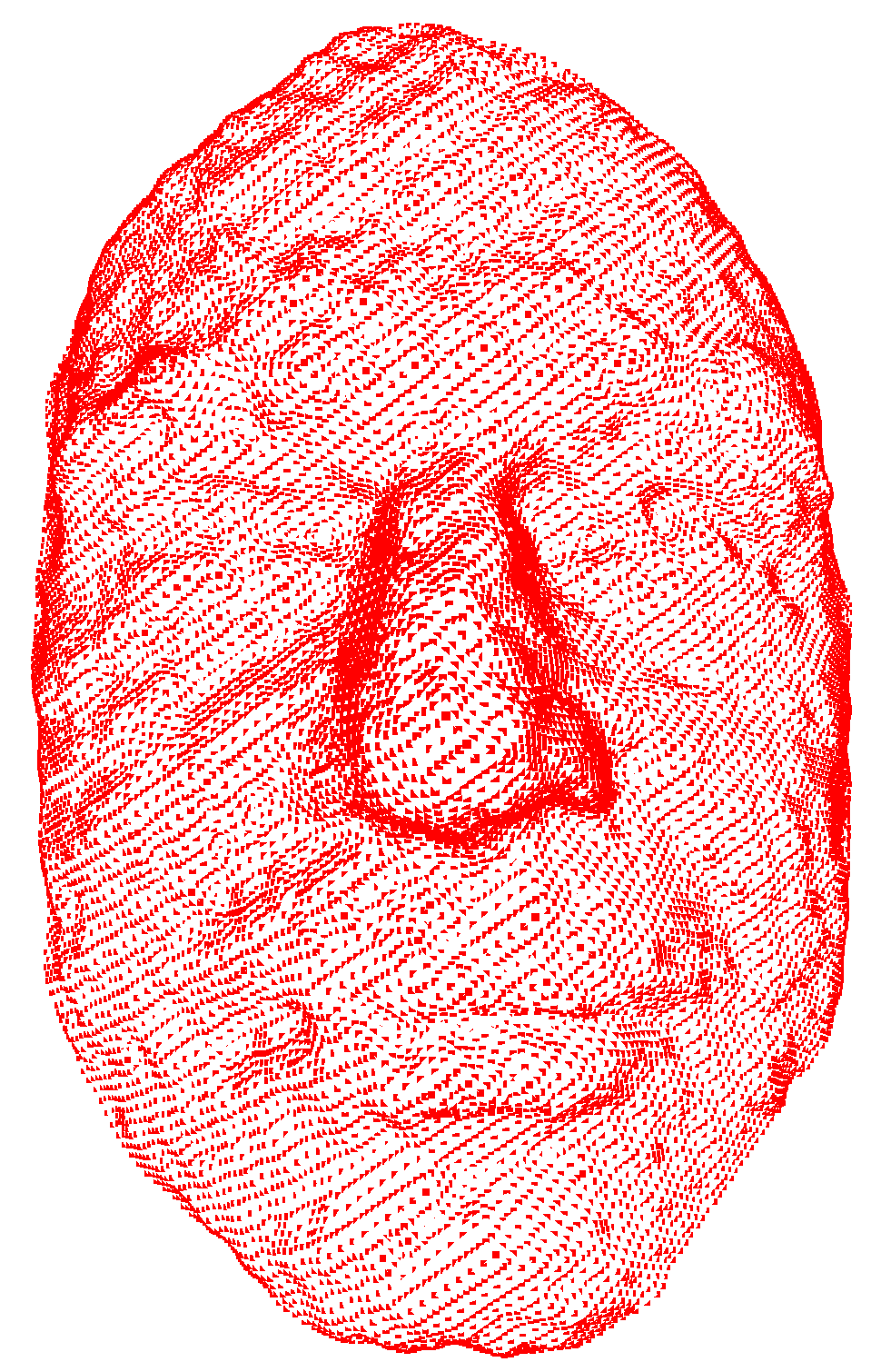}
\includegraphics[height=30mm]{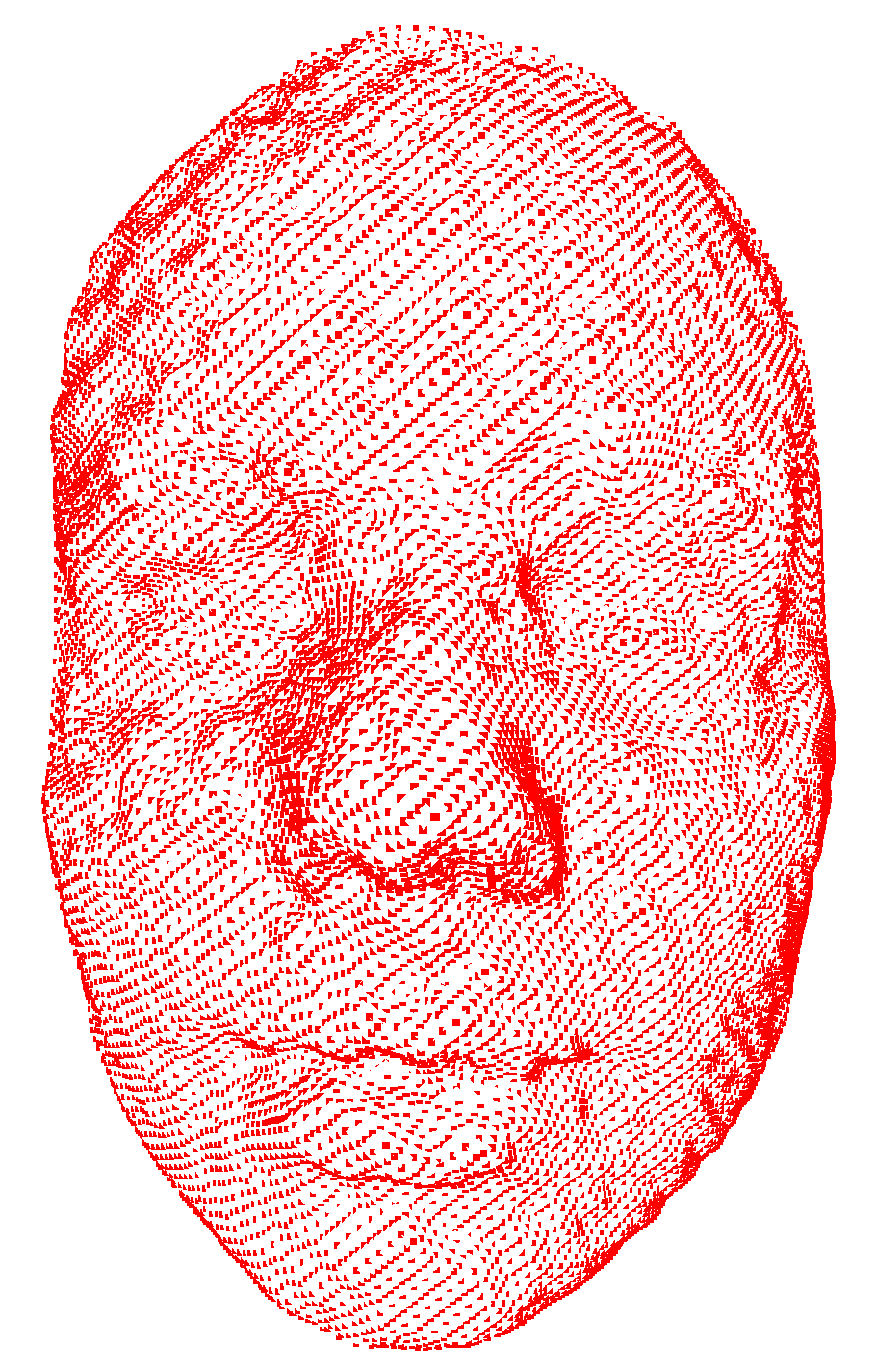}
\includegraphics[height=30mm]{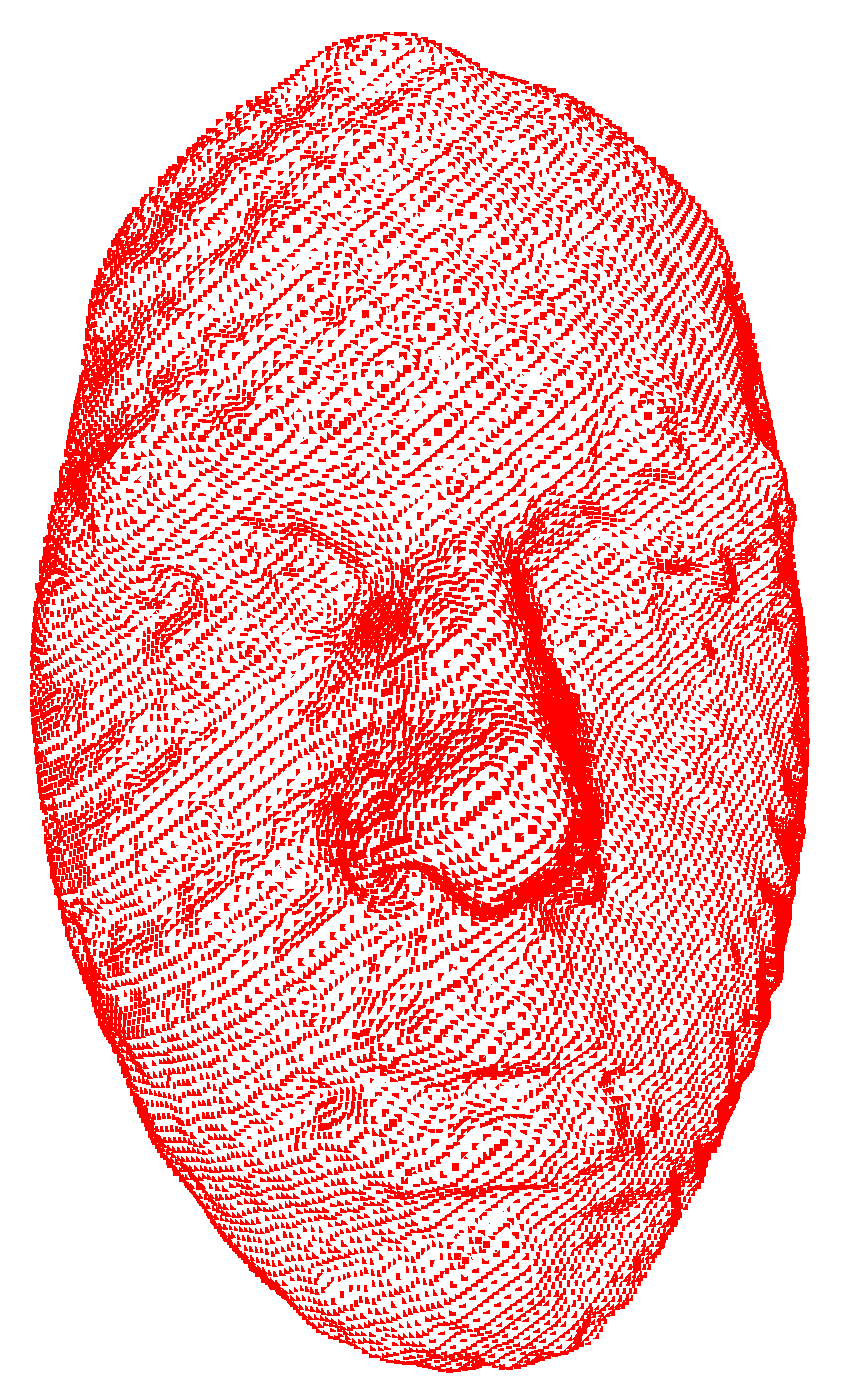}
\includegraphics[height=30mm]{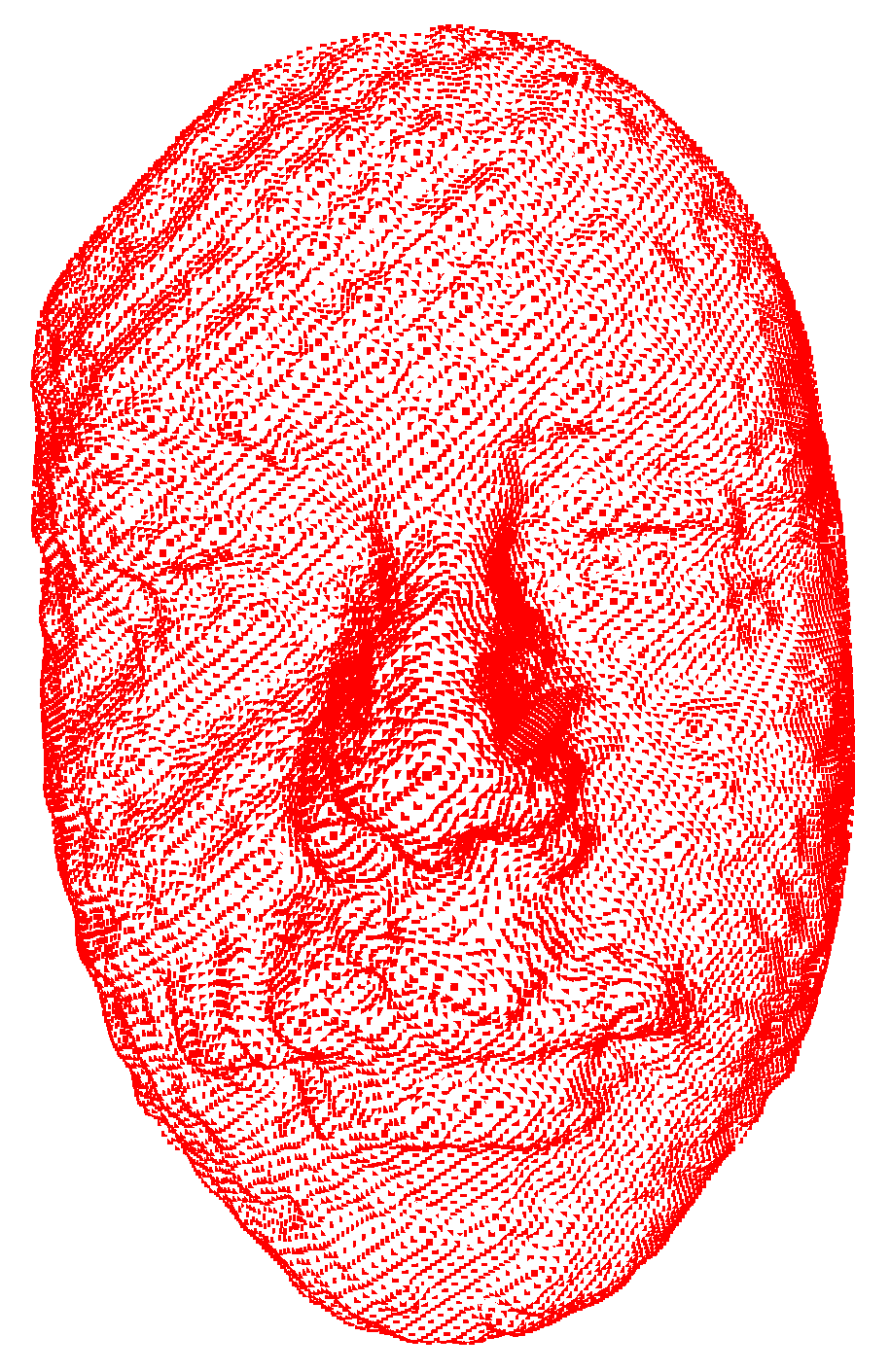}
\includegraphics[height=30mm]{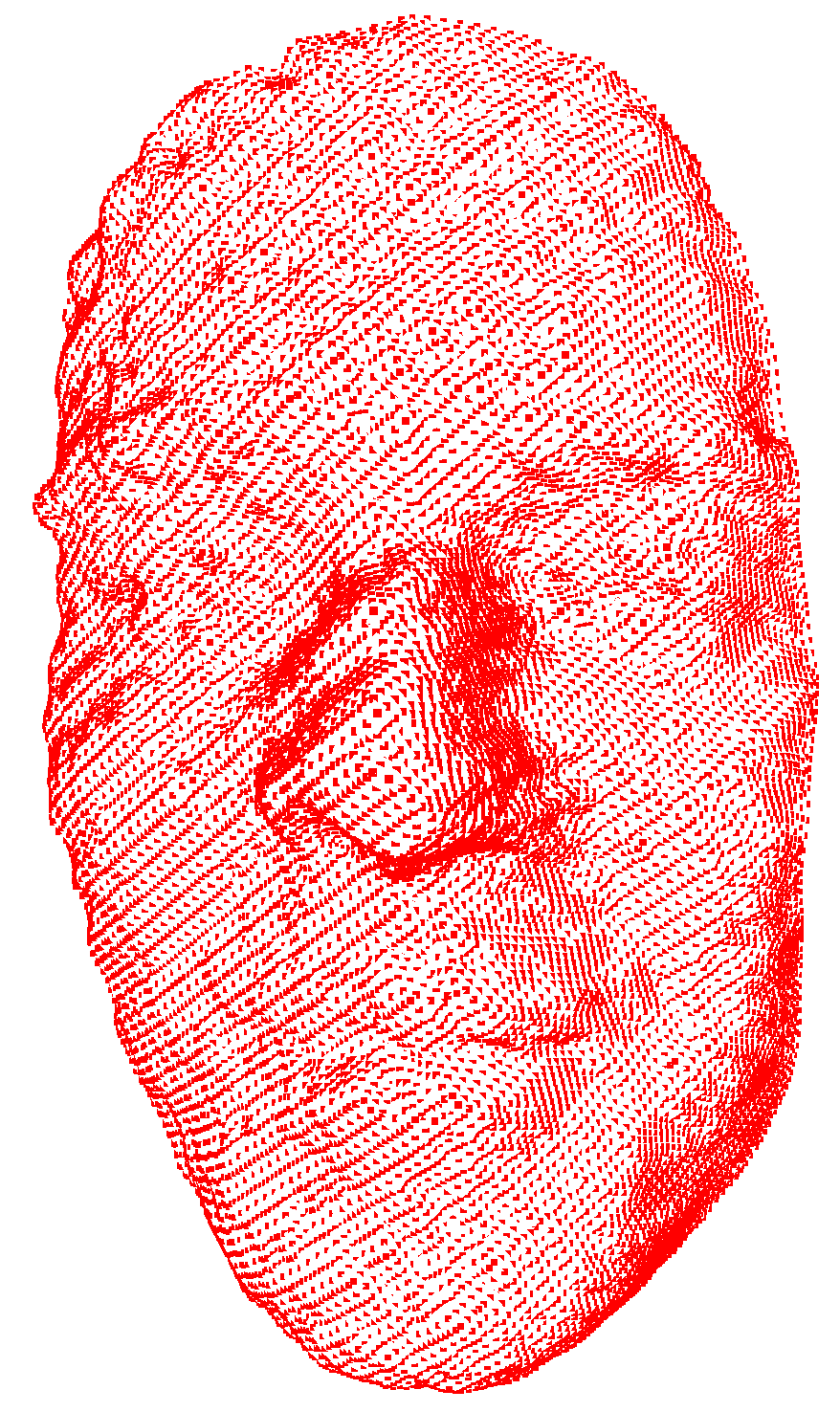}\\
\includegraphics[height=30mm]{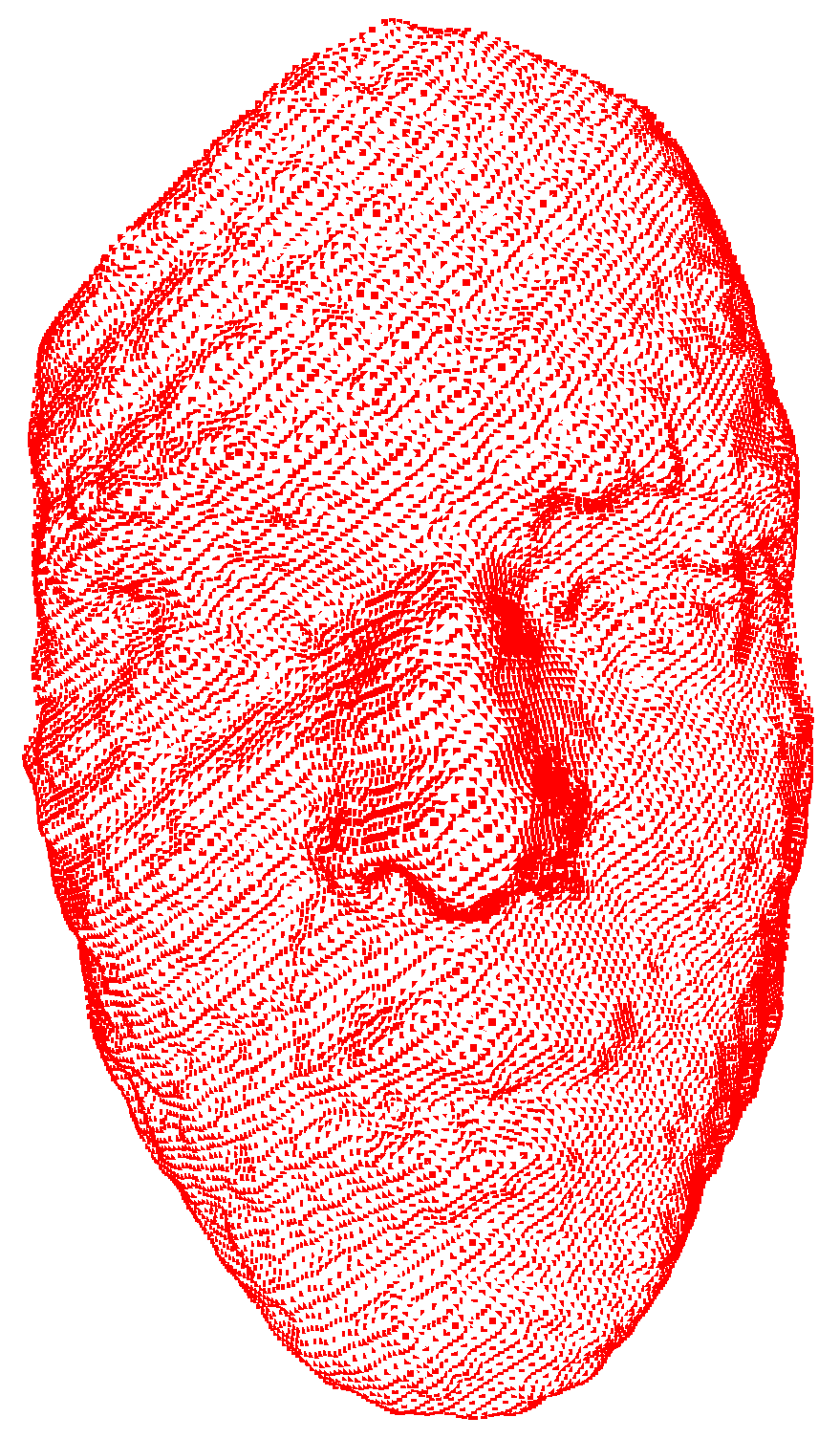}
\includegraphics[height=30mm]{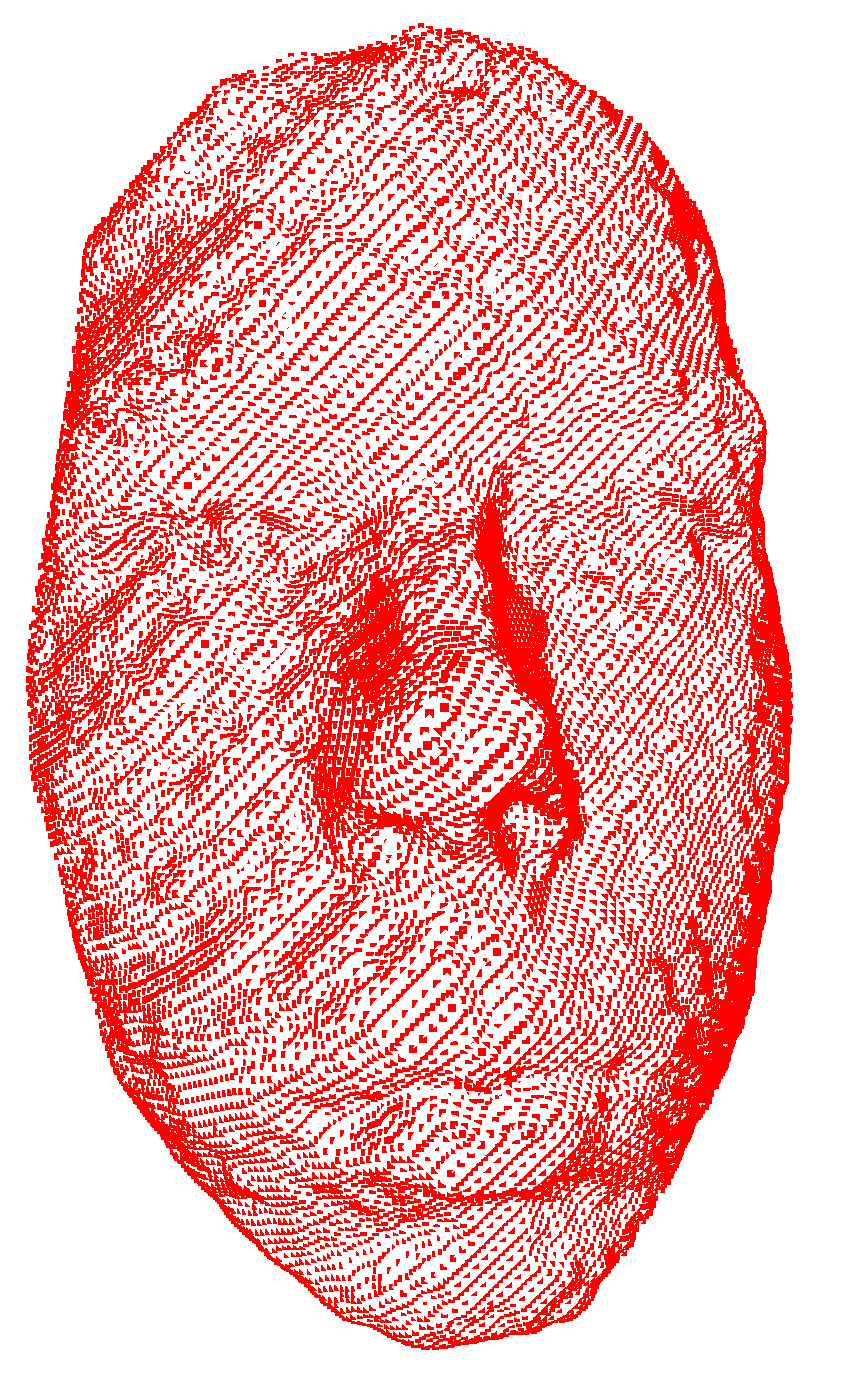}
\includegraphics[height=30mm]{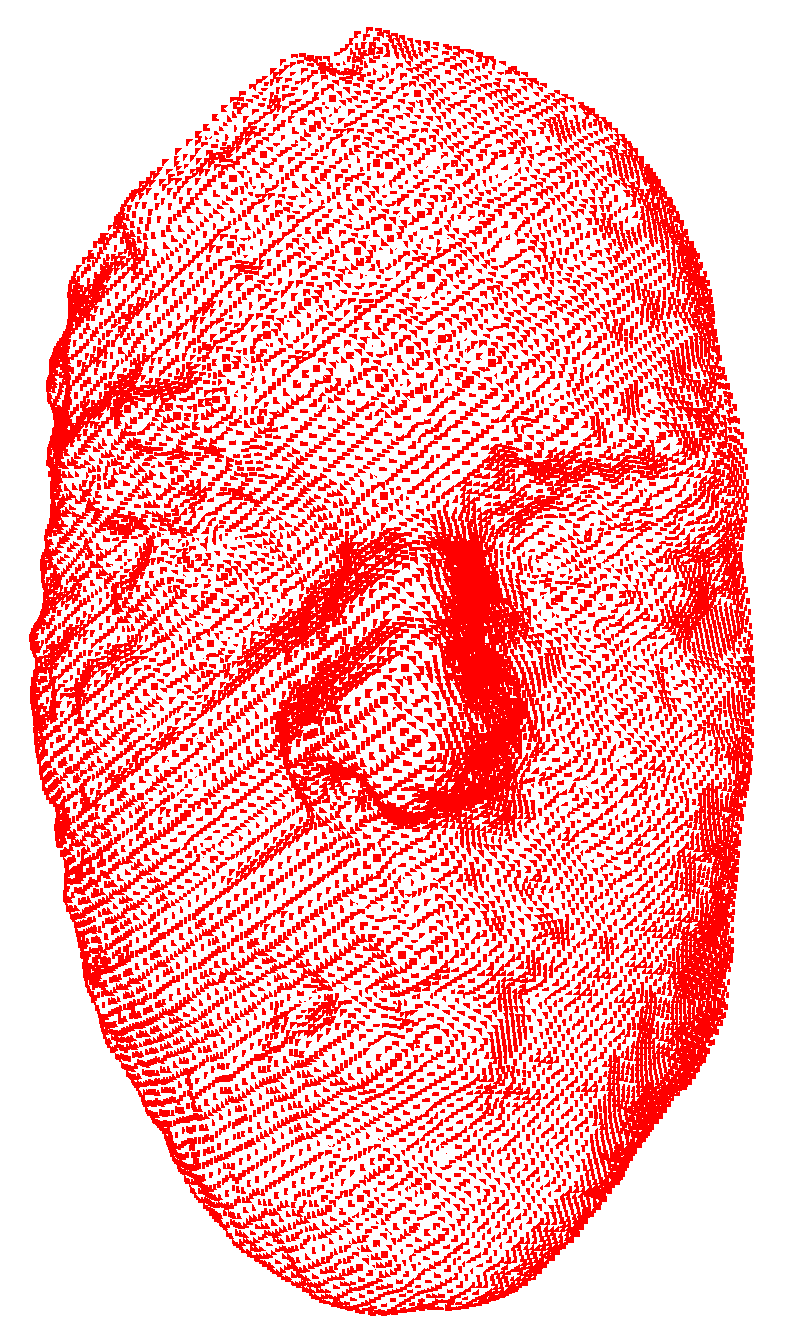}
\includegraphics[height=30mm]{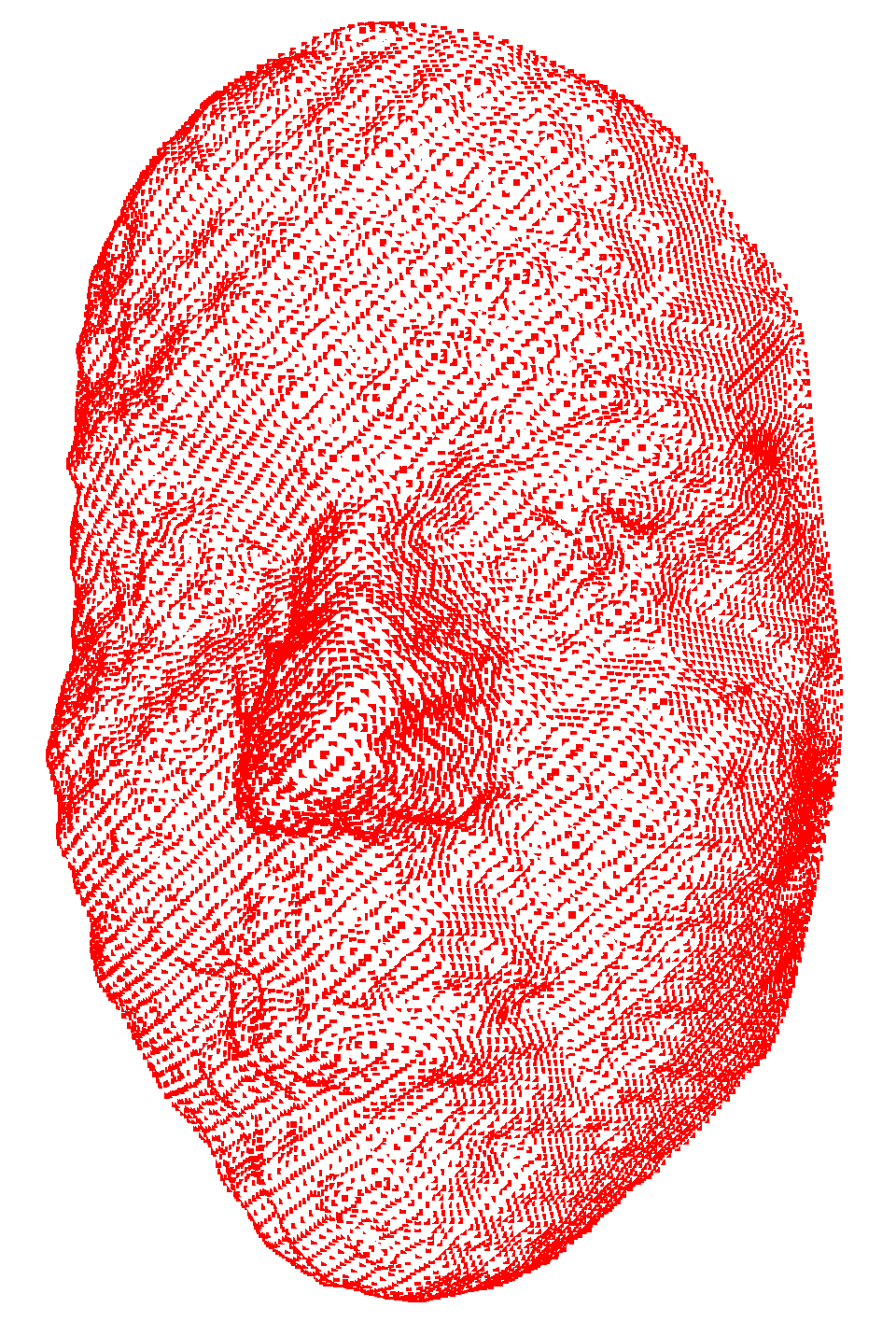}
\includegraphics[height=30mm]{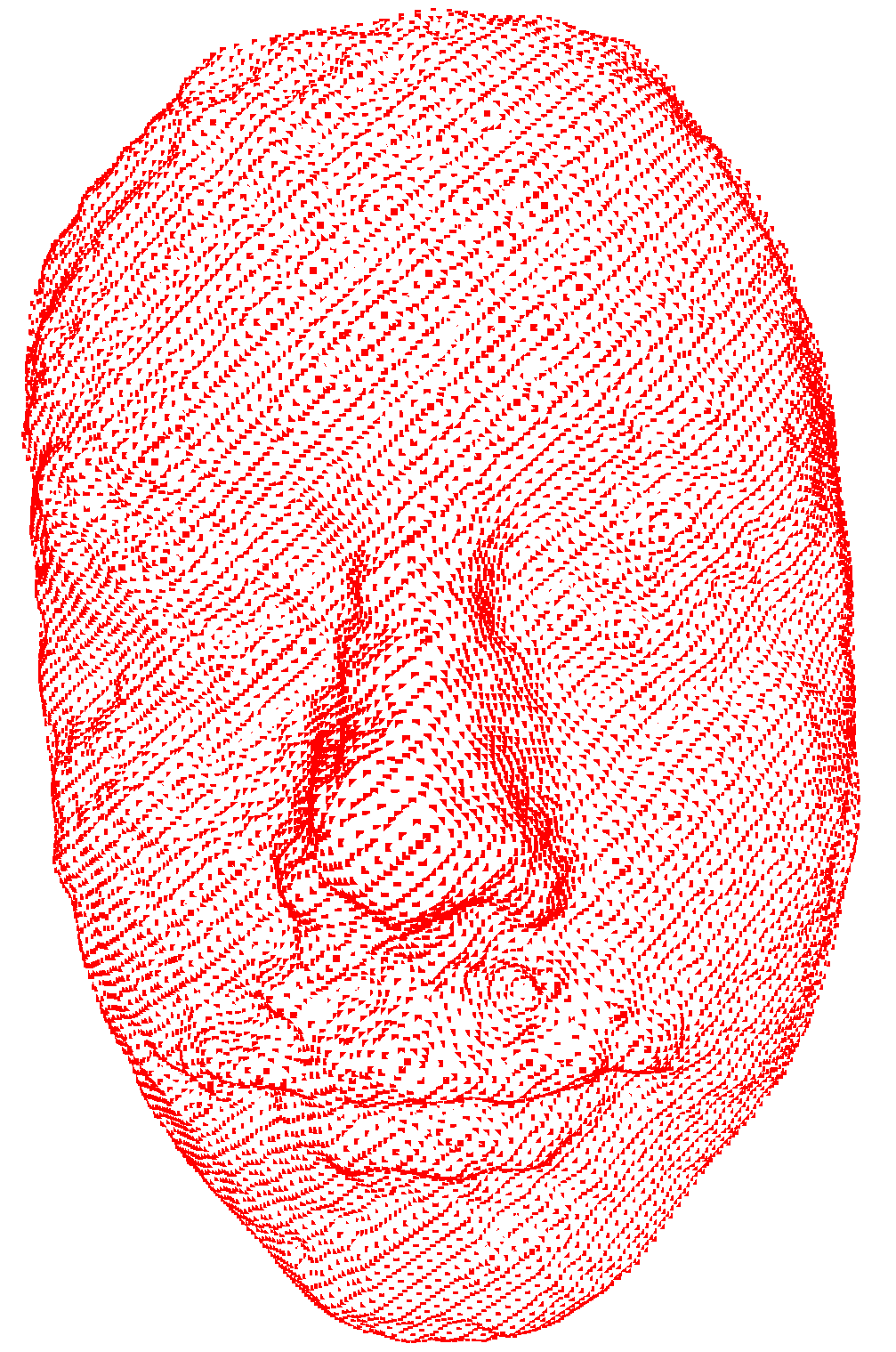}
\includegraphics[height=30mm]{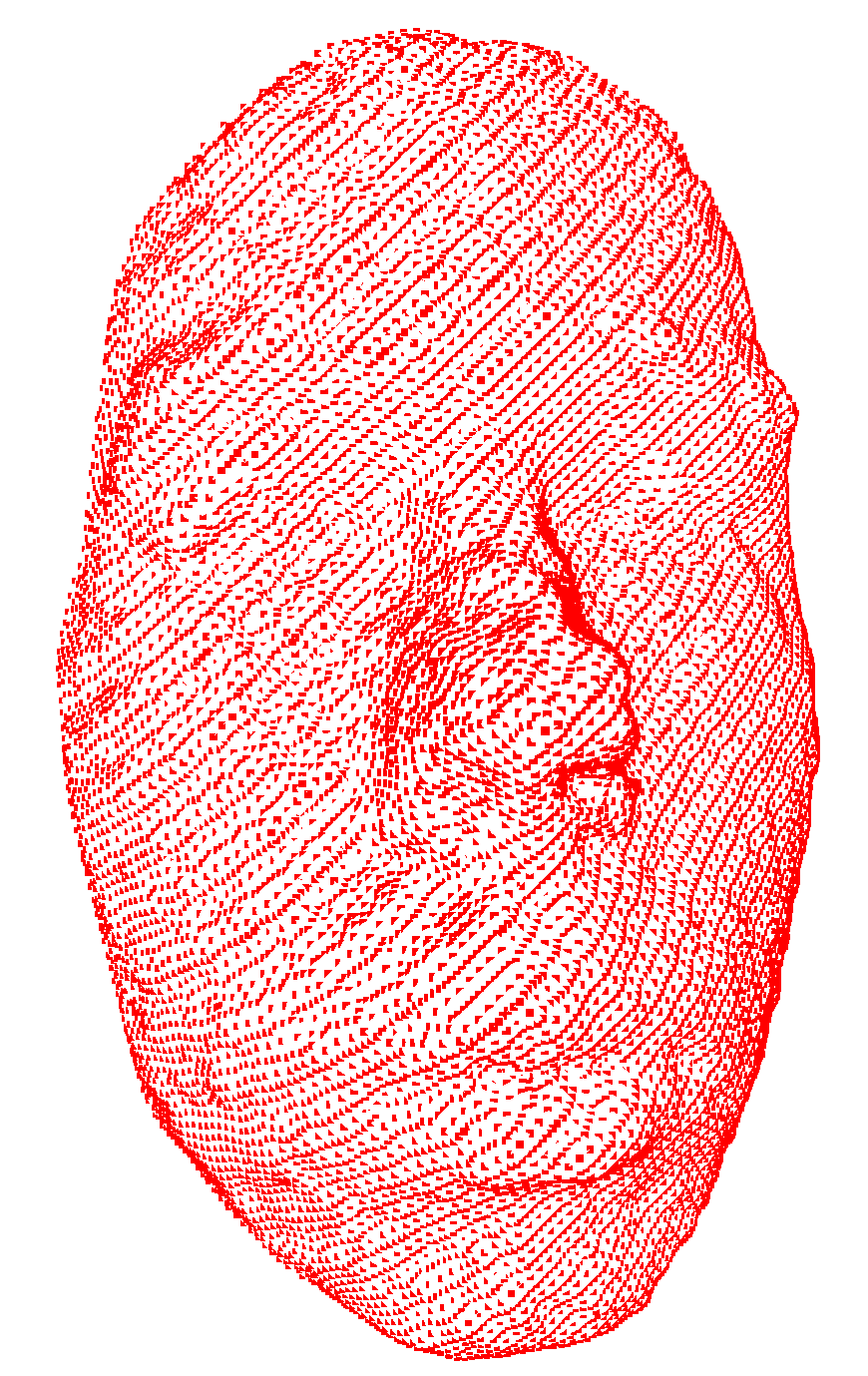}
\caption{A partial set of the facial point clouds used in our second experiment, adapted from the 3D human face database \cite{Beumier00}.}
\label{fig:face_dataset2}
\end{figure}

Using our proposed Teichm\"uller parameterization scheme, Teichm\"uller registration between point clouds with prescribed landmark constraints can be efficiently computed as described in Algorithm \ref{alg:registration}. Figure \ref{fig:face_registration} shows the registration of two human face point clouds by our proposed method. It can be observed that the landmark constraints are satisfied in the registration result. Moreover, the geometric features are optimally preserved under the registration. There is no unnatural distortion in any specific part of the registered point cloud.

The Teichm\"uller metric induced by the landmark matching Teichm\"uller parameterization serves as an effective dissimilarity metric for shape analysis of point clouds. To demonstrate the effectiveness of the Teichm\"uller metric, we consider a facial point cloud classification problem. Suppose we are given a set of facial point clouds. It is conceivable that even with different facial expressions, point clouds sampled from the same human should be similar in shape. We aim to correctly classify the point clouds into several groups, where each of the groups represents one human.

In our experiments, each subject is with multiple facial expressions. Landmark constraints at the most prominent parts of the faces, such as the eyes, nose and mouth of each point cloud are manually labeled to ensure the accuracy of the classification. We compare our Teichm\"uller metric with another dissimilarity metric computed by a feature-endowed point cloud mapping algorithm in \cite{Irfanoglu04}. For a fair comparison, the automatically detected landmarks in \cite{Irfanoglu04} are replaced by our manually labeled landmarks. In the experiments, we first compute our proposed distance matrix using Algorithm \ref{alg:metric}. On the other hand, we compute another distance matrix using \cite{Irfanoglu04}. To assess the classification performance, we apply the multidimensional scaling (MDS) method on the two distance matrices. Two experiments are presented below.

In the first experiment, we are interested in classifying facial point clouds with prominent facial expressions. 16 facial point clouds with 4 specific expressions (neutral, happy, sad, angry) are adapted from \cite{lcln,Bronstein07,Zhang04} or sampled by Kinect. Figure \ref{fig:face_dataset} shows the dataset. 12 landmark constraints are manually labeled on each point cloud for the computation. The MDS results of our distance matrix and the distance matrix in \cite{Irfanoglu04} are shown in Figure \ref{fig:mds}. It can be observed that distinct subjects are effectively clustered using our distance matrix. Even with highly different facial expressions, the facial point clouds sampled from the same subject can be grouped. On the contrary, the result based on the distance matrix in \cite{Irfanoglu04} cannot identify distinct subjects. This implies that the Teichm\"uller metric leads to a better classification result when compared with the distance matrix in \cite{Irfanoglu04}. For a more detailed comparison, the leave-one-out cross-validation (LOOCV) is applied to evaluate the classification accuracies based on our distance matrix and the distance matrix computed in \cite{Irfanoglu04}. Our distance matrix results in a classification accuracy of 94\%, while the distance matrix in \cite{Irfanoglu04} results in a classification accuracy of 50\%. The comparison demonstrates the effectiveness of our TEMPO method in shape analysis of point clouds.

In the second experiment, we consider a larger dataset. More specifically, we adapt the 3D human face database \cite{Beumier00}, which is the database used in \cite{Irfanoglu04}. Each subject in the database is with 3 random facial expressions. Figure \ref{fig:face_dataset2} shows a partial set of the facial point clouds in \cite{Beumier00}. 10 landmark constraints on each point cloud are labeled for the computation. Note that in our experiment, we do not perform any triangulation or smoothing procedure on the raw point cloud data. Our distance matrix results in a LOOCV accuracy of 79\%, while the distance matrix in \cite{Irfanoglu04} results in a LOOCV accuracy of 50\%. The results again reflect the effectiveness of our proposed dissimilarity metric.

\section{Conclusion and future work} \label{conclusion}
In this paper, we have developed the notion of PCT-maps, a discrete analogue of the Teichm\"uller mappings on point clouds. The consistency between the PCT-maps and the continuous T-maps has been rigorously established. Based on the notion of PCT-maps, we have proposed a novel method called TEMPO for computing landmark-matching Teichm\"uller mappings of point clouds with disk topology. Firstly, we have introduced a hybrid quasi-conformal reconstruction scheme for computing quasi-conformal mappings on point clouds. Secondly, an improved method has been proposed for approximating differential operators on point clouds with disk topology. The new approximation method produces more accurate results. Thirdly, we have presented an algorithm for computing the Teichm\"uller parameterizations of disk-type point clouds with rigorous mathematical guarantees. Unlike the conventional approaches, our method allows prescribed landmark constraints on the point clouds and guarantees uniform conformality distortions. Fourthly, the induced Teichm\"uller metric enables us to accurately classify different feature-endowed point clouds. Experimental results have demonstrated the effectiveness of our proposed TEMPO algorithm. In the future, we aim to extend our Teichm\"uller parameterization and registration methods for point clouds with arbitrary topologies.

\end{document}